\documentclass[numbib]{imaiai-DK_edit}

\usepackage{amsmath,amssymb,amsthm}
\usepackage{graphicx}
\usepackage{subfigure}
\usepackage{url}
\usepackage{color}
\usepackage[mathscr]{eucal}
\usepackage{bm}
\usepackage{mdframed}
\usepackage{rotating}
\usepackage{algorithm}
\usepackage{soul}
\usepackage{xcolor}

\DeclareMathOperator{\N}{\mathcal{N}}
\DeclareMathOperator{\R}{\mathbb{R}}
\DeclareMathOperator{\E}{\mathbb{E}}
\DeclareMathOperator{\prob}{\mathbb{P}}

\newcommand\rank[1]{\text{rank}\left(#1\right)}
\newcommand\vct[1]{\text{vec}\left(#1\right)}
\newcommand\vect[1]{\overrightarrow{#1}}
\newcommand\wE[1]{\widehat{\E}[#1]}
\newcommand\TP{{T _{x_0}\negthinspace\mathcal{M}}}
\newcommand\NS{{N _{x_0}\negthinspace\mathcal{M}}}

\begin{document}

\title{Non-Asymptotic Analysis of Tangent Space Perturbation}

\shorttitle{Tangent Space Perturbation} 
\shortauthorlist{D. N. Kaslovsky and F. G. Meyer} 

\author{{
    \sc Daniel N. Kaslovsky}$^*$ {\sc and} {\sc Fran\c{c}ois G. Meyer}\\[2pt]
  Department of Applied Mathematics, University of Colorado, Boulder, Boulder, CO, USA\\
  $^*${\email{kaslovsky@colorado.edu \\ fmeyer@colorado.edu}}}

\maketitle

\begin{abstract}
  {Constructing an efficient parameterization of a
    large, noisy data set of points lying close to a smooth manifold
    in high dimension remains a fundamental problem.  One approach
    consists in recovering a local parameterization using the local
    tangent plane.  Principal component analysis (PCA) is often the
    tool of choice, as it returns an optimal basis in the case of
    noise-free samples from a linear subspace.  To process noisy data samples from a nonlinear manifold,
    PCA must be applied locally, at a scale small enough such that the
    manifold is approximately linear, but at a scale large enough such
    that structure may be discerned from noise.  Using eigenspace
    perturbation theory and non-asymptotic random matrix theory, we study the stability of the subspace
    estimated by PCA as a function of scale, and bound (with high
    probability) the angle it forms with the true tangent space.  By
    adaptively selecting the scale that minimizes this bound, our
    analysis reveals an appropriate scale for local tangent plane
    recovery. We also introduce a geometric uncertainty principle
    quantifying the limits of noise-curvature perturbation for stable recovery.
    \textcolor{black}{  With the purpose of providing perturbation bounds that can be used
  in practice, we propose plug-in estimates that make it possible to
  directly apply the theoretical results to real data sets.}}
  {manifold-valued data, tangent space, principal component analysis, subspace perturbation, local linear models, curvature, noise.}
  \\
  2000 Math Subject Classification: 62H25, 15A42, 60B20
\end{abstract}

\section{Introduction and Overview of the Main Results} \label{sec:intro}

\subsection{Local Tangent Space Recovery: Motivation and Goals}

Large data sets of points in high-dimension often lie close to a
smooth low-dimensional manifold.  A fundamental problem in processing
such data sets is the construction of an efficient parameterization
that allows for the data to be well represented in fewer dimensions.
Such a parameterization may be realized by exploiting the inherent
manifold structure of the data.  However, discovering the geometry of
an underlying manifold from only noisy samples remains an open topic
of research.

The case of data sampled from a linear subspace is well studied (see
\cite{Johnstone,Jung-Marron,Nadler}, for example).  The optimal
parameterization is given by principal component analysis (PCA), as
the singular value decomposition (SVD) produces the best low-rank
approximation for such data.  However, most interesting
manifold-valued data organize on or near a nonlinear manifold.  PCA,
by projecting data points onto the linear subspace of best fit, is not
optimal in this case as curvature may only be accommodated by
choosing a subspace of dimension higher than that of the manifold.
Algorithms designed to process nonlinear data sets typically proceed
in one of two directions.  One approach is to consider the data
globally and produce a nonlinear embedding.  Alternatively, the data
may be considered in a piecewise-linear fashion and linear methods
such as PCA may be applied locally.  The latter is the subject of this
work.

Local linear parameterization of manifold-valued data requires the estimation
of the local tangent space (``tangent plane'') from a neighborhood of points.
However, sample points are often corrupted by high-dimensional noise and any local neighborhood deviates from the linear assumptions of PCA due to the curvature of the manifold.
Therefore, the subspace recovered by local PCA is a perturbed version of the true tangent space.
The goal of the present work is to characterize the stability and accuracy of local
tangent space estimation using eigenspace perturbation theory.
\begin{figure}
  \centering
  \subfigure[small neighborhoods]{\includegraphics[scale=.5,clip,trim=.75in 0in 1in 0in]{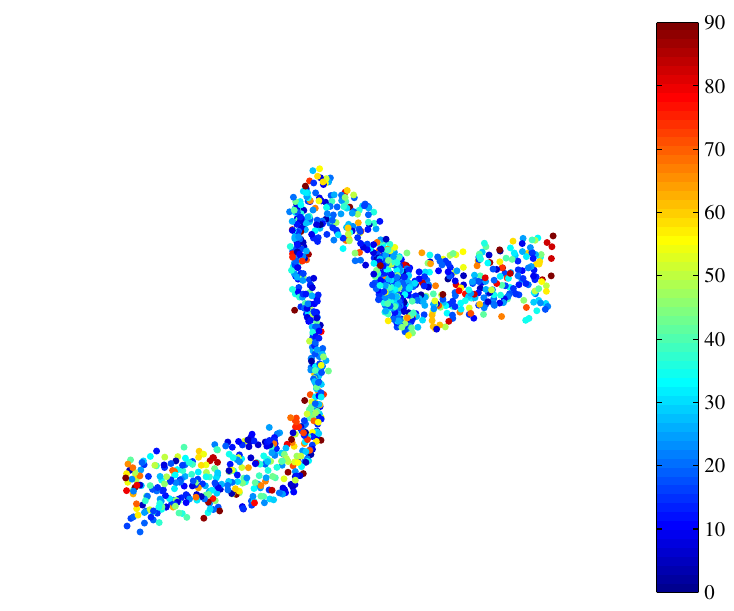}}
  \hskip1ex
  \subfigure[large neighborhoods]{\includegraphics[scale=.5,clip,trim=.75in 0in 1in 0in]{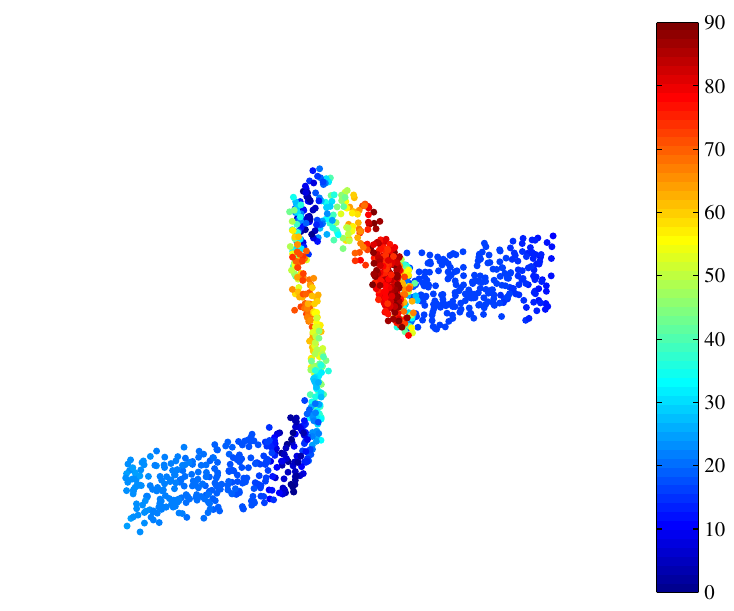}}
  \hskip1ex
  \subfigure[adaptive neighborhoods]{\includegraphics[scale=.5,clip,trim=.75in 0in 0in 0in]{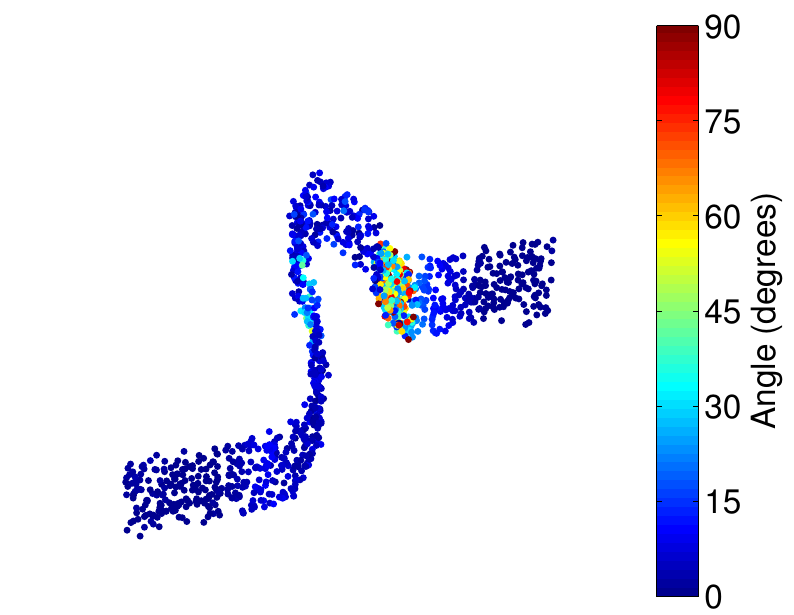}}
  \caption{Angle between estimated and true tangent planes at each point of a noisy 2-dimensional data set embedded in $\R^{3}$.  The estimated tangent planes are (a) randomly oriented when computed from small neighborhoods within the noise; (b) misaligned when computed from large neighborhoods exhibiting curvature; and (c) properly oriented when computed from adaptively defined neighborhoods given by the analysis in this work.}
  \label{fig:example}
\end{figure}

The proper neighborhood for local tangent space recovery must be a function
of intrinsic (manifold) dimension, curvature, and noise level; these properties often vary
as different regions of the manifold are explored.    
However, local PCA approaches proposed in the data analysis
and manifold learning literature often define locality via an \textit{a priori} fixed number of neighbors \textcolor{black}{or
  as the output of clustering and partitioning algorithms} (e.g., \cite{LLE,Zhang-Zha,Kam-Leen,Yang}).
Other methods \cite{Brand,Ohtake,Lin-Tong} adaptively estimate local neighborhood
size but are not tuned to the perturbation of the recovered subspace.
Our approach studies this perturbation as the size of the neighborhood varies
to guide the definition of locality.  On the one hand, a neighborhood
must be small enough so that it is approximately linear and avoids
curvature.  On the other hand, a neighborhood must be be large enough
to overcome the effects of noise.
A simple yet instructive example of these competing criteria is shown in Figure
\ref{fig:example}.  The tangent plane at every point of a noisy 2-dimensional data set embedded in $\R^3$ is computed via local PCA.  Each point is color coded according to the angle formed with the true tangent plane.  Three different neighborhood definitions are used: a small, fixed radius (Figure \ref{fig:example}a); a large, fixed radius (Figure \ref{fig:example}b); and radii defined adaptively according \textcolor{black}{to} the analysis presented in this work (Figure \ref{fig:example}c).
As small neighborhoods may be within the noise level and large neighborhoods exhibit curvature, the figure shows that neither allows for accurate tangent plane recovery.  In fact, because the curvature varies across the data, only the adaptively defined neighborhoods avoid random orientation due to noise (as seen in Figure \ref{fig:example}a) and misalignment due to curvature (as seen in Figure \ref{fig:example}b).  Figure \ref{fig:example}c shows accurate and stable recovery at almost every data point, with misalignment only in the small region of very high curvature that will be troublesome for any method.  The present work quantifies this observed behavior in the high-dimensional setting.

We present a non-asymptotic, eigenspace perturbation analysis to 
bound, with high probability, the angle between the recovered linear
subspace and the true tangent space as the size of the local neighborhood varies.
The analysis accurately tracks the subspace recovery error as a function of neighborhood size, noise, and curvature.
Thus, we are able to adaptively select the neighborhood that
minimizes this bound, yielding the best estimate to the local tangent space from a large
but finite number of noisy manifold samples.
Further, the behavior of this bound demonstrates the non-trivial
existence of such an optimal scale. 
We also introduce a geometric uncertainty principle quantifying
the limits of noise-curvature perturbation for tangent space recovery.

\textcolor{black}{An important technical matter that one needs to address when
  analyzing points that are sampled from a manifold blurred with
  Gaussian noise concerns the probability distribution of the noisy
  samples. Indeed, after perturbation with Gaussian noise, the
  probability density function of the noisy points can be expressed as
  the convolution of the probability density function of the clean
  points on the manifold with a Gaussian kernel.
  Geometrically, the points are diffused into a tube around the manifold,
  and the corresponding density of the points is thinned.
  This concept has been studied in
  great detail in \cite{Maggioni-MIT,LittleThesisDuke} as well as in
  \cite{Niyogi11,Genovese12b}.  
  The practical implication of these
  studies is that concentration of measure helps us to guarantee that
  the volume of noisy points in a ball centered on the clean manifold
  can be estimated from the volume of the corresponding ball of clean
  points, provided one applies a correction of the radius.
  We take advantage of these ideas in our analysis
  by replacing the ball of noisy points in the tube with a ball of similar volume extracted from the clean manifold,  perturbed by Gaussian noise.
  We introduce the resulting,
  necessary modification to the radii in Section \ref{sec:numerical2}. A related
  issue concerns the determination of the point $x_0$ about which we
  estimate the tangent plane. From a practical perspective, one can
  only observe noisy samples, and it is therefore reasonable
  that the perturbation bound should account for the fact that the
  analysis cannot be centered around the clean manifold.  The
  expected effect of this additional source of uncertainty has been
  explored in detail in \cite{Maggioni-MIT,LittleThesisDuke}. In this
  paper, we propose a different approach. We
  devise a plug-in method to estimate a clean point $x_0$ on the
  manifold using the observed noisy data. As a result, the theoretical
  analysis can proceed assuming that $x_0$ is given by an oracle.
  Our experiments confirm that the 
  local origin $x_0$ on the manifold can be estimated from
  the noisy neighborhood of observed points and that the
  perturbation error can be accurately tracked
  in practice.
  In addition, we expect this novel denoising
  algorithm to provide a universal tool for the analysis of noisy
  point cloud data.}

Our analysis is related to the very recent work of Tyagi, {\it et al.} \cite{Tyagi}, in which neighborhood size and sampling density conditions are given to ensure a small angle between the PCA subspace and the true tangent space of a noise-free manifold.  Results are extended to  arbitrary smooth embeddings of the manifold model, which we do not consider.
In contrast, we envision the scenario in which no control is given over the sampling and explore the case of data sampled according to a fixed density and corrupted by high-dimensional noise.  Crucial to our results is a careful analysis of the interplay between the perturbation due to noise and the perturbation due to curvature.  Nonetheless, our results can be shown to recover those of  \cite{Tyagi} in the noise-free setting.
Our approach is also similar to the analysis presented by Nadler in
\cite{Nadler}, who studies the finite-sample properties of the PCA
spectrum.  Through matrix perturbation theory, \cite{Nadler} examines the
angle between the leading finite-sample-PCA eigenvector and that of
the leading population-PCA eigenvector.  As a linear model is assumed,
perturbation results from noise only.  Despite this difference,
the two analyses utilize similar techniques to bound the effects of
perturbation on the PCA subspace and our results recover those of
\cite{Nadler} in the curvature-free setting.

\textcolor{black}{Application of multiscale PCA for geometric analysis of data sets has also been studied in \cite{Fukunaga-Olsen,Froehling,Wang-Marron,Broomhead}.
In parallel to our work \cite{Kaslovsky11a,Kaslovsky11b,Meyer12b,Kaslovsky12a}, Maggioni and coauthors have developed results \cite{LittleThesisDuke,Maggioni-long,Maggioni-MIT} addressing similar questions as those examined in this paper.  These results are discussed above as well as in more detail in Section \ref{sec:numerical2} and Section \ref{sec:discussion}.}
Other recent related works include that of Singer and Wu
\cite{Singer}, who use local PCA to build a tangent plane basis and
give an analysis for the neighborhood size to be used in the absence
of noise.  Using the hybrid linear model, Zhang, {\it et al.}
\cite{Lerman} assume data are samples from a collection of ``flats''
(affine subspaces) and choose an optimal neighborhood size from which
to recover each flat by studying the least squares approximation error
in the form of Jones' $\beta$-number (see \cite{Jones} and also
\cite{Jones-new} in which this idea is used for curve denoising).
Finally, an analysis of noise and curvature for normal estimation of smooth curves
and surfaces in $\mathbb{R}^2$ and $\mathbb{R}^3$ is presented by
Mitra, {\it et al.} \cite{Guibas-normal-estimation} with application
to computer graphics.

\subsection{Overview of the Results} \label{sec:overview}

We consider the problem of recovering
the best approximation to a local tangent space
of a nonlinear $d$-dimensional Riemannian manifold $\mathcal{M}$ from
noisy samples presented in dimension $D > d$.  Working about a
reference point $x_0$, an approximation to the tangent space of
$\mathcal{M}$ at $x_0$ is given by the span of the top $d$ eigenvectors of the centered data covariance matrix (where ``top'' refers to the $d$
eigenvectors or singular vectors associated with the $d$ largest eigenvalues or singular values).
The question becomes: how many neighbors of $x_0$ should be used (or
in how large of a radius about $x_0$ should we work) to recover the
best approximation?  We will often use the term ``scale'' to refer
to this neighborhood size or radius.

To answer this question, we consider the perturbation of the eigenvectors spanning the estimated tangent space in the context of the ``noise-curvature trade-off.''
To balance the effects of noise and curvature (as observed in the example of the previous subsection, Figure \ref{fig:example}), we seek a 
scale large enough to be above the noise level but
still small enough to avoid curvature.  This scale reveals a linear
structure that is sufficiently decoupled from both the noise and the
curvature to be well approximated by a tangent plane.
At this scale, the recovered eigenvectors span a subspace corresponding very closely to the true tangent space of the manifold at $x_0$.
We note that the concept of noise-curvature trade-off has been a subject of
interest for decades in dynamical systems theory \cite{Froehling}.

The main result of this work is a bound on the angle between the computed and true tangent spaces.
Define $P$ to be the orthogonal projector onto the true tangent space
and let $\widehat{P}$ be the orthogonal projector constructed from
the $d$-dimensional eigenspace of the neighborhood covariance matrix.
Then the distance $\|P-\widehat{P}\|_F^2$ corresponds to the 
sum of the squared sines of the principal angles
between the computed and true tangent spaces and we use eigenspace perturbation theory to bound this norm.  Momentarily neglecting probability-dependent constants to ease the presentation, the first-order approximation of this bound has the following form:\\

\noindent \textbf{Informal Main Result.}
\begin{equation} \label{eq:practical1}
  \textcolor{black}{
    \|P-\widehat{P}\|_F ~\leq~
    \frac{\frac{2\sqrt{2}}{\sqrt{N}}\left[K^{(+)}r^3 + \sigma\sqrt{d(D-d)}\left(\sigma + \frac{r}{\sqrt{d+2}}+ \frac{\mathcal{K}^{1/2}r^2}{(d+2)\sqrt{2(d+4)}}\right)\right]}
    {\frac{r^2}{d+2} - \frac{\mathcal{K}r^4}{2(d+2)^2(d+4)}-\sigma^2\left(\sqrt{d} + \sqrt{D-d}\right)},
  }
\end{equation}
where $r$ is the radius (measured in the tangent plane) of the neighborhood containing $N$ points,  $\sigma$ is the noise level, \textcolor{black}{and $K^{(+)}$ and $\mathcal{K}$ are functions of curvature}.\\[2ex]
\noindent \textcolor{black}{To aid the interpretation, we note that
  $K^{(+)}$ corresponds to the Frobenius norm of the matrix of
  principal curvatures and $\mathcal{K}$ has size $2d(D-d)\kappa^2$ in
  the case where all principal curvatures are equal to $\kappa$.}  The
quantities $N$, $r$, $\sigma$, \textcolor{black}{$K^{(+)}$, and
  $\mathcal{K}$}, as well as the sampling assumptions are more
formally defined in Sections \ref{sec:prelim} and
\ref{sec:mainresult}, and the formal result is presented in Section
\ref{sec:mainresult}.

The denominator of this bound, denoted here by $\delta_{\text{informal}}$,
\begin{equation} \label{eq:delta_informal}
  \textcolor{black}{
    \delta_{\text{informal}} = \frac{r^2}{d+2} - \frac{\mathcal{K}r^4}{2(d+2)^2(d+4)}-\sigma^2\left(\sqrt{d} + \sqrt{D-d}\right)
  }
\end{equation}
quantifies the separation between the spectrum of the linear subspace
($\approx r^2$) and the perturbation due to curvature
\textcolor{black}{($\approx \mathcal{K}r^4$)} and noise ($\approx
\sigma^2(\sqrt{d}+\sqrt{D-d})$). Clearly, we must have
$\delta_{\text{informal}} > 0$ to approximate the appropriate linear
subspace, a requirement made formal by Theorem \ref{thm:mainresult1}
in Section \ref{sec:mainresult}.  In general, when
$\delta_{\text{informal}}$ is zero (or negative), the bound becomes
infinite (or negative) and is not useful for subspace recovery.
However, the geometric information encoded by \eqref{eq:practical1}
offers more insight.  For example, we observe that a small
$\delta_{\text{informal}}$ indicates that the estimated subspace
contains a direction orthogonal to the true tangent space (due to the
curvature or noise).  We therefore consider $\delta_{\text{informal}}$
to be the condition number for subspace recovery and use it to develop
our geometric interpretation for the bound.

The noise-curvature trade-off is readily apparent from \eqref{eq:practical1}.
The linear and curvature contributions are small for small values of $r$.  Thus for a small neighborhood ($r$ small), the denominator \eqref{eq:delta_informal} is either negative or ill conditioned for most values of $\sigma$ and the bound becomes large.  This matches our intuition as we have not yet encountered much curvature but the linear structure has also not been explored.  Therefore, the noise dominates the early behavior of this bound and an approximating subspace may not be recovered from noise.  As the neighborhood radius $r$ increases, the conditioning of the denominator improves, and the bound is controlled by the $1/\sqrt{N}$ behavior of the numerator.  This again corresponds with our intuition: the addition of more points serves to overcome the effects of noise as the linear structure is explored.  Thus, when $\delta_{\text{informal}}^{-1}$ is well conditioned, the bound on the angle may become smaller with the inclusion of more points.  Eventually $r$ becomes large enough such that the curvature contribution approaches the size of the linear contribution and $\delta_{\text{informal}}^{-1}$ becomes large.  The $1/\sqrt{N}$ term is overtaken by the ill conditioning of the denominator and the bound is again forced to become large.  The noise-curvature trade-off, seen analytically here in \eqref{eq:practical1} and \eqref{eq:delta_informal}, will be demonstrated numerically in Section~\ref{sec:numerical}.

Enforcing a well conditioned recovery bound \eqref{eq:practical1}
yields a geometric uncertainty principle quantifying the amount of curvature and noise we may tolerate.
To recover an approximating subspace, we must have:\\

\noindent \textbf{Geometric Uncertainty Principle.}
\begin{equation} \label{eq:uncertainty}
  \textcolor{black}{
    \mathcal{K}\sigma^2 ~<~ \frac{d+4}{2(\sqrt{d} + \sqrt{D-d})}
  }
\end{equation}
By preventing the curvature and noise level from simultaneously
becoming large, this requirement ensures that the linear structure of
the data is recoverable.  With high probability, the noise component
normal to the tangent plane concentrates on a sphere with mean
curvature $1/(\sigma\sqrt{D-d})$.  As will be shown, this uncertainty
principle expresses the intuitive notion that the curvature of the
manifold must be less than the curvature of this noise-ball.
Otherwise, the combined effects of noise and curvature perturbation
prevent an accurate estimate of the local tangent space.

\textcolor{black}{We note that the concept of a geometric uncertainty principle
  also appears in the context of the computation of the homology of
  the manifold $\mathcal{M}$ in \cite{Niyogi11}. As explained in
  detail in Section \ref{sec:uncertainty}, the two principles are
  strikingly similar.}

The remainder of the paper is organized as follows.  Section \ref{sec:prelim}
provides the notation, geometric model, and necessary mathematical
formulations used throughout this work.  Eigenspace perturbation theory is reviewed in this section.  The main results are stated formally in Section \ref{sec:mainresult}.  We demonstrate the accuracy of our results and test the sensitivity to errors in parameter estimation in Section \ref{sec:numerical}.
\textcolor{black}{Section \ref{sec:numerical2} presents the modifications
  that are needed to account for the sampling density of the noisy
  points, and introduces two plug-in estimates that can be used in
practice to  apply the theoretical results of Section
\ref{sec:mainresult} to a real data set.}
We conclude in Section \ref{sec:discussion} with a discussion of the relationship to previously established results and further algorithmic considerations.  Technical results and proofs are presented in the appendices.

\section{Mathematical Preliminaries} \label{sec:prelim}

\subsection{Geometric Data Model}
\textcolor{black}{ 
A $d$-dimensional Riemannian manifold of
  codimension~1 may be described locally about a reference point $x_0$
  by the surface $y = f(\ell_1,\dots, \ell_d)$, where $\ell_i$ is a
  coordinate in the tangent plane, $\TP$, to the manifold at $x_0$.
  After translating $x_0$ to the origin, we have
  \begin{equation*}
    x_0 = [0 ~ 0 ~ \cdots ~ 0]^T,
  \end{equation*}
  and a
  rotation of the coordinate system can align the coordinate axes with
  the principal directions associated with the principal curvatures at $x_0$.
  Aligning the coordinate axes with
  the plane tangent to $\mathcal{M}$ at $x_0$ gives a local quadratic
  approximation to the manifold.  Using this choice of coordinates, the
  manifold may be described locally \cite{Giaquinta} by the Taylor
  series of $f$ at $x_0$:
  \begin{equation}
    y = f(\ell_1,\dots, \ell_d) = \frac{1}{2}(\kappa_1 \ell_1^2 + \dots + \kappa_d \ell_d^2) + o\left(\ell_1^2+\dots+\ell_d^2\right),
  \end{equation}
  where $\kappa_1, \dots, \kappa_d$ are the principal curvatures of $\mathcal{M}$ at $x_0$.  In this coordinate system, a point $x$ in a neighborhood of $x_0$ has the form
  \begin{equation*}
    x = [\ell_1 ~ \ell_2 ~ \cdots ~ \ell_d ~ f(\ell_1,\dots, \ell_d)]^T.
  \end{equation*}
}
Generalizing to a $d$-dimensional manifold of arbitrary codimension in $\mathbb{R}^D$, there exist $(D-d)$ functions
\begin{equation} \label{eq:local-model-full}
  f_{i}(\ell) = \frac{1}{2}(\kappa^{(i)}_1 \ell_1^2 + \dots + \kappa^{(i)}_d \ell_d^2) + o\left(\ell_1^2+\dots+\ell_d^2\right),
\end{equation}
for $i = (d+1), \dots, D$, with $\kappa^{(i)}_1, \dots,
\kappa^{(i)}_d$ representing the principal curvatures in
\textcolor{black}{the $i$th normal direction} at $x_0$.  Then, given the
coordinate system aligned with the principal directions, a point in a
neighborhood of $x_0$ has coordinates $\left[\ell_1, \dots, \ell_d,
  f_{d+1}, \dots, f_{D}\right]$.  We truncate the Taylor expansion
\eqref{eq:local-model-full} and use the quadratic approximation
\begin{equation}
  f_{i}(\ell) = \frac{1}{2}(\kappa^{(i)}_1 \ell_1^2 + \dots + \kappa^{(i)}_d \ell_d^2),
  \label{eq:local-model}
\end{equation}
$i = (d+1),\dots,D$, to describe the manifold locally.

Consider now discrete samples from $\mathcal{M}$ obtained by uniformly
sampling the first $d$ coordinates ($\ell_1,\dots,\ell_d)$ in the
tangent space inside $B^d_{x_0}(r)$, the $d$-dimensional ball of
radius $r$ centered at $x_0$, with the remaining $(D-d)$ coordinates
given by \eqref{eq:local-model}.  Because we are sampling from a
noise-free linear subspace, the number of points $N$ captured inside
$B^d_{x_0}(r)$ is a function of the sampling density $\rho$:
\begin{equation}
  N = \rho v_d r^d,
  \label{eq:sampling-density}
\end{equation} 
where $v_d$ is the volume of the $d$-dimensional unit ball.
\textcolor{black}{ The sampled points are assumed to be in general
  linear position, a standard assumption when sampling from a linear
  subspace (see Remark \ref{remark:sampling})}.
  
Finally, we assume the sample points of $\mathcal{M}$ are contaminated with an additive Gaussian noise vector $e$ drawn from the $\N\left(0,\sigma^2I_D\right)$ distribution.  Each sample $x$ is a $D$-dimensional vector, and $N$ such samples may be stored as columns of a matrix $X \in \mathbb{R}^{D\times N}$.  The coordinate system above allows the decomposition of $x$ into its linear (tangent plane) component $\ell$, its quadratic (curvature) component $c$, and noise $e$, three $D$-dimensional vectors
\begin{align}
  \ell &= [\ell_1 ~ \ell_2 ~ \cdots ~~ \ell_d ~~ 0 ~~~ \cdots~~~ 0]^T \label{eq:vec-ell}\\
  c &= [0 ~~~~\cdots~~~~ 0 ~ c_{d+1} ~ \cdots ~ c_D]^T \label{eq:vec-c}\\
  e &= [e_1 ~ e_2 ~ \quad\quad\;\; \cdots \quad\quad\;\; e_D]^T \label{eq:vec-e}
\end{align}
such that the last $(D-d)$ entries of $c$ are of the form $c_i = f_i(\ell)$, $i=(d+1),\dots,D$.
We may store the $N$ samples of $\ell$, $c$, and $e$ as columns of matrices $L$, $C$, $E$, respectively, such that our data matrix is decomposed as
\begin{equation}
  X = L + C + E.
  \label{eq:decomp}
\end{equation}

The true tangent space we wish to recover is given by the PCA of $L$.
Because we do not have direct access to $L$, we work with $X$ as a
proxy, and instead recover a subspace spanned by the corresponding
eigenvectors of $XX^T$.  We will study how close this recovered
invariant subspace of $XX^T$ is to the corresponding invariant
subspace of $LL^T$ as a function of scale.  Throughout this work,
scale refers to the number of points $N$ in the local neighborhood
within which we perform PCA.  Given a fixed density of points, scale
may be equivalently quantified as the radius $r$ about the reference
point $x_0$ defining the local neighborhood.

\begin{remark}
  Of course it is unrealistic for the data to be
  observed in the described coordinate system.  As noted, we may use a
  rotation to align the coordinate axes with the principal directions
  associated with the principal curvatures.  Doing so allows us to write
  \eqref{eq:local-model} as well as \eqref{eq:decomp}.  Because we will
  ultimately quantify the norm of each matrix using the
  unitarily-invariant Frobenius norm, this rotation will not affect our
  analysis.  We therefore proceed by assuming that the coordinate axes
  align with the principal directions.
\end{remark}

\begin{remark}
  Equation \eqref{eq:local-model} represents an exact quadratic embedding of $\mathcal{M}$.  While it may be interesting to consider more general embeddings, as is done for the noise-free case in \cite{Tyagi}, a Taylor expansion followed by rotation and translation will result in an embedding of the form \eqref{eq:local-model-full}.  Noting that the numerical results of \cite{Tyagi} indicate no loss in accuracy when truncating higher-order terms, proceeding with an analysis of \eqref{eq:local-model} remains sufficiently general.
\end{remark}

\begin{remark}
  \label{remark:sampling}
  \textcolor{black}{In a non-pathological configuration (e.g., points observed in general linear position), only $d+1$ sample points are needed
    to ensure that the top $d$ eigenvectors of $LL^T$ span the true
    tangent space.  It has been noted in the literature (e.g., \cite{Rudelson-Isotropic,Vershynin-HowClose})
    that $\mathcal{O}(d\log d)$ points should be sampled for the empirical covariance
    matrix to be close in norm to the population covariance, with high
    probability.  Strictly enforcing this sampling condition is a very
    mild requirement for our setting, in which the sampling density
    $\rho$ (see equation \eqref{eq:sampling-density}) is usually large
    and the extra logarithmic factor of $d$ is easily achieved.
    Further, this logarithmic factor is implicitly present in our
    analysis as a consequence of the lower bound on the smallest
    eigenvalue of $LL^T$ (see Appendix
    \ref{subsec:linearEigenvalues}).
    We also note that we do not intend to analyze the
    extremely small scales (very small $N$) where finite sample
    effects create instability and prevent a meaningful analysis.}
\end{remark}

\subsection{Perturbation of Invariant Subspaces} \label{sec:subspace-perturb}

Given the decomposition of the data \eqref{eq:decomp}, we have
\begin{equation}
  XX^T ~=~ LL^T + CC^T + EE^T + LC^T + CL^T + LE^T + EL^T + CE^T + EC^T.
  \label{eq:AA^T}
\end{equation}
We introduce some notation to account for the centering required by PCA.  Define the sample mean of $N$ realizations of random vector $m$ as
\begin{equation}
  \overline{m} = \frac{1}{N}\sum_{i=1}^N m^{(i)},
  \label{eq:sample-mean}
\end{equation}
where $m^{(i)}$ denotes the $i$th realization.  Letting $\mathbf{1}_N$ represent the column vector of $N$ ones, define
\begin{equation}
  \overline{M} = \overline{m}\mathbf{1}_N^T
\end{equation}
to be the matrix with $N$ copies of $\overline{m}$ as its columns.
Finally, let $\widetilde{M}$ denote the centered version of $M$:
\begin{equation}
  \widetilde{M} ~=~ M - \overline{M}.
\end{equation}
Then we have
\begin{equation}
  \widetilde{X}\widetilde{X}^T ~=~ \widetilde{L}\widetilde{L}^T + \widetilde{C}\widetilde{C}^T + \widetilde{E}\widetilde{E}^T + \widetilde{L}\widetilde{C}^T + \widetilde{C}\widetilde{L}^T + \widetilde{L}\widetilde{E}^T + \widetilde{E}\widetilde{L}^T + \widetilde{C}\widetilde{E}^T + \widetilde{E}\widetilde{C}^T.
  \label{eq:AA^T-TILDE}
\end{equation}

The problem may be posed as a perturbation analysis of invariant subspaces.  Rewrite \eqref{eq:AA^T} as
\begin{equation}
  \frac{1}{N}\widetilde{X}\widetilde{X}^T = \frac{1}{N}\widetilde{L}\widetilde{L}^T + \Delta,
  \label{eq:Delta}
\end{equation}
where
\begin{equation}
  \Delta = \frac{1}{N}(\widetilde{C}\widetilde{C}^T + \widetilde{E}\widetilde{E}^T + \widetilde{L}\widetilde{C}^T + \widetilde{C}\widetilde{L}^T + \widetilde{L}\widetilde{E}^T + \widetilde{E}\widetilde{L}^T + \widetilde{C}\widetilde{E}^T + \widetilde{E}\widetilde{C}^T)
  \label{eq:fullDelta}
\end{equation}
is the perturbation that prevents us from working directly with $\widetilde{L}\widetilde{L}^T$.  The dominant eigenspace of $\widetilde{X}\widetilde{X}^T$ is therefore a perturbed version of the dominant eigenspace of $\widetilde{L}\widetilde{L}^T$.  Seeking to minimize the effect of this perturbation, we look for the scale $N^*$ (equivalently $r^*$) at which the dominant eigenspace of $\widetilde{X}\widetilde{X}^T$ is closest to that of $\widetilde{L}\widetilde{L}^T$.  Before proceeding, we review material on the perturbation of eigenspaces relevant to our analysis.  The reader familiar with this topic is invited to skip directly to Theorem \ref{thm:Stewart}.

The distance between two subspaces of $\mathbb{R}^D$ can be defined as
the spectral norm of the difference between their respective
orthogonal projectors \cite{GVL}.  As we will always be considering
two equidimensional subspaces, this distance is equal to the sine of
the largest principal angle between the subspaces.
To control all such principal angles, we state our results using the
Frobenius norm of this difference.
Our goal is therefore to control the behavior of $\|P-\widehat{P}\|_F$,
where $P$ and $\widehat{P}$ are the orthogonal projectors onto the subspaces computed from $L$ and $X$, respectively.

The norm $\|P-\widehat{P}\|_F$ may be bounded by the classic
$\sin\Theta$ theorem of Davis and Kahan \cite{Davis}.  We will use a
version of this theorem presented by Stewart (Theorem V.2.7 of
\cite{Stewart}), modified for our specific purpose.  First, we
establish some notation, following closely that found in
\cite{Stewart}.  Consider the eigendecompositions
\begin{align}
  \frac{1}{N}\widetilde{L}\widetilde{L}^T &= U\Lambda U^T = [U_1~U_2]~\begin{bmatrix} \Lambda_1 & \\ & \Lambda_2 \end{bmatrix} ~[U_1~U_2]^T, \label{eq:eigdecompL} \\
  \frac{1}{N}\widetilde{X}\widetilde{X}^T &= \widehat{U} \widehat{\Lambda} \widehat{U}^T = [\widehat{U}_1~\widehat{U}_2]~\begin{bmatrix} \widehat{\Lambda}_1 & \\ & \widehat{\Lambda}_2 \end{bmatrix}~[\widehat{U}_1~\widehat{U}_2]^T,  \label{eq:eigdecompA}
\end{align}
such that the columns of $U$ are the eigenvectors of $\frac{1}{N}\widetilde{L}\widetilde{L}^T$ and the columns of $\widehat{U}$ are the eigenvectors of $\frac{1}{N}\widetilde{X}\widetilde{X}^T$.
The eigenvalues of $\frac{1}{N}\widetilde{L}\widetilde{L}^T$ are arranged in descending order as the entries of diagonal matrix $\Lambda$.  The eigenvalues are also partitioned such that diagonal matrices $\Lambda_1$ and $\Lambda_2$ contain the $d$ largest entries of $\Lambda$ and the $(D-d)$ smallest entries of $\Lambda$, respectively.
The columns of $U_1$ are those eigenvectors associated with the $d$ eigenvalues in $\Lambda_1$, the columns of $U_2$ are those eigenvectors associated with the $(D-d)$ eigenvalues in $\Lambda_2$, and the eigendecomposition of $\frac{1}{N}\widetilde{X}\widetilde{X}^T$ is similarly partitioned.  The subspace we recover is spanned by the columns of $\widehat{U}_1$ and we wish to have this subspace as close as possible to the tangent space spanned by the columns of $U_1$.  The orthogonal projectors onto the tangent and computed subspaces, $P$ and $\widehat{P}$ respectively, are given by
\begin{align*}
  P &= U_1 U_1^T \quad \text{and} \quad
  \widehat{P} = \widehat{U}_1 \widehat{U}_1^T.
\end{align*}
Define $\lambda_d$ to be the $d$th largest eigenvalue of $\frac{1}{N}\widetilde{L}\widetilde{L}^T$, or the last entry on the diagonal of $\Lambda_1$.  This eigenvalue corresponds to variance in a tangent space direction.

We are now in position to state the theorem.  Note that we have made
use of the fact that the columns of $U$ are the eigenvectors of
$\widetilde{L}\widetilde{L}^T$, that $\Lambda_1, \Lambda_2$ are
Hermitian (diagonal) matrices, and that the Frobenius norm is used to
measure distances.  The reader is referred to \cite{Stewart} for the
theorem in its original form.
\begin{theorem}[Davis \& Kahan \cite{Davis}, Stewart \cite{Stewart}] \label{thm:Stewart}
  Let
  \begin{equation*}
    \delta = \lambda_d - \left\|U_1^T \Delta U_1\right\|_F  - \left\|U_2^T \Delta U_2\right\|_F
  \end{equation*}
  and consider
  \begin{itemize}
  \item (Condition 1) \quad $\delta > 0$
  \item (Condition 2) \quad  $\left\|U_1^T \Delta U_2\right\|_F \left\|U_2^T \Delta U_1\right\|_F < \frac{1}{4}\delta^2$.
  \end{itemize}
  Then, provided that conditions 1 and 2 hold,
  \begin{equation}
    \left\|P - \widehat{P}\right\|_F ~\leq~ 2\sqrt{2}~\frac{\left\|U_2^T \Delta U_1\right\|_F}{\delta}.\\
    \label{eq:stewart_thm}
  \end{equation}
\end{theorem}
\vskip2ex

\textcolor{black}{It is instructive to consider the perturbation
  $\Delta$ as an operator with range in $\mathbb{R}^D$ and quantify its effect on the existing invariant subspaces.  Consider first the idealized case where $U_1$ is an invariant subspace of $\Delta$, i.e.,
  $\Delta$ maps points from the
  column space of $U_1$ to the column space of $U_1$.
  Clearly, $U_2^T \Delta U_1 = 0$ in this case as the subspace spanned by $U_1$ remains invariant
  under the action of $\Delta$, and the perturbation angle is zero.
  In general, however, we cannot expect such an idealized restriction
  for the range of $\Delta$ and we therefore expect that $\Delta U_1$
  will have a component that is normal to the tangent space.  The
  numerator $\|U_2^T\Delta U_1\|_F$ of \eqref{eq:stewart_thm} measures
  this normal component, thereby quantifying the effect of the
  perturbation on the tangent space.}  Then $\|U_1^T\Delta U_1\|_F$
measures the component that remains in the tangent space after the
action of $\Delta$.  As this component does not contain curvature,
$\|U_1^T\Delta U_1\|_F$ corresponds to the spectrum of the noise
projected in the tangent space.  Similarly, $\|U_2^T\Delta U_2\|_F$
measures the spectrum of the curvature and noise perturbation normal
to the tangent space.  Thus, when $\Delta$ leaves the column space of
$U_1$ mostly unperturbed (i.e., $\|U_2^T\Delta U_1\|_F$ is small) and
the spectrum of the tangent space is well separated from that of the
noise and curvature, the estimated subspace will form only a small
angle with the true tangent space.  In the next section, we use the
machinery of this classic result to bound the angle caused by the
perturbation $\Delta$ and develop an interpretation of the conditions
of Theorem \ref{thm:Stewart} suited to the noise-curvature trade-off.

\section{Main Results} \label{sec:mainresult}

Given the framework for analysis developed above, the terms appearing in the statement of Theorem \ref{thm:Stewart} ($\left\|U_1^T \Delta U_1\right\|_F$, $\left\|U_2^T \Delta U_2\right\|_F$, $\left\|U_2^T \Delta U_1\right\|_F$, $\left\|U_1^T \Delta U_2\right\|_F$, and $\lambda_d$) must be controlled.
We notice that $\Delta$ is a symmetric matrix, so that $\left\|U_1^T \Delta U_2\right\|_F = \left\|U_2^T \Delta U_1\right\|_F$.
Using the triangle inequality and the geometric constraints
\begin{equation}
  U_1^TC = 0 \quad \text{and} \quad U_2^TL = 0,
  \label{eq:geometric-constraints}
\end{equation}
the norms may be controlled by bounding the contribution of each term in the perturbation $\Delta$:
\begin{align*}
  \left\|U_1^T \Delta U_1\right\|_F &\leq 2\left\|U_1^T \frac{1}{N}\widetilde{L}\widetilde{E}^T U_1\right\|_F + \;\;\, \left\|U_1^T\frac{1}{N}\widetilde{E}\widetilde{E}^T U_1\right\|_F, \\
  \left\|U_2^T \Delta U_2\right\|_F &\leq \;\;\, \left\|U_2^T\frac{1}{N}\widetilde{C}\widetilde{C}^T U_2\right\|_F + 2\left\|U_2^T \frac{1}{N}\widetilde{C}\widetilde{E}^T U_2\right\|_F + \left\|U_2^T\frac{1}{N}\widetilde{E}\widetilde{E}^T U_2\right\|_F, \\
  \left\|U_2^T \Delta U_1\right\|_F & \leq \;\;\, \left\|U_2^T\frac{1}{N}\widetilde{C}\widetilde{L}^T U_1\right\|_F + \;\;\, \left\|U_2^T \frac{1}{N}\widetilde{E}\widetilde{L}^T U_1\right\|_F + \,\left\|U_2^T\frac{1}{N}\widetilde{C}\widetilde{E}^T U_1\right\|_F + \, \left\|U_2^T\frac{1}{N}\widetilde{E}\widetilde{E}^T U_1\right\|_F.
\end{align*}
Importantly, we seek control over each (right-hand side) term in the finite-sample regime, as we assume a possibly large but finite number of sample points $N$.  Therefore, bounds are derived through a careful analysis employing concentration results and techniques from non-asymptotic random matrix theory.  The technical analysis is presented in the appendix and proceeds by analyzing three distinct cases: the covariance of bounded random matrices, unbounded random matrices, and the interaction of bounded and unbounded random matrices.  The eigenvalue $\lambda_d$ is bounded again using random matrix theory.  In all cases, care is taken to ensure that bounds hold with high probability that is independent of the ambient dimension $D$.

\begin{remark}
  Other, possibly tighter, avenues of analysis may be possible for some of the bounds presented in the appendix.  However, the presented analysis avoids large union bounds and dependence on the ambient dimension to state results holding with high probability.  Alternative analyses are possible, often sacrificing probability to exhibit sharper concentration.  We proceed with a theoretical analysis holding with the highest probability while maintaining accurate results.
\end{remark}

\subsection{Bounding the Angle Between Subspaces} \label{sec:mainresult-angle}

We are now in position to apply Theorem \ref{thm:Stewart} and state our main result.
First, we make the following definitions involving the principal curvatures:
\begin{equation}
  K_i = \sum_{n=1}^d \kappa_n^{(i)}, \quad
  K = \left(\sum_{i=d+1}^D K_i^2\right)^{\frac{1}{2}},
  \label{eq:K}
\end{equation}
\begin{equation}
  K_{nn}^{ij} = \sum_{n=1}^d \kappa_n^{(i)} \kappa_n^{(j)}, \quad 
  K_{mn}^{ij} = \sum_{\substack{m,n=1\\m\neq n}}^d \kappa_m^{(i)}
  \kappa_n^{(j)},
\label{sectional-curvature}
\end{equation}
and
\begin{equation}
  \textcolor{black}{
    \mathcal{K} = \left[\sum_{i=d+1}^D\sum_{j=d+1}^D \left[(d+1)K_{nn}^{ij}-K_{mn}^{ij}\right]^2\right]^{\frac{1}{2}}.
  }
\end{equation}
The constant $K_i$ is the mean curvature (rescaled by a factor of $d$)
in \textcolor{black}{normal direction $i$}, for $(d+1)\leq i \leq D$.
\textcolor{black}{The curvature of the local model is quantified by $K$,
  which is a natural result of our use of the Frobenius norm, and
  $\mathcal{K}$, which results from the expectation of the norm of the
  curvature covariance.}  Note that $K_iK_j = K_{nn}^{ij} +
K_{mn}^{ij}$.  We also define the constants
\begin{equation}
  K_{i}^{(+)} = \left(\sum_{n=1}^d |\kappa_n^{(i)}|^2\right)^{\frac12}, \quad \text{and} \quad K^{(+)} = \left(\sum_{i=d+1}^D (K_i^{(+)})^2\right)^{\frac{1}{2}}
\end{equation}
to be used when strictly positive curvature terms are required.

The main result is formulated in the appendix and makes the following benign assumptions on the number of sample points $N$ and the probability constants $\xi$ and $\xi_{\lambda}$:
\begin{equation*}
  N > 4(\max(\sqrt{d},\sqrt{D-d}) + \xi), \quad
  \xi < 0.7\sqrt{d(D-d)}, \quad \text{and} \quad \xi_{\lambda} < \frac{3}{\sqrt{d+2}} \sqrt{N},
\end{equation*}
\textcolor{black}{in addition to the requirement that $N \geq \mathcal{O}(d\log d)$ for the points observed in general linear position (see Remark \ref{remark:sampling}).
We note that the assumptions are easily satisfied as
  we envision a sampling density such that $N$ is large (but finite).}
Further, the assumptions listed above are not crucial to the result but allow for a
more compact presentation.

\begin{theorem}[Main Result]  \label{thm:mainresult1}
  \textcolor{black}{
    Let
    \begin{align}
      \delta &= \frac{r^2}{d+2} - \frac{\mathcal{K}r^4}{2(d+2)^2(d+4)} - \sigma^2\left(\sqrt{d} + \sqrt{D-  d}\right) - \frac{1}{\sqrt{N}}\zeta_{\text{denom}}(\xi,\xi_{\lambda})\\
      &\hspace{-0.9in} \text{and} \nonumber \\
      \beta &= \frac{1}{\sqrt{N}}\left[K^{(+)}r^3 \nu(\xi) + \sigma\sqrt{d(D-d)}\eta(\xi,\xi_{\lambda}) + \frac{1}{\sqrt{N}}\zeta_{\text{numer}}(\xi)\right].
    \end{align}
    If the following conditions hold (in addition to the benign assumptions stated above):
    \begin{itemize}
    \item (Condition 1) \quad $\delta > 0$,
    \item (Condition 2) \quad $\beta < \frac{1}{2}\delta$,
    \end{itemize}
    then
    \begin{equation}
      \left\|P - \widehat{P}\right\|_F \leq~ \frac{2\sqrt{2}\beta}{\delta} ~~=~~
      \frac{2\sqrt{2}\,\frac{1}{\sqrt{N}}\left[K^{(+)}r^3 \nu(\xi) + \sigma\sqrt{d(D-d)}\eta(\xi,\xi_{\lambda}) + \frac{1}{\sqrt{N}}\zeta_{\text{numer}}(\xi)\right]}
      {\frac{r^2}{d+2} - \frac{\mathcal{K}r^4}{2(d+2)^2(d+4)} - \sigma^2\left(\sqrt{d} + \sqrt{D-  d}\right) - \frac{1}{\sqrt{N}}\zeta_{\text{denom}}(\xi,\xi_{\lambda})}
      \label{eq:main-result1}
    \end{equation}
    with probability greater than
    \begin{equation}
      1-2de^{-\xi_{\lambda}^2}-9e^{-\xi^2}
      \label{eq:main-result-prob}
    \end{equation}
    over the joint random selection of the sample points and random realization of the noise,
    where the following definitions have been made to ease the presentation:
  }
  \begin{itemize}
    
  \item \textcolor{black}{geometric and noise terms}
    {\allowdisplaybreaks
      \begin{align*}
        \nu(\xi) &= \frac{1}{2}\frac{(d+3)}{(d+2)}p_1(\xi), & \text{(linear--curvature)} \\
        \eta_1 &= \sigma, & \text{(noise)} \\
        \eta_2(\xi_{\lambda}) &= \frac{r}{\sqrt{d+2}}p_2(\xi_{\lambda}), & \text{(linear--noise)}\\ 
        \eta_3(\xi) &= \frac{\mathcal{K}^{1/2}r^2}{(d+2)\sqrt{2(d+4)}}p_5(\xi), & \text{(curvature--noise)} \\
        \eta(\xi,\xi_{\lambda}) &= p_3(\xi,\sqrt{d(D-d)}) \bigg[\eta_1 + \eta_2(\xi_{\lambda}) + \eta_3(\xi)\bigg],
      \end{align*}
    }
    
  \item finite sample correction terms (numerator)
    {\allowdisplaybreaks
      \begin{align*}
        \zeta_1(\xi) &= \frac{1}{2}K^{(+)}r^3 p_1^2(\xi), & \text{(linear--curvature)} \\
        \zeta_2(\xi) &= \sigma^2\sqrt{d(D-d)}p_3(\xi,\sqrt{d(D-d)})p_4(\xi,\sqrt{D-d}), & \text{(noise)}\\
        \zeta_{\text{numer}}(\xi) &= \zeta_1(\xi) + \zeta_2(\xi),
      \end{align*}
    }

  \item finite sample correction terms (denominator)
    {\allowdisplaybreaks
      \begin{align*}
        \zeta_3(\xi_{\lambda}) &= \frac{r^2}{d+2}\left[p_0(\xi_\lambda) + \left(\frac{2}{\sqrt{N}}-\frac{1}{N^{3/2}}\right)\left(1-\frac{p_0(\xi_\lambda)}{\sqrt{N}}\right)\right], & \text{(linear)}\\
        \zeta_4(\xi) &= \frac{(K^{(+)})^2r^4}{4}\left(p_1(\xi) + \frac{1}{\sqrt{N}}p_1^2(\xi)\right), &\text{(curvature)}\\
        \zeta_5(\xi,\xi_{\lambda}) &= 2 r \sigma \frac{d}{\sqrt{d+2}} p_2(\xi_{\lambda}) p_3(\xi,d), & \text{(linear--noise)}\\
        \zeta_6(\xi) &= 2 \mathcal{K}^{\frac{1}{2}} r^2 \sigma \frac{(D-d)}{(d+2)\sqrt{2(d+4)}}p_3(\xi,D-d)p_5(\xi), &\text{(curvature--noise)} \\
        \zeta_7(\xi) &= \frac{5}{2}\sigma^2\left[\sqrt{d}p_4(\xi,\sqrt{d}) + \sqrt{D-d}p_4(\xi,\sqrt{D-d}) \right], &\text{(noise)}\\
        \zeta_{\text{denom}}(\xi,\xi_{\lambda}) &= \zeta_3(\xi) + \zeta_4(\xi) + \zeta_5(\xi,\xi_{\lambda}) + \zeta_6(\xi) + \zeta_7(\xi),
      \end{align*}
    }

    and

  \item probability-dependent terms (i.e., terms depending on the probability constants)
    {\allowdisplaybreaks
      \begin{equation*}
        p_0(\xi) = \xi\frac{\sqrt{8(d+2)}}{(1-\frac{1}{N})}, \qquad
        p_1(\xi) = \left(2+\xi\sqrt{2}\right), \qquad
        p_2(\xi) = \left(1 + \xi \frac{5\sqrt{d+2}}{\sqrt{N}} \right),
      \end{equation*}
      \begin{equation*}
        p_3(\xi,\omega) = \left(1 + \frac{6}{5} \frac{\xi}{\omega}\right), \qquad
        p_4(\xi,\omega) = \left(\omega + \xi\sqrt{2} \right),
      \end{equation*}
      \begin{equation*}
        p_5(\xi) = \left(1 + \frac{1}{\sqrt{N}}\frac{(K^{(+)})^2}{2\mathcal{K}}(d+2)^2(d+4)(p_1(\xi)+\frac{1}{\sqrt{N}} p_1^2(\xi))\right)^{1/2}.
      \end{equation*}
    }

  \end{itemize}

  \noindent \textcolor{black}{Finally, we recall the relationship $N = \rho v_d r^d$ given by \eqref{eq:sampling-density}.}
  
\end{theorem}

\begin{proof}
  Condition 2 is simplified from its original statement in Theorem \ref{thm:Stewart} by noticing that $\Delta$ is a symmetric matrix so that $\left\|U_1^T \Delta U_2\right\|_F = \left\|U_2^T \Delta U_1\right\|_F$.
  Then, applying the norm bounds computed in the appendix to
  Theorem \ref{thm:Stewart} and choosing the probability constants
  \begin{equation} \label{eq:prob-consts}
    \xi_{\lambda_d} = \xi_{\lambda_1} = \xi_{\lambda} \quad \text{and} \quad \xi_{cc} = \xi_{c\ell} = \xi_{e\ell} = \xi_{ce} = \xi_{e_1} = \xi_{e_2} = \xi_{e_3} = \xi_{c} = \xi
  \end{equation}
  yields the result.
\end{proof}
\vskip2ex

The bound \eqref{eq:main-result1} will be demonstrated in Section \ref{sec:numerical} to accurately track the angle between the true and computed tangent spaces at all scales.
\textcolor{black}{We experimentally observe that the bound is, in general, either decreasing (for the curvature-free case), increasing (for the noise-free case), or decreasing at small scales and increasing at large scales (for the general case).  We therefore expect to be able to locate a scale at which the bound is minimized.  Based on this observation, the optimal scale, $N^*$, for tangent space recovery may be selected as the $N$ for which \eqref{eq:main-result1} is minimized (an equivalent notion of the optimal scale may be given in terms of the neighborhood radius $r$).}
Note that the constants $\xi$ and $\xi_{\lambda}$ need to be selected to ensure that this bound holds with high probability.  For example, setting $\xi = 2$ and $\xi_{\lambda} = 2.75$ yields probabilities of 0.81, 0.80, and 0.76 when $d=3, 10,$ and $50$, respectively.  We also note that the probability given by \eqref{eq:main-result-prob} is more pessimistic than we expect in practice.

As introduced in Section \ref{sec:overview}, we may interpret $\delta^{-1}$ as the condition number for tangent space recovery.  Noting that the denominator in \eqref{eq:main-result1} is a lower bound on $\delta$, we analyze the condition number via the bounds for $\lambda_d$, $\|U_1^T \Delta U_1\|_F$, and $\|U_2^T \Delta U_2\|_F$.  Using these bounds in the Main Result \eqref{eq:main-result1}, we see that when $\delta^{-1}$ is small, we recover a tight approximation to the true tangent space.  Likewise, when $\delta^{-1}$ becomes large, the angle between the computed and true subspaces becomes large.  The notion of an angle loses meaning as $\delta^{-1}$ tends to infinity, and we are unable to recover an approximating subspace.

Condition 1, requiring that the denominator be bounded away from zero, has an important
geometric interpretation.  As noted above, the conditioning of the subspace
recovery problem improves as $\delta$ becomes large.
Condition 1 imposes that the spectrum corresponding to the linear
subspace ($\lambda_d$) be well separated from the spectra of the
noise and curvature
perturbations encoded by $\|U_1^T \Delta U_1\|_F + \|U_2^T \Delta U_2\|_F$.  In this way, condition 1 quantifies our requirement
that there exists a scale such that the linear subspace is
sufficiently decoupled from the effects of curvature and noise.  When
the spectra are not well separated, the angle between the subspaces
becomes ill defined.  In this case, the approximating subspace
contains an eigenvector corresponding to a direction orthogonal to the
true tangent space.
Condition 2 is a technical requirement of Theorem \ref{thm:Stewart}.  Provided that
condition 1 is satisfied, we observe that a sufficient sampling density will ensure
that Condition 2 is met.  Further, we numerically observe that the Main Result \eqref{eq:main-result1} accurately tracks the subspace recovery error even in the case when condition 2 is violated.  In such a case, the bound may not remain as tight as desired but its behavior at all scales remains consistent with the subspace recovery error tracked in our experiments.

Before numerically demonstrating our main result, we
quantify the separation needed between the linear structure and the
noise and curvature with a geometric uncertainty principle.

\subsection{Geometric Uncertainty Principle for Subspace Recovery} \label{sec:uncertainty}

Condition 1 indeed imposes a geometric requirement for tangent space recovery.
Solving for the range of scales for which condition 1 is satisfied and
requiring the solution \textcolor{black}{to} be real yields the geometric uncertainty principle \eqref{eq:uncertainty} stated in Section \ref{sec:overview}.
We note that this result is derived using $\delta_{\text{informal}}$, defined in equation \eqref{eq:delta_informal}, as the full expression for $\delta$ does not allow for an algebraic solution.

The geometric uncertainty principle \eqref{eq:uncertainty} expresses a
natural requirement for the subspace recovery problem, ensuring that
the perturbation to the tangent space is not too large.  Recall that,
with high probability, the noise orthogonal to the tangent space
concentrates on a sphere with mean curvature $1/(\sigma\sqrt{D-d})$.
We therefore expect to require that the curvature of the manifold be
less than the curvature of this noise-ball.  \textcolor{black}{ To
  compare the curvature of the manifold to that of the noise-ball,
  consider the case where all principal curvatures of the manifold are
  equal, and denote them by $\kappa$.  Then \eqref{eq:uncertainty}
  requires that
  \begin{equation}
    \kappa < \frac{1}{\sigma\sqrt{D-d}} \sqrt{\frac{d+4}{4d\left(\sqrt{d}+\sqrt{D-d}\right)}}.
    \label{eq:uncertainty-tight}
  \end{equation}
}
\textcolor{black}{Noting that, for $d \geq 1$, we have
  \begin{equation*}
    \frac{d+4}{4d\left(\sqrt{d}+\sqrt{D-d}\right)} < 1,
  \end{equation*}
  we see that the uncertainty principle \eqref{eq:uncertainty}
  indeed requires that the mean curvature of the manifold be less than
  that of the perturbing noise-ball.}

\textcolor{black}{Intuitively, we might expect that the uncertainty principle would be of the form 
\begin{equation*}
\text{(curvature)} \times (\text{noise-ball radius}) < 1.
\end{equation*}
However, \eqref{eq:uncertainty} is, in fact, more restrictive than our intuition, as illustrated by \eqref{eq:uncertainty-tight}.
  As only finite-sample
  corrections have been neglected in $\delta_{\text{informal}}$,
  \eqref{eq:uncertainty} is of the correct order.  Interestingly, this
  more restrictive requirement for tangent space recovery is only
  accessible through the careful perturbation analysis presented above
  and an estimate obtained by a more naive analysis would be too lax.
}
  \textcolor{black}{The authors in \cite{Niyogi11} present an algorithm to compute
  the homology of a manifold from a data set of noisy points. The
  authors assume that the data are clean samples from a
  manifold perturbed with $(D-d)$-dimensional Gaussian noise along the
  normal fibers. In the context of our model, this is equivalent to
  removing the first $d$ components of the noise vector. The authors
  prove that the algorithm computes, with high probability, the correct
  homology of $\mathcal{M}$, provided that the noise variance $\sigma^2$ satisfies
\begin{equation}
\frac{1}{\mathscr{R}} < \frac{1}{\sigma \sqrt{D-d}}\;c \frac{\sqrt{9} -  \sqrt{8}}{9\sqrt{8}}\quad \text{with $c<1$}.
    \label{eq:uncertainty3}
\end{equation}
The parameter $1/\mathscr{R}$ is an upper bound on all the principal
curvatures ($\mathscr{R}$ is also known as the {\em reach} \cite{Federer59}).
This condition is almost identical to \eqref{eq:uncertainty-tight}.
The geometric uncertainty principle \eqref{eq:uncertainty}
is clearly not an artifact of our analysis, but is deeply rooted in
the geometric and topological understanding of noisy manifolds.}

\section{Experimental Results \textcolor{black}{I: Validating the Theory}\label{sec:numerical}}

In this section we present an experimental study of the tangent space perturbation results given above.  In particular, we demonstrate that the bound presented in the Main Result (Theorem \ref{thm:mainresult1}) accurately tracks the subspace recovery error at all scales.  As this analytic result requires no decompositions of the data matrix, our analysis provides an efficient means for obtaining the optimal scale for tangent space recovery.
We first present a practical use of the Main Result, demonstrating its accuracy when the intrinsic dimensionality, curvature, and noise level are known.  We then experimentally test the stability of the bound when these parameters are only imprecisely available, as is the case when they must be estimated from the data.  Finally, we demonstrate the accurate estimation of the noise level and local curvature.

\subsection{Subspace Tracking and Recovery} \label{subsec:numerical-results}

We generate a data set sampled from a 3-dimensional manifold embedded
in $\mathbb{R}^{20}$ according to the local model
\eqref{eq:local-model} by uniformly sampling $N=1.25\times 10^6$
points inside a ball of radius $1$ in the tangent plane.  Curvature
and the standard deviation $\sigma$ of the added Gaussian noise will
be specified in each experiment.  We compare our bound with the true
subspace recovery error.  The tangent plane at reference point $x_0$
is computed at each scale $N$ via PCA of the $N$ nearest neighbors of
$x_0$.  The true subspace recovery error $\|P-\widehat{P}\|_F$ is then
computed at each scale.  Note that computing the true error requires
$N$ SVDs.  A ``true bound'' is computed by applying Theorem
\ref{thm:Stewart} after measuring each perturbation norm directly from
the data.  While no SVDs are required, this true bound utilizes
information that is not practically available and represents the best
possible bound that we can hope to achieve.  We will compare the mean
of the true error and mean of the true bound over 10 trials (with
error bars indicating one standard deviation) to the bound given by
our Main Result in Theorem \ref{thm:mainresult1}, holding with
probability greater than 0.8.

For the experiments in this section, the bound \eqref{eq:main-result1} is computed with full knowledge of the necessary parameters.  
\textcolor{black}{In our experience, we observe in practice (results not shown) that the deviation of the empirical eigenvalue $\lambda_d$ from its expectation is insignificant over the entire range of relevant scales and therefore neglect its correction term (derived using a Chernoff bound in Appendix \ref{subsec:linearEigenvalues}) for the experiments.  We further note that knowledge of $d$ provides an exact expression for this expectation as no additional geometric information is encoded by $\lambda_d$.}  As the principle curvatures are known, we compute a tighter bound for $\|U_2^T CL^T U_1\|_F$ using $\mathcal{K}$ in place of $K^{(+)}$.  Doing so only affects the height of the curve; its trend as a function of scale is unchanged.
In practice, the important information is captured by tracking the trend of the true error regardless of whether it provides an upper bound to any random fluctuation of the data.
In fact, the numerical results indicate that an accurate tracking of error is possible even when condition 2 of Theorem \ref{thm:mainresult1} is violated.

\begin{table}[!h]
  \caption{Principal curvatures of the manifold for Figures \ref{fig:results}b and \ref{fig:results}c.}
  \label{tab:kk2}
  \centering
  \begin{tabular}{|c||c|c|c|}
    \hline
    $\kappa_i^{(j)}$ & $i=1$ & $i=2$ & $i=3$ \\
    \hline
    \hline
    $j=4,\dots,6$ & 3.0000 & 1.5000 & 1.5000  \\
    \hline
    $j=7,\dots,20$ & 1.6351 & 0.1351 & 0.1351  \\
    \hline
  \end{tabular}
\end{table}

\begin{figure}[!h]
  \centering
  
  \subfigure[]{\includegraphics[width=70mm]{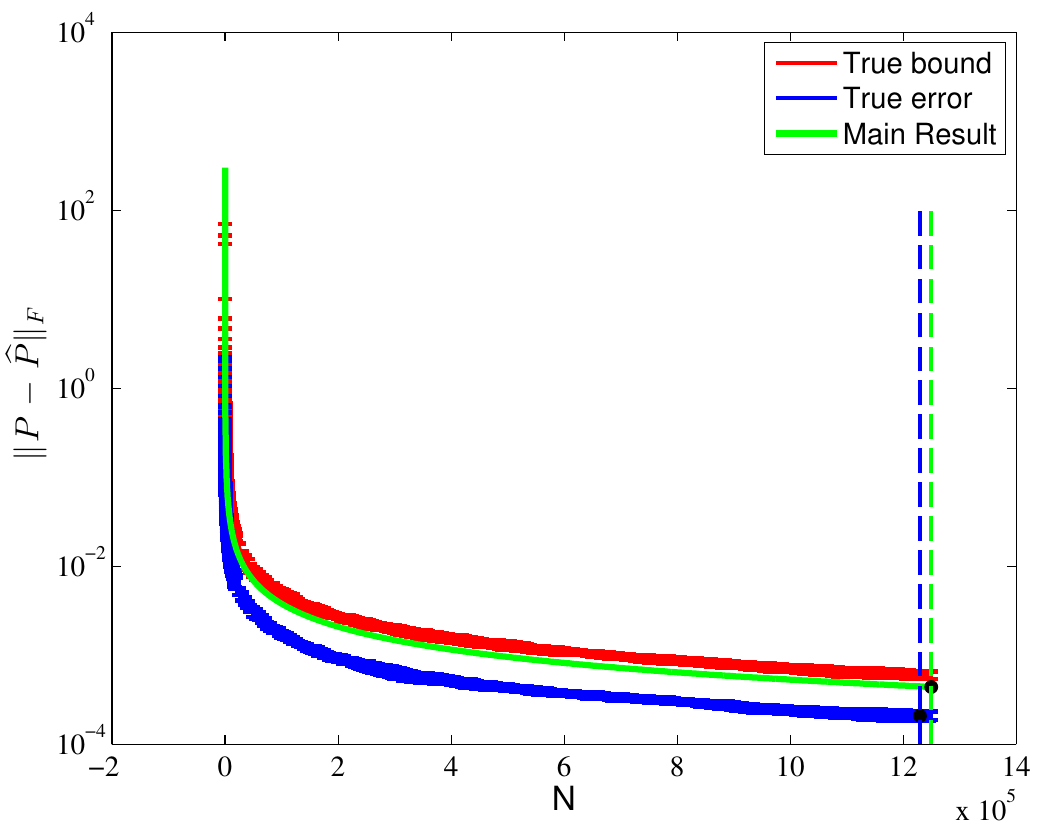}}
  \subfigure[]{\includegraphics[width=70mm]{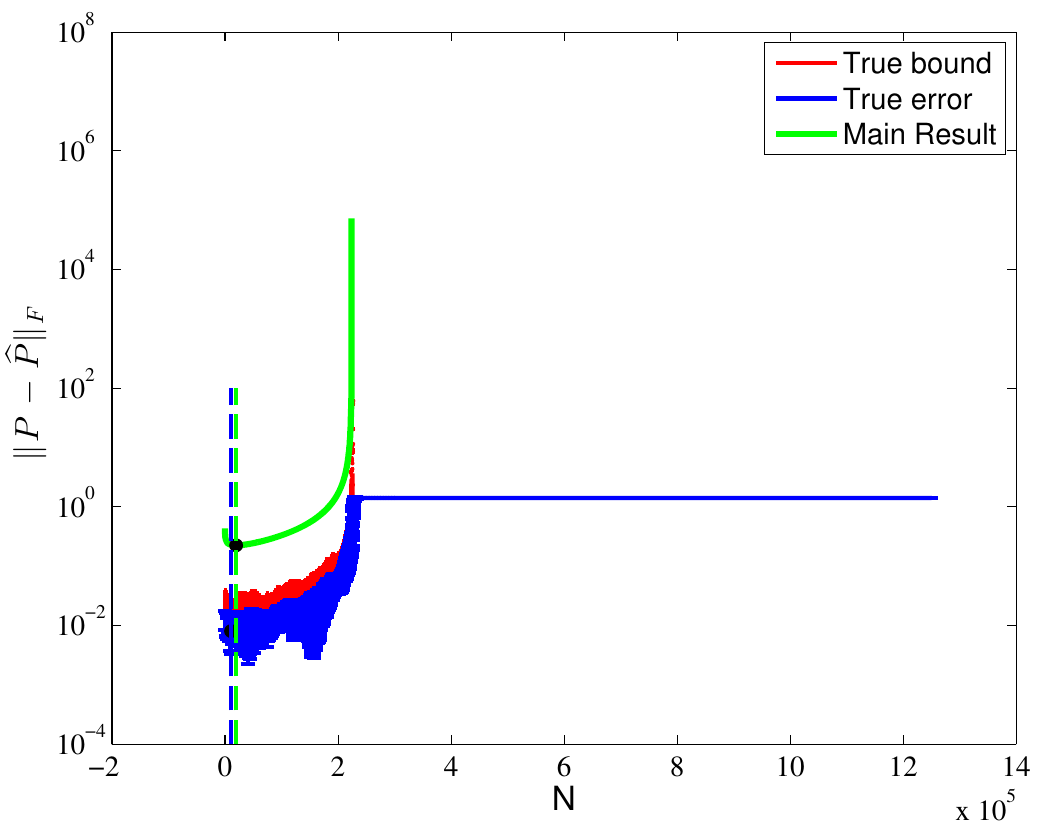}}
  \subfigure[]{\includegraphics[width=70mm]{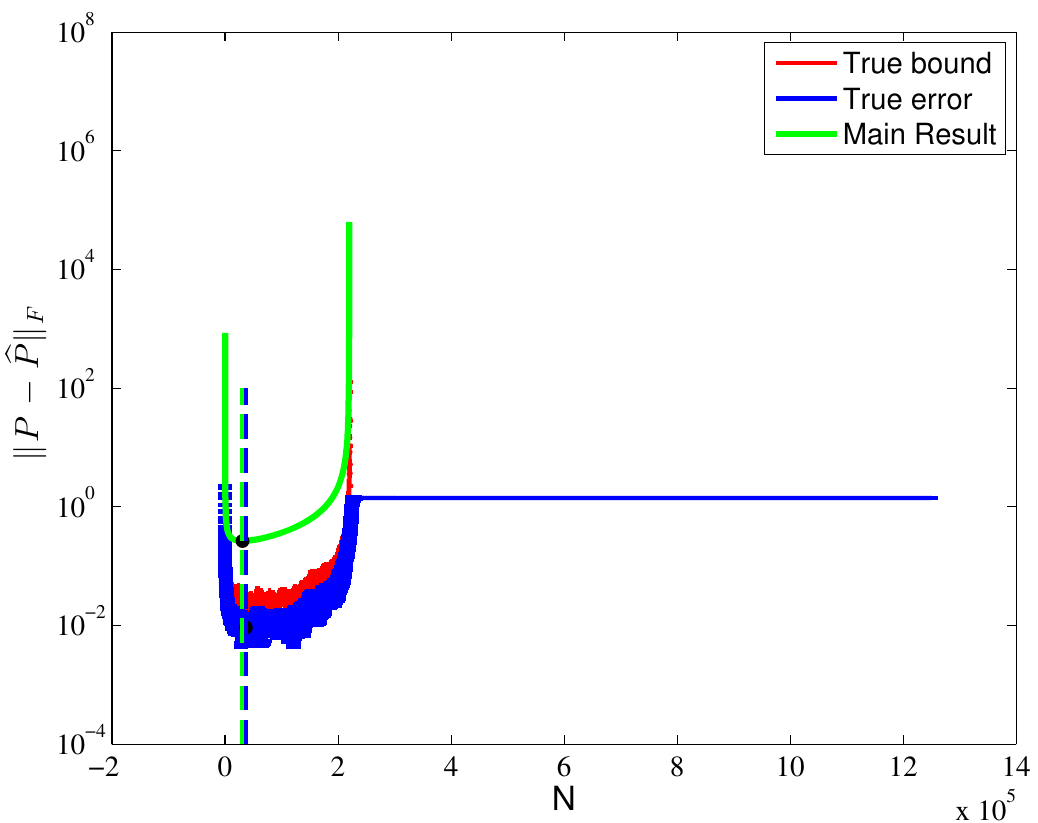}}
  \subfigure[]{\includegraphics[width=70mm]{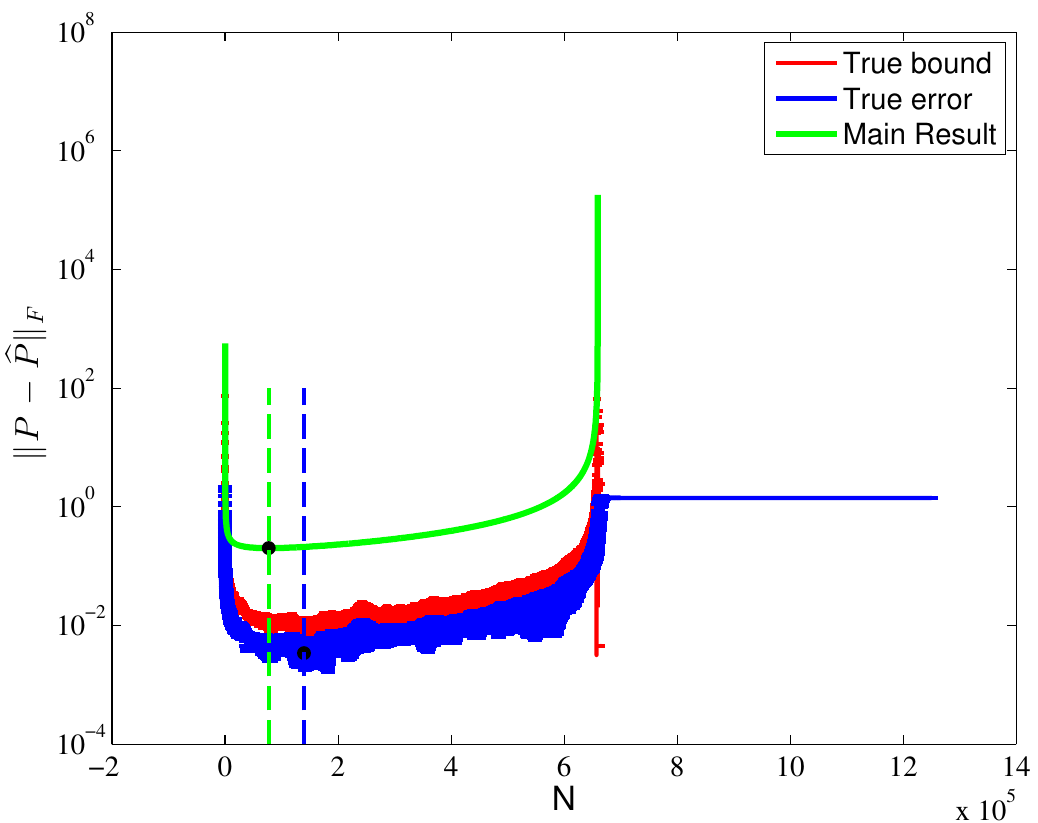}}
  
  \caption{Norm of the perturbation \textcolor{black}{using tangent plane radius $r$}: (a) flat manifold with noise, (b) curved (tube-like) manifold with no noise, (c) curved (tube-like) manifold with noise, (d) curved manifold with noise.  Dashed vertical lines indicate minima of the curves.  Note the logarithmic scale on the Y-axes.  See text for discussion.}
  \label{fig:results}
\end{figure}

The results are displayed in Figure \ref{fig:results}.  Panel (a) shows the noisy $(\sigma=0.01)$ curvature-free (linear subspace) result.  As the only perturbation is due to noise, we expect the error to decay as $1/\sqrt{N}$ as the scale increases.  The curves are shown on a logarithmic scale (for the Y-axis) and decrease monotonically, indicating the expected decay.  Our bound (green) accurately tracks the behavior of the true error (blue) and is nearly identical to the true bound (red).
Panel (b) shows the results for a noise-free manifold with principal curvatures given in Table \ref{tab:kk2} such that $K=12.6025$.  Notice that three of the \textcolor{black}{normal directions} exhibit high curvature while the others are flatter, giving a tube-like structure to the manifold.  In this case, perturbation is due to curvature only and the error increases monotonically (ignoring the slight numerical instability at extremely small scales), as predicted in the discussion of Sections \ref{sec:overview} and \ref{sec:mainresult-angle}.  Eventually, a scale is reached at which there is too much curvature and the bounds blow up to infinity.  This corresponds exactly to where the true error plateaus at its maximum value, indicating that the computed subspace is now orthogonal to the true tangent space.
In this case, condition 1 of Theorem \ref{thm:mainresult1} is violated as there is no longer separation between the linear and curvature spectra, $\delta^{-1}$ becomes large, and our analysis predicts that the computed eigenspace contains a direction orthogonal to the true tangent space.

Figure \ref{fig:results}c shows the results for a noisy ($\sigma = 0.01$) version of the manifold used in panel (b).  Note that the error is large at small scales due to noise and large at large scales due to curvature.  At these scales the bounds are accordingly ill conditioned and track the behavior of the true error when well conditioned.  Figure \ref{fig:results}d shows the results for a manifold again with $K=12.6025$, but with the principal curvatures equal in all \textcolor{black}{normal directions} ($\kappa_i^{(j)} = 1.0189$ for $i=1,\dots,3$ and $j=4,\dots,20$), and noise ($\sigma = 0.01$) is added.  We observe the same general behavior as seen in panel (c), but both the true error and the bounds remain well conditioned at larger scales.  This is explained by the fact that higher curvature is encountered at smaller scales for the manifold corresponding to panel (c) but is not encountered until larger scales in panel (d).
Similar results are shown in Figure \ref{fig:saddle} for a 2-dimensional, noise-free saddle ($\kappa_1^{(3)} = 3, \kappa_2^{(3)} = -3$) embedded in $\R^3$, demonstrating an accurate bound for the case of principle curvatures of mixed signs.

\textcolor{black}{
  The true bound (red) tightly tracks the true error (blue) and is tighter than our bound (green) in all cases except for the curvature-free setting, where a difference on the order of $10^{-3}$ is observed.  This curvature-free bound may be understood by observing that the noise analysis is more precise than that for the curvature (see appendices) and that the height of the bound is controlled by the probability-dependent constants, which have been fixed across all plots for consistency.  In fact, it is possible to choose the probability-dependent constants much larger for the curvature-free setting without violating Condition 2.  Doing so increases the height of the bound (green) to match the height of the ``true bound'' (red) curve (result not shown).  Note that a similar increase for nonzero curvature results in a curve that violates Condition 2.
}

\textcolor{black}{
  In all of the presented experiments, the bound accurately tracks the behavior of the true error.  In fact, the curves are shown to be parallel on a logarithmic scale, indicating that they differ only by multiplicative constants.  These observations further indicate that the triangle inequalities used in bounding the norms $\|U_m^T \Delta U_n\|_F$, $m,n = \{1,2\},$ are reasonably tight.}
As no matrix decompositions are needed to compute our bounds, we have efficiently tracked the tangent space recovery error.  The dashed vertical lines in Figure \ref{fig:results} indicate the locations of the minima of the true error curve (dashed blue) and the Main Result bound (dashed green).  In general, we see agreement of the locations at which the minima occur, indicating the scale that will yield the optimal tangent space approximation.  The minimum of the Main Result bound falls within a range of scales at which the true recovery error is stable.
In particular, we note that when the location of the bound's minimum does not correspond with the minimum of the true error (such as in panel (d)), the discrepancy occurs at a range of scales for which the true error is quite flat.  In fact, in panel (d), the difference between the error at the computed optimal scale and the error at the true optimal scale is on the order of $10^{-2}$.  Thus the angle between the computed and true tangent spaces will be less than half of a degree and the computed tangent space is stable in this range of scales.  For a large data set it is impractical to examine every scale and one would instead most likely use a coarse sampling of scales.  The true optimal scale would almost surely be missed by such a coarse sampling scheme.  Our analysis indicates that despite missing the exact true optimum, we may recover a scale that yields an approximation to within a fraction of a degree of the optimum.

\begin{figure}
  \centering
  \includegraphics[width=70mm]{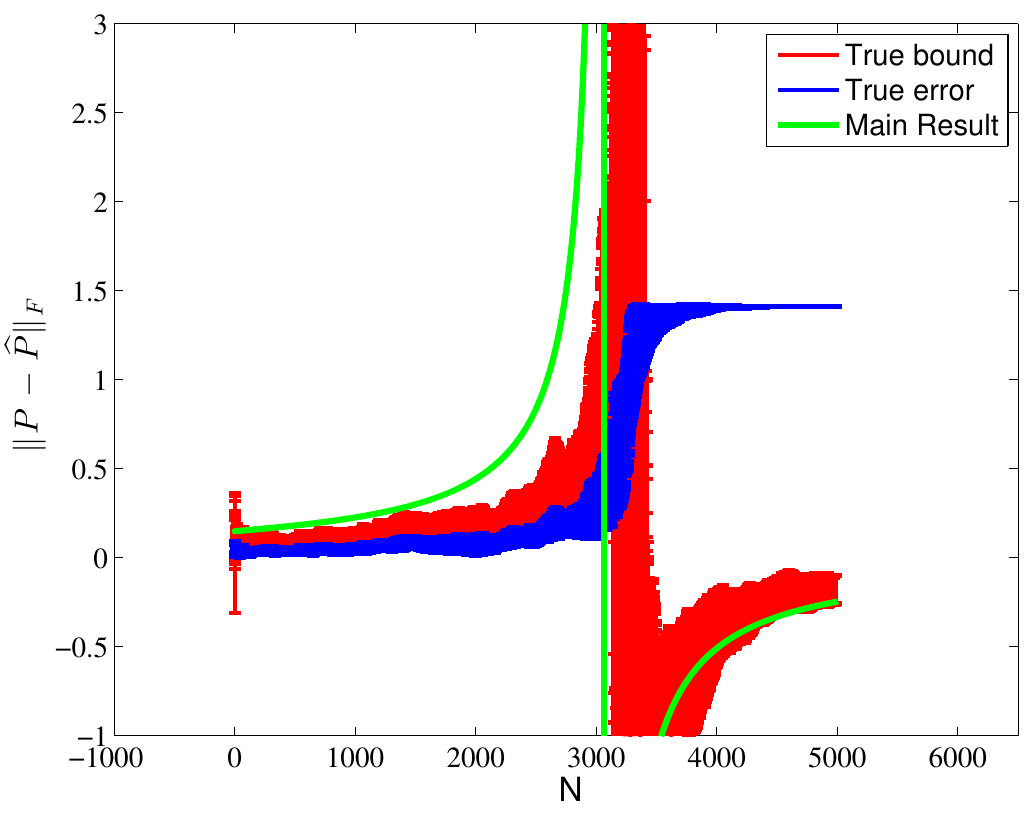}
  \caption{Bounds for a 2-dimensional saddle (noise free) with $\kappa_1^{(3)}=3$ and $\kappa_2^{(3)}=-3$.}
  \label{fig:saddle}
\end{figure}

\subsection{Sensitivity to Error in Parameters}

As is often the case in practice, parameters such as intrinsic dimension, curvature, and noise level are unknown and must be estimated from the data.  It is therefore important to experimentally test the sensitivity of tangent space recovery to errors in parameter estimation.  In the following experiments, we test the sensitivity to each parameter by tracking the optimal scale as one parameter is varied with the others held fixed at their true values.  For consistency across experiments, the optimal scale is reported in terms of neighborhood radius and denoted by $r^*$.  The relationship between neighborhood radius $r$ and number of sample points $N$ is defined by equation \eqref{eq:sampling-density}.  In all experiments, we generate data sets sampled from a 4-dimensional manifold embedded in $\R^{10}$ according to the local model \eqref{eq:local-model}.
\begin{figure}
  \centering
  
  \subfigure[$\kappa_j^{(i)}=2$, $\sigma=0.01$]{\includegraphics[width=.45\textwidth]{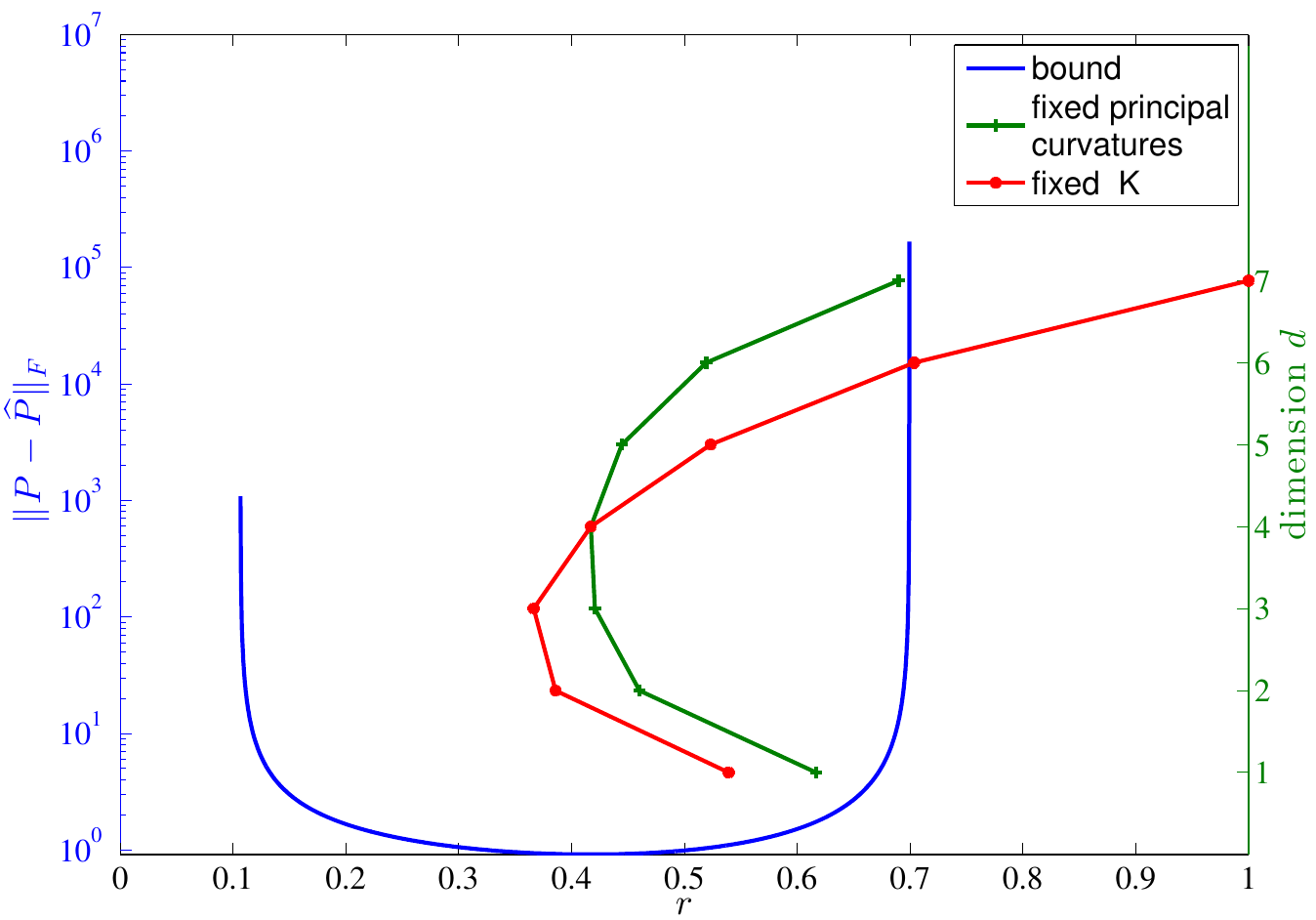}}
  \subfigure[$\kappa_j^{(i)}=3$, $\sigma=0.01$]{\includegraphics[width=.45\textwidth]{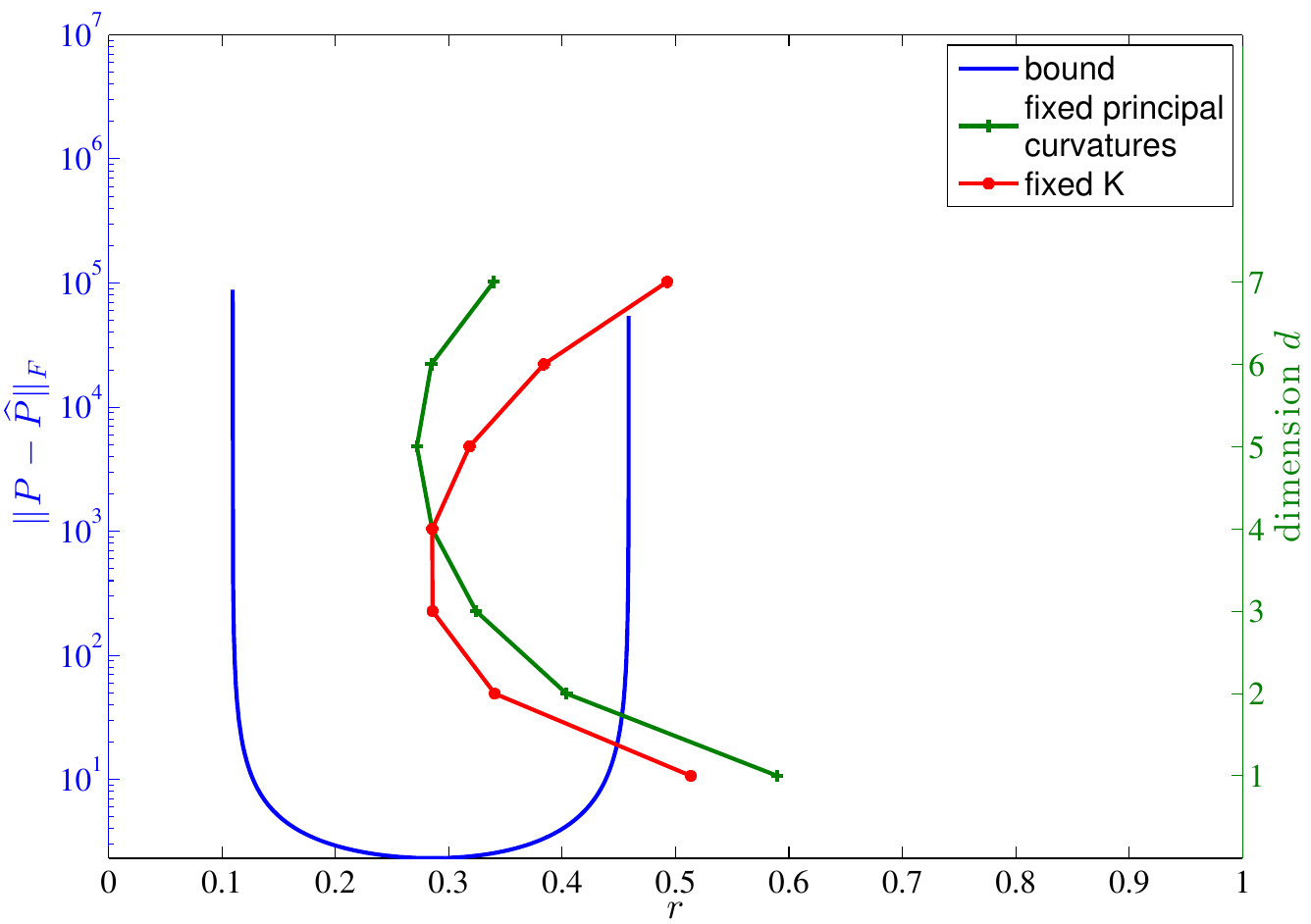}}
  \caption{The optimal radius is shown to be sensitive to error in estimates of $d$.  The Main Result bound (blue) tracks the subspace recovery error (left ordinate).  The green and red curves show the computed optimal radii for varying $d$ (right ordinate) with fixed $\kappa_j^{(i)}$ and fixed $K$, respectively.  See text for details.}
  \label{fig:test_d}
  
  \subfigure[$\kappa_j^{(i)}=1.5$, $\sigma=0.025$]{\includegraphics[width=.45\textwidth]{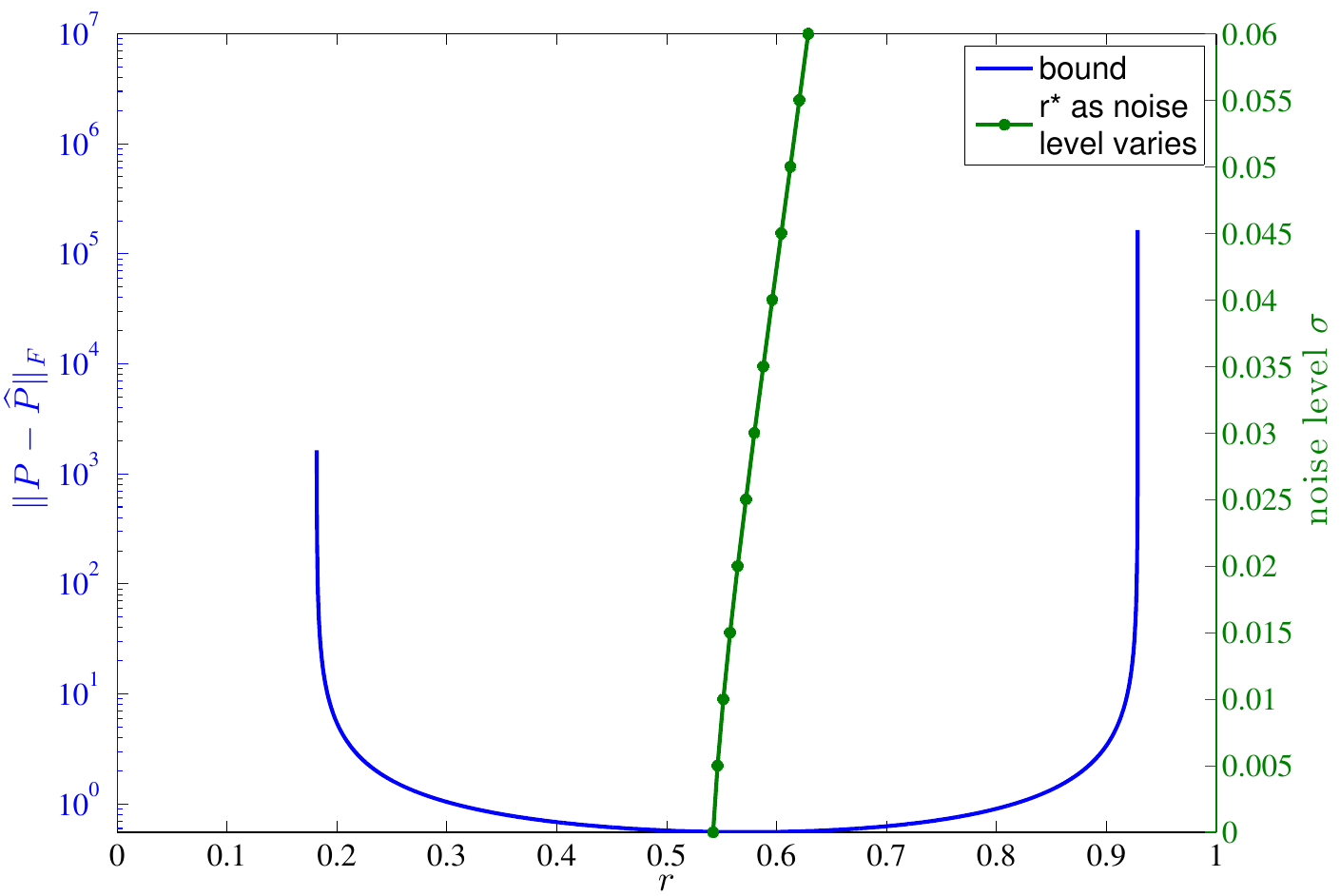}}
  \subfigure[$\kappa_j^{(i)}=2$, $\sigma=0.05$]{\includegraphics[width=.45\textwidth]{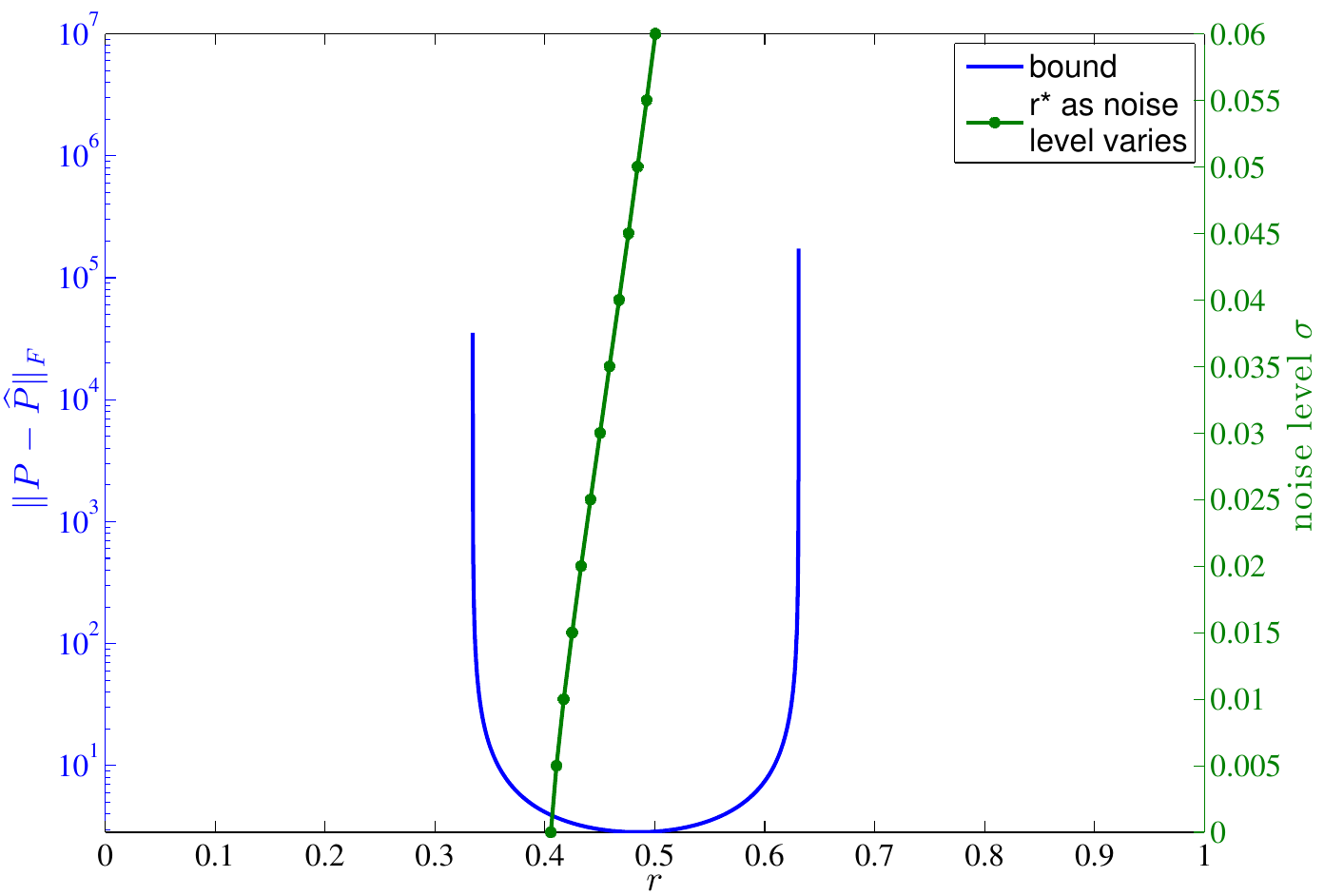}}
  \caption{The sensitivity to error in estimates of $\sigma$ is shown to be mild.  The Main Result bound (blue) tracks the subspace recovery error (left ordinate) and the optimal radius is computed (green) for varying values of $\sigma$ (right ordinate).  See text for details.}
  \label{fig:test_sigma}
  
  \subfigure[$K=12.25$, $\sigma=0.01$]{\includegraphics[width=.45\textwidth]{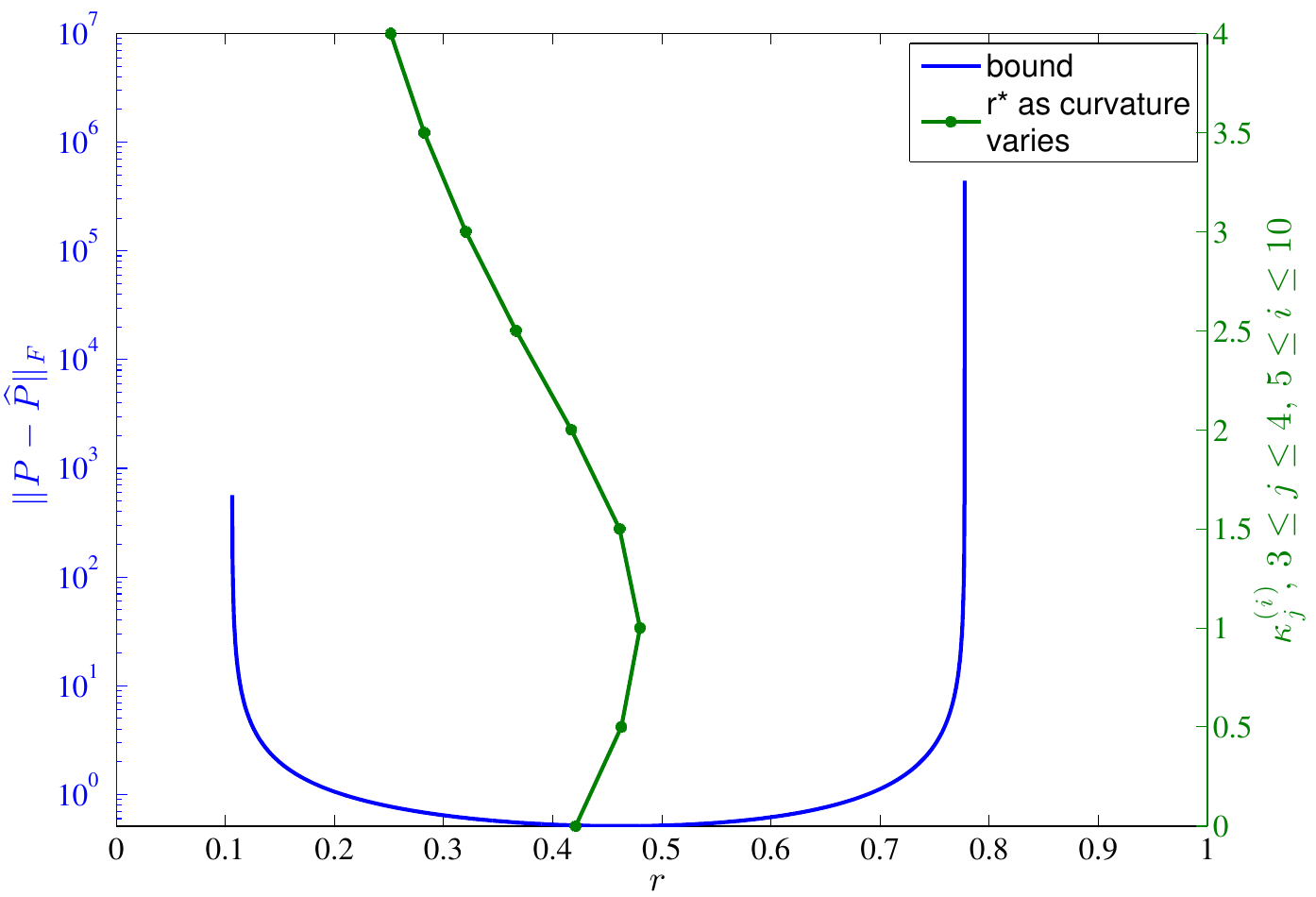}}
  \subfigure[$K=19.6$, $\sigma=0.025$]{\includegraphics[width=.45\textwidth]{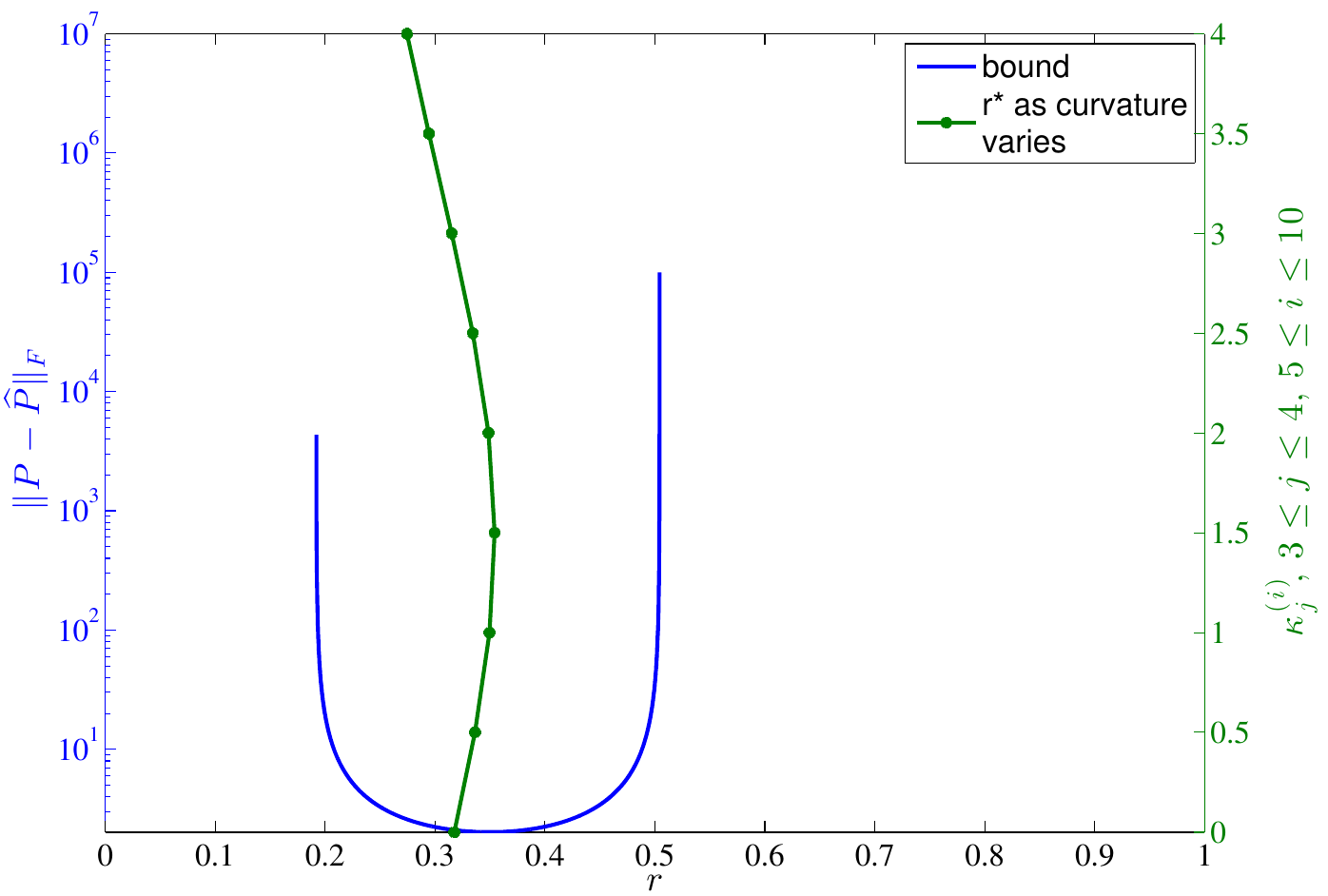}}
  \caption{The sensitivity to error in estimates of curvature is shown to be mild.  The Main Result bound (blue) tracks the subspace recovery error and the optimal radius is computed (green) for varying values of $\kappa_3^{(i)}$ and $\kappa_4^{(i)}$ with $\kappa_1^{(i)}$ and $\kappa_2^{(i)}$ held fixed.  See text for details.}
  \label{fig:test_k}
  
\end{figure}

Figure \ref{fig:test_d} shows that the optimal scale $r^*$ is sensitive to errors in the intrinsic dimension $d$.  A data set is sampled from a noisy, bowl-shaped manifold with equal principal curvatures in all directions.  We set the noise level $\sigma = 0.01$ and the principal curvatures $\kappa_j^{(i)} = 2$ in panel (a) and $\kappa_j^{(i)} = 3$ in panel (b).  Noting that the true intrinsic dimension is $d=4$, we test the sensitivity of $r^*$ as $d$ is varied.
There are three axes in each panel of Figure \ref{fig:test_d}: the neighborhood radius $r$ on the abscissa; the angle $\|P-\widehat{P}\|_F$ on the left ordinate; and the values used for the dimension $d$ on the right ordinate.
Our Main Result bound is shown in blue and tracks the subspace recovery error (angle, on the left ordinate) as a function of neighborhood radius $r$ for the true values of $d$, $\sigma$ and $\kappa_j^{(i)}$.  Holding the noise and curvature fixed, we then compute $r^*$ using incorrect values for $d$ ranging from $d=1$ to $d=7$.  The green and red curves show the computed $r^*$ for each value of $d$ (on the right ordinate) according to the two ways to fix curvature while varying $d$:
(1) hold the value of each $\kappa_j^{(i)}$ fixed, thereby allowing $K$ to change with $d$ (shown in green); or (2) hold $K$ fixed, necessitating that the $\kappa_j^{(i)}$ change with $d$ (shown in red).
The Main Result bound (blue) indicates an optimal radius of $r^* \approx 0.45$ in (a) and $r^* \approx 0.30$ in (b).  However, the $r^*$ computed using inaccurate estimates of $d$ show great variation, ranging between a radius close to the optimum and a radius close to the size of the entire manifold.  These experimental results indicate the importance of properly estimating the intrinsic dimension of the data.

Next, the sensitivity to error in the estimated noise level is shown to be mild in Figure \ref{fig:test_sigma}.  A data set is sampled from a noisy, bowl-shaped manifold with equal principal curvatures in all directions.  The true values for the parameters are: $d=4$, $\kappa_j^{(i)} = 1.5$, and $\sigma = 0.025$ in \ref{fig:test_sigma}a; and $d=4$, $\kappa_j^{(i)} = 2$, and $\sigma = 0.05$ in \ref{fig:test_sigma}b.  Our Main Result bound (blue) tracks the subspace recovery error (left ordinate) as a function of $r$ (abscissa) using the true parameter values and indicates an optimal radius of $r^* \approx 0.55$ and $r^* \approx 0.5$ for (a) and (b), respectively.  Holding the dimension and curvature constant, we then compute $r^*$ using incorrect values for $\sigma$ ranging from $\sigma = 0$ to $\sigma = 0.06$.  The green curve shows the computed $r^*$ for each value of $\sigma$ (on the right ordinate).
In both \ref{fig:test_sigma}a and \ref{fig:test_sigma}b, the computed $r^*$ remain close to the optimum as the noise level varies and are within the range of radii where the recovery is stable (as indicated by the Main Result bound in blue).
This behavior is in agreement with our experimental observations (not shown) indicating that increasing the noise level reduces the range for stable recovery but leaves the minimum of the Main Result bound relatively unaltered.
We note that the range for stable recovery is smaller in (b) as is expected in the higher curvature and noise setting.

Finally, Figure \ref{fig:test_k} shows mild sensitivity to error in estimated curvature.  A data set is sampled from a noisy manifold with two large principal curvatures ($\kappa_1^{(i)}$ and $\kappa_2^{(i)}$) and two small principal curvatures ($\kappa_3^{(i)}$ and $\kappa_4^{(i)}$) in each \textcolor{black}{normal direction} $i$.  This tube-like geometry provides more insight for sensitivity to error in curvature by avoiding the more stable case where all principal curvatures are equal.  The true values for the parameters are: $d=4$, $\sigma = 0.01$, $\kappa_1^{(i)} = \kappa_2^{(i)} = 2$, and $\kappa_3^{(i)} = \kappa_4^{(i)} = 0.5$ for $5 \leq i \leq 10$ in (a); and $d=4$, $\sigma = 0.025$, $\kappa_1^{(i)} = \kappa_2^{(i)} = 3$, and $\kappa_3^{(i)} = \kappa_4^{(i)} = 1$ for $5 \leq i \leq 10$ in (b).
Our Main Result bound (blue) tracks the subspace recovery error (left ordinate) as a function of $r$ (abscissa) using the true parameter values and indicates an optimal radius of $r^* \approx 0.45$ and $r^* \approx 0.35$ for (a) and (b), respectively. 
Holding the dimension, noise level, and large principal curvatures $\kappa_1^{(i)}$ and $\kappa_2^{(i)}$ constant, we then compute the $r^*$ using incorrect values for the smaller principal curvatures $\kappa_3^{(i)}$ and $\kappa_4^{(i)}$, $5 \leq i \leq 10$.  The green curve shows the $r^*$ computed for values of $\kappa_3^{(i)} = \kappa_4^{(i)}$ indicated on the right ordinate, $5 \leq i \leq 10$.  The computed $r^*$ remain within the range of radii where the recovery is stable (as indicated by the Main Result bound in blue) in both (a) and (b).  We observe less variation in the higher curvature and higher noise case shown in \ref{fig:test_k}b.  In this case, the larger principal curvatures anchor the bound, leaving $r^*$ less sensitive to error in the estimated smaller principal curvatures.  \textcolor{black}{As can be expected, experimental results (not shown) indicate that $r^*$ is sensitive to perturbations of these anchoring, large principal curvatures}.

\color{black}
\section{Practical Application \& Experimental Results II} \label{sec:numerical2}

With the purpose of providing perturbation bounds that can be used in
practice, we provide in this section the algorithmic tools that make
it possible to directly apply the theoretical results of Section
\ref{sec:mainresult} to a real dataset.

  The first tool is a ``translation rule'' to compare distances measured
  in the tangent plane $\TP$ and distances in $\R^D$: given a point
  $x$ at a distance $R$ from the origin, we provide an estimate,
  $\hat{r}(R)$, of the distance of the projection of $x$ in $\TP$ to
  the origin $x_0$.
  The second tool is a plug-in method to compute a ``clean
  estimate'', $\hat{x}_0$, of the point $x_0$ on $\mathcal{M}$ that
  serves as the origin of the coordinate system in our analysis.
  Equipped with these two tools, the practitioner can compute the
  perturbation bound as a function of the radius $R$ measured from
  $\hat{x}_0$ in the ambient space $\R^D$.
  
\subsection{Effective Distance in the Tangent Plane $\TP$}
Our Main Result, Theorem \ref{thm:mainresult1}, is presented in terms
of the radius $r$ corresponding to the distance from the origin $x_0$ of a
point's noise-free tangential component.  Because $r$ cannot be
observed in practice, we provide an estimate $\hat{r}(R)$ of $r$ for
any point $x$ at distance $R$ from the origin. In the presentation
that follows, we assume oracle knowledge of the local origin $x_0 \in
\mathcal{M}$; recovery of this origin is addressed in the next
section.  

As previously introduced, a point $x$ in a neighborhood of $x_0$ may
be decomposed as $x = x_0 + \ell + c + e$ and we recognize $r^2 =
\|\ell\|^2$. To explore the relationship between $r$ and $R$, we
compute
\begin{equation}
\label{eq:R^2}
R^2 = \|x-x_0\|^2 = \|x_0 + \ell + c + e - x_0\|^2 = r^2 + \|c\|^2 +
\|e\|^2 + 2\langle \ell + c, e\rangle, 
\end{equation}
where we use that $\langle \ell,c \rangle = 0$.
The terms on the right hand side depend on the realizations of the
sample point $x$ and noise $e$.  To understand their sizes, we compute in
expectation, 
\begin{equation*}
\mathbb{E}[\|c\|^2] = \gamma r^4, \quad \mathbb{E}[\|e\|^2] = \sigma^2
D, \quad \text{and} \quad \mathbb{E}[\langle \ell + c, e\rangle] = 0, 
\end{equation*}
where
\begin{equation}
\gamma = \frac{\sum_{i=d+1}^D 3K_{nn}^{ii} + K_{mn}^{ii}}{2(d+2)(d+4)}.
\label{eq:gamma-expectation}
\end{equation}
Injecting these terms into \eqref{eq:R^2}, we solve for positive and
real $r$ and arrive at an approximation $\hat{r} (R) $ of the (tangent plane) radius
$r$ given the observable (ambient) radius $R$:
\begin{equation}
\hat{r} (R) = \sqrt{\frac{1}{2\gamma}\left(-1+\sqrt{1+4\gamma (R^2-\sigma^2 D)} \right)}.
\label{eq:solve4r}
\end{equation}
\begin{remark}
  Another approach to determine the relationship between $r$ and $R$
  proceeds as follows. We calculate the volume of the $d$-dimensional
  ball $B^d_{x_0}(r)$ given by the pre-image of the points in the ball
  $B^D_{x_0}(R)$ of radius $R$ in $\R^D$, and use this volume to
  derive an effective radius $r$.

In the noise-free case, we can get some insight into this problem
using a result from Gray \cite{Gray74} that gives the volume of a
geodesic ball $B^{\mathcal{M}}_{x_0}(\omega)$ on $\mathcal{M}$ centered at
$x_0$ as a function of the radius $\omega$ measured along the
manifold. We have
\begin{equation*}
V(B^{\mathcal{M}}_{x_0}(\omega)) = \omega^d v_d
\left(1 
- \frac{S(x_0)}{6(n+2)} \omega^2
+ o(\omega^2)
\right),
\end{equation*}
where $S(x_0)$ is the scalar curvature of the manifold at $x_0$ and
$v_d$ is the volume of the $d$-dimensional unit ball. Let $r$ be
the radius of the smallest ball that encloses the pre-image of
$B^{\mathcal{M}}_{x_0}(\omega)$ in the tangent plane $\TP$,
\begin{equation*}
\forall \ell = (\ell_1,\ldots,\ell_d) \in B^d_{x_0}(r), \quad
x = \begin{bmatrix}
\ell_1 & \cdots & \ell_d & f_{d+1}(\ell)& f_D(\ell)
\end{bmatrix}
\in B^{\mathcal{M}}_{x_0}(\omega).
\end{equation*}
In our coordinate system, $B^d_{x_0}(r)$ is the smallest ball that
encloses the projection of $B^{\mathcal{M}}_{x_0}(\omega)$ in the
tangent plane $\TP$, and therefore the volume of $B^d_{x_0}(r)$ is
smaller than the volume of $B^{\mathcal{M}}_{x_0}(\omega)$.  Finally,
we note that $V(B^{\mathcal{M}}_{x_0}(\omega))$ corresponds to the
volume of an ``effective ball'' in $\R^d$ of radius $r_{eff}$,
\begin{equation}
r_{eff} = 
\omega \left(1 
- \frac{S(x_0)}{6(d+2)} \omega^2
+ o(\omega^2)
\right)^{1/d}.
\end{equation}
Because $V(B^d_{x_0}(r)) \leq V(B^{\mathcal{M}}_{x_0}(\omega)) =
V(B^d_{x_0}(r_{eff}))$, we have $r \leq r_{eff}$. We note that if $\omega$ is
small, we can approximate the chordal distance $R$ with
the geodesic distance $\omega$. If we use $r_{eff}$ as an estimate for $r$,
we obtain
\begin{equation}
r \approx
R \left(1 
- \frac{S(x_0)}{6(d+2)} R^2
\right)^{1/d}
\approx 
R \left(1 
- \frac{S(x_0)}{6d(d+2)} R^2
\right).
\end{equation}
The computation of the sectional curvature in our coordinate
system yields the following expression,
\begin{equation}
S(x_0) = 
\sum_{\substack{m,n=1 \\ m \neq n}}^d ~ \sum_{i=d+1}^D \kappa^{(i)}_m\kappa^{(i)}_n
= \sum_{i=d+1}^D K^{ii}_{mn},
\end{equation}
using the notation defined in (\ref{sectional-curvature}). We finally obtain
the following estimate of $r$,
\begin{equation}
R \left(1 
- \frac{\sum_{i=d+1}^D K^{ii}_{mn}}{6d(d+2)} R^2
\right).
\end{equation}
In comparison, the estimate $\hat{r}(R)$ given by (\ref{eq:solve4r})
is approximately equal to 
\begin{equation}
R\left(1 - \frac{\sum_{i=d+1}^D 3K_{nn}^{ii} +
    K_{mn}^{ii}}{4(d+2)(d+4)} R^2 \right),
\end{equation}
for small values of $R$. The two estimates, which capture the effect
of curvature on the relationship between $r$ and $R$, are indeed very
similar, confirming the general form of the approximation given by (\ref{eq:solve4r}).
\end{remark}
\begin{remark}
  In a manner similar to the previous derivation, we can estimate the
  effect of the noise on the volume a ball $B^D_{x_0}(R)$ of noisy
  samples centered around $x_0$. We define the normal space $\NS$ to
  be the orthogonal complement of $\TP$ in $\R^D$. When $D$ is
  sufficiently large, we expect that the Gaussian noise will
  concentrate on the surface of a sphere of radius $\sigma
  \sqrt{D}$. The probability density function of the noisy samples is
  given by the convolution of the uniform distribution on the manifold
  (seen as a distribution in $\R^D$ localized on $\mathcal{M}$) with
  the Gaussian kernel.  If the manifold is flat, the probability
  density function of the noisy samples points $X$ becomes uniform in
  the tube
\begin{equation}
\mathcal{M}_{\sigma} = \left\{x = y + u, y \in \mathcal{M}, u \in \NS,
  \|u\| \leq \sigma \sqrt{D} \right\}.
\end{equation}
Because the noisy points are spread uniformly in
$\mathcal{M}_{\sigma}$, the measure of the set of noisy points in the
ball centered at $x_0$ of radius $R$, $B^D_{x_0}(R)$, is given by
\begin{equation}
\frac{V_D(B^D_{x_0}(R)  \cap \mathcal{M}_\sigma)}{(2 \sigma\sqrt{D})^{D-d}},
\end{equation}
where the factor $1/(2 \sigma\sqrt{D})^{D-d}$ accounts for the uniform
distribution of the noisy points in $\mathcal{M}_{\sigma}$ along the
direction $\NS$. We can approximate the set $B^D_{x_0}(R) \cap
\mathcal{M}_\sigma$ by a smaller enclosed cylinder
\begin{equation*}
B^d_{x_0}(\sqrt{R^2-d \sigma^2 D}) \oplus[-\sigma\sqrt{D}, \sigma\sqrt{D}]^{D-d}
\end{equation*}
as soon as the radius $R$ extends beyond the tube $\mathcal{M}_\sigma$
in the direction $\NS$. This yields the following estimate
for the volume of $V_D(B^D_{x_0}(R) \cap \mathcal{M}_\sigma)$,
\begin{equation}
v_d (R^2-d \sigma^2 D)^{d/2} (2 \sigma \sqrt{D})^{D-d}.
\end{equation}
We conclude that the set of noisy point in $B^D_{x_0}(R)$ has a measure
given by
\begin{equation}
v_d (R^2-d \sigma^2 D)^{d/2} = v_d \left[R\sqrt{1 - \frac{d\sigma^2 D}{R^2}}\;\right]^d
\end{equation}
This measure corresponds to an effective radius $r$ in the tangent
plane given by
\begin{equation}
r = R\sqrt{1 - \frac{d\sigma^2 D}{R^2}}.
\label{flat-noise}
\end{equation}
Because we compute a lower bound on the measure of the set of noisy
points in $B^D_{x_0}(R)$, the effective radius (\ref{flat-noise})
introduces a correction $d \sigma^2 D$ to $R^2$ that is $d$ times
larger than the correction obtained in (\ref{eq:solve4r}),
$\sigma^2 D$. While a more precise computation of $V_D(B^D_{x_0}(R) \cap
\mathcal{M}_\sigma)$ can remove the dependency on the dimension $d$,
this computation confirms that the effect of noise can be accounted
for by a simple subtraction of a term of the form $\sigma^2 D$ from
$R^2$, as indicated in the less formal calculation that leads to
(\ref{eq:solve4r}).

The same line of argument can be followed when the manifold is not
flat. The authors in \cite{Genovese12} prove that when the noise is
uniformly distributed along the normal fibers, then the probability
distribution of the noisy points is still approximately uniform. The
authors in \cite{Genovese12} bound the departure from the uniform
distribution using geometric constants analogous to $\gamma$ or the
scalar curvature $S$. Because the Gaussian will lead to a uniform
distribution in the tube $\mathcal{M}_\sigma$, quantitatively similar
result can be obtained when the noise a Gaussian, as confirmed by the
thorough analysis performed in \cite{Maggioni-MIT,LittleThesisDuke}.
While a more accurate estimate of $r$, which would account for
curvature and noise, could be obtained using this route, our
experiments in the next section indicate that the rough approximation
provided by (\ref{eq:solve4r}) accurately tracks the true $r$.
\end{remark}

\color{black}
Let us examine the quality of the approximation of $r$ given by $\hat{r}(R)$ in \eqref{eq:solve4r} using the two data sets from Section \ref{sec:numerical} that correspond to Figures \ref{fig:results}(c) and \ref{fig:results}(d).  The first data set consists of noisy ($\sigma = 0.01$) points sampled from a 3-dimensional manifold embedded in $\mathbb{R}^{20}$, where the principal curvatures of the manifold are equal in all normal directions (``bowl geometry'').  The second data set consists of noisy ($\sigma = 0.01$) points sampled from a 3-dimensional manifold embedded in $\mathbb{R}^{20}$ where the principal curvatures (given in Table \ref{tab:kk2}) are such that three of the normal directions exhibit significantly greater curvature than the others (``tube geometry'').  Figure \ref{fig:radii_R} shows the radius $r$ measured in the tangent plane (blue) and its estimate $\hat{r}(R)$ (black) given by \eqref{eq:solve4r}.  The radius $R$ measured in the ambient space, from which the estimate $\hat{r}(R)$ is computed, is shown in green for reference.  The bowl geometry is shown in Figure \ref{fig:radii_R}(a) and the tube geometry is shown in Figure \ref{fig:radii_R}(b).  We see that for both geometries, $r$ and $\hat{r}(R)$ are nearly indistinguishable over all relevant scales (the disagreement at the largest scales for the tube geometry occurs well after the computed tangent plane becomes orthogonal to the true tangent plane).  The results shown in this figure indicate that $\hat{r}(R)$ can be used to reliably estimate $r$ from the observed $R$ and, therefore, to compute the Main Result bound \eqref{eq:main-result1} from quantities that are observable in practice.
\begin{figure}
  \centering
  \subfigure[bowl geometry]{\includegraphics[width=70mm]{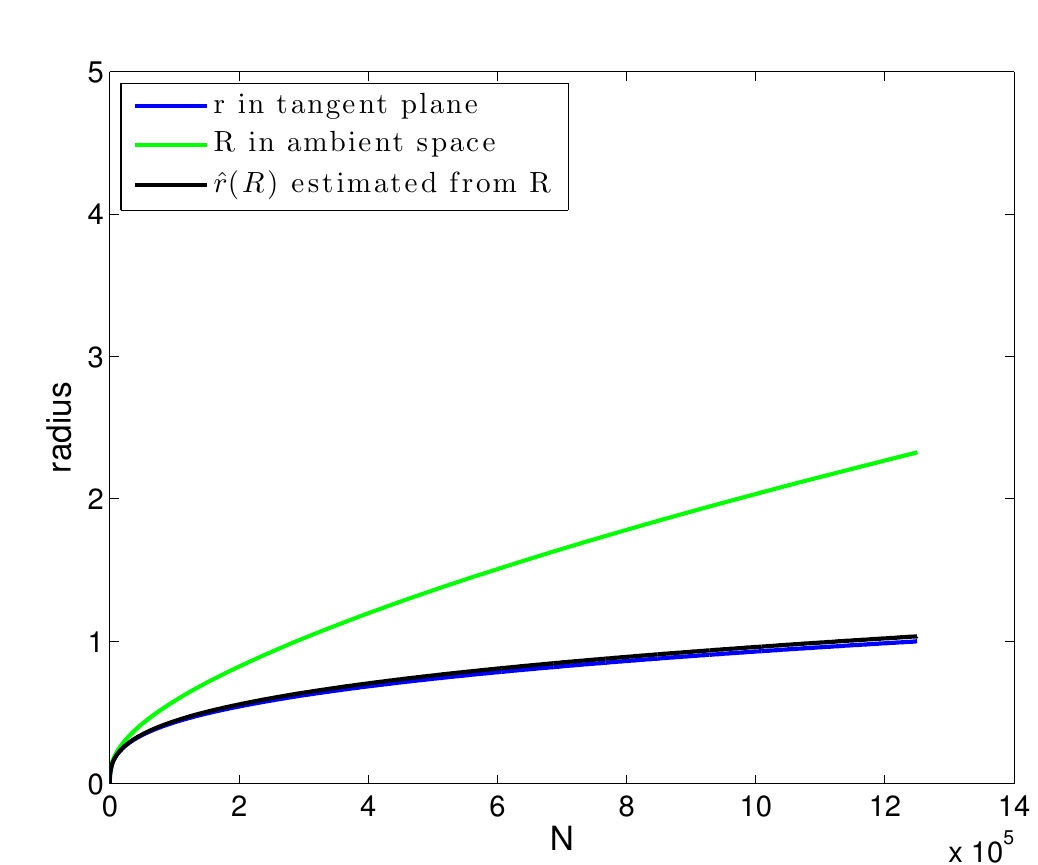}}
  \subfigure[tube geometry]{\includegraphics[width=70mm]{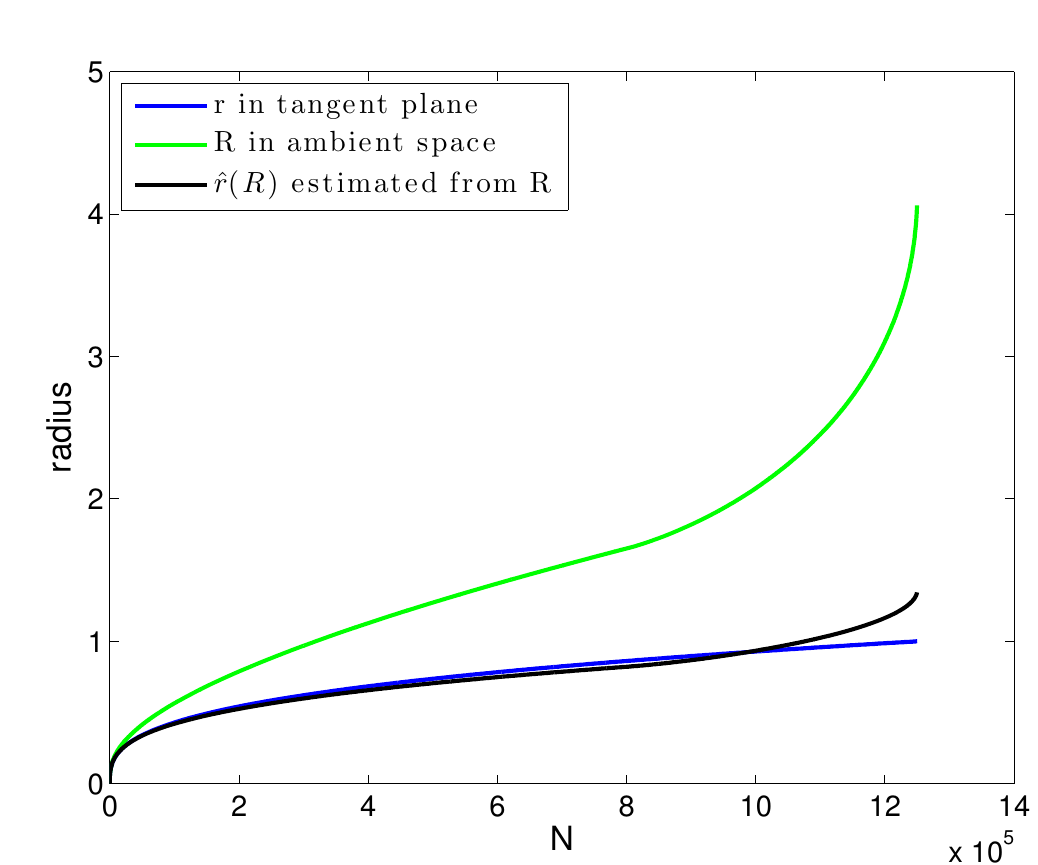}}
  \caption{\textcolor{black}{The tangent plane radius $r$ (blue) and its approximation $\hat{r}(R)$ (black) given by equation \eqref{eq:solve4r} are shown to be indistinguishable over all relevant scales for two different geometries. The ambient radius $R$ from which the estimate $\hat{r}(R)$ is computed is shown in green.  See text for discussion.}}
  \label{fig:radii_R}
\end{figure}

\subsection{Subspace Tracking and Recovery using the Ambient Radius}

We now repeat the experiments of Section \ref{subsec:numerical-results} by recomputing the subspace recovery error and subspace recovery bounds using the radius in the ambient space, $R$, in place of the tangent plane radius, $r$.  We demonstrate that, after converting the ambient radius $R$ to its corresponding tangent plane radius $\hat{r}(R)$, the bound presented in the Main Result Theorem \ref{thm:mainresult1} accurately tracks the subspace recovery error.
The presented results demonstrate that the Main Result may be used for tangent space recovery in the practical setting where only the ambient radius is available.
\begin{figure}
  \centering
  \subfigure[]{\includegraphics[width=70mm]{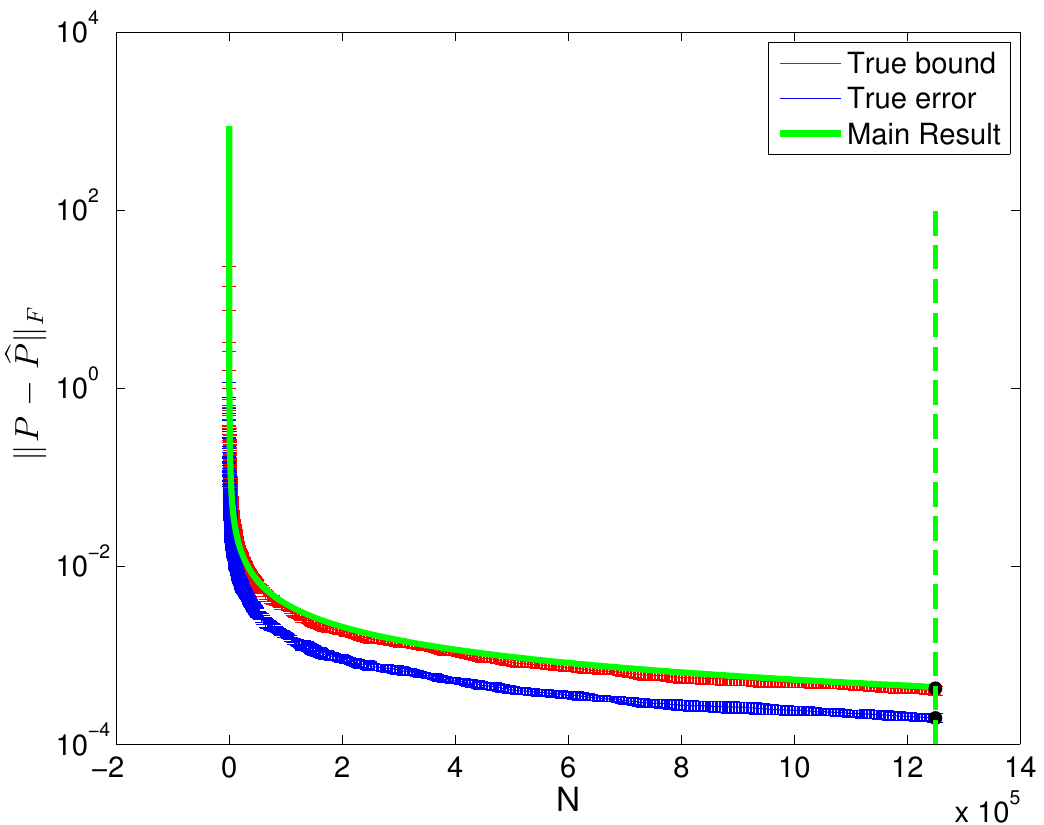}}
  \subfigure[]{\includegraphics[width=70mm]{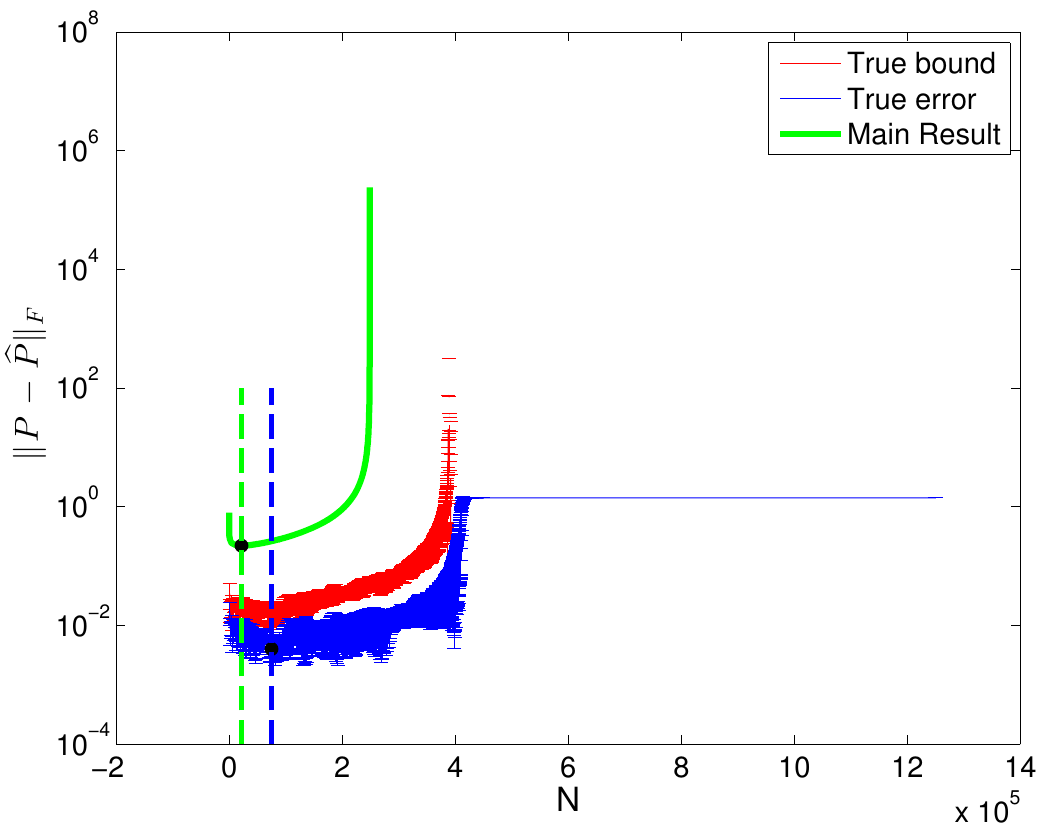}}
  \subfigure[]{\includegraphics[width=70mm]{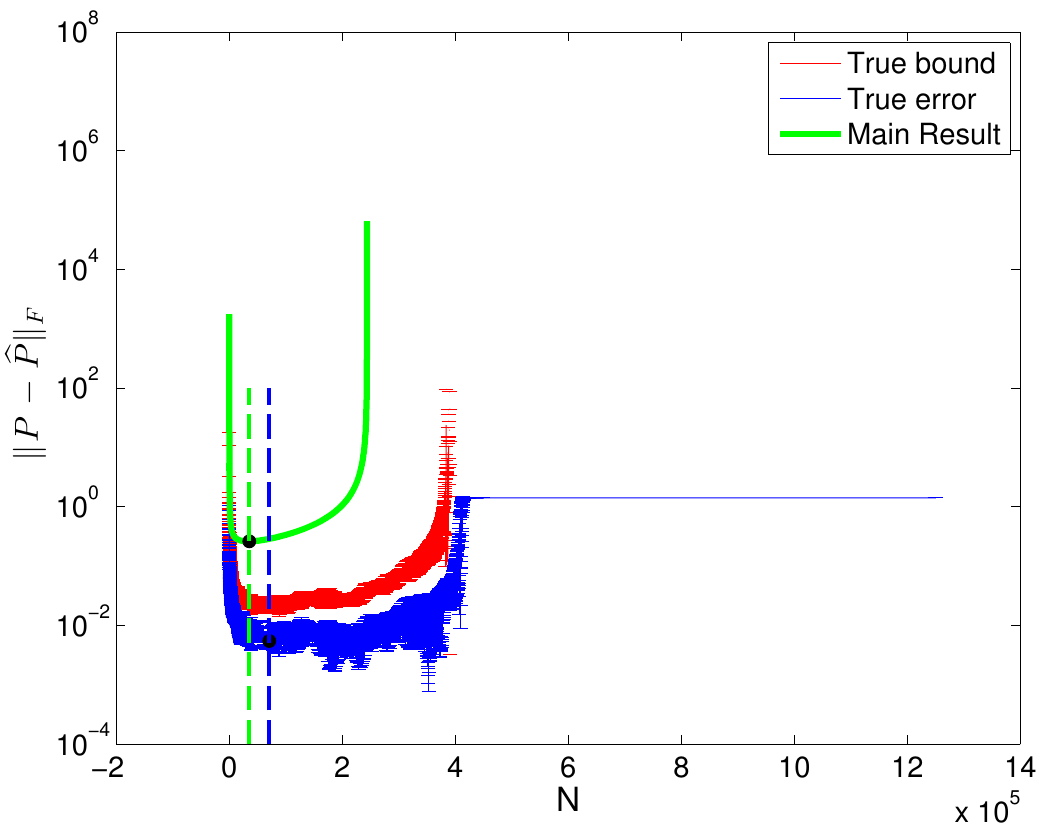}}
  \subfigure[]{\includegraphics[width=70mm]{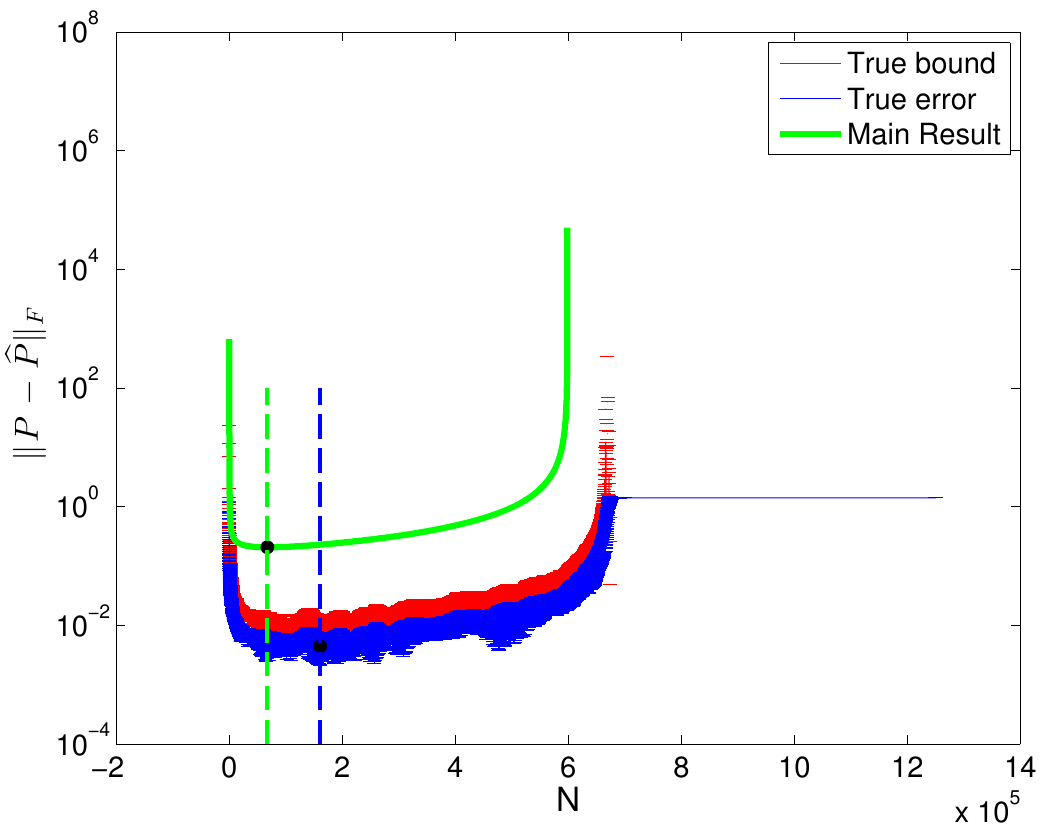}}
  \caption{\textcolor{black}{Norm of the perturbation using the ambient radius $R$: (a) flat manifold with noise, (b) curved (tube-like) manifold with no noise, (c) curved (tube-like) manifold with noise, (d) curved (bowl-like) manifold with noise.  Dashed vertical lines indicate minima of the curves.  Note the logarithmic scale on the Y-axes.  Compare with Figure \ref{fig:results} and see text for discussion.}}
  \label{fig:experiments_R}
\end{figure}

We begin by generating 3-dimensional data sets embedded in $\mathbb{R}^{20}$ according to the specifications given in Section \ref{subsec:numerical-results}.
The curvature is chosen such that $K = 12.6025$ for all manifolds (excluding the linear subspace example).
The tube geometry is implemented by choosing principal curvatures as given in Table \ref{tab:kk2} and the bowl geometry has all principal curvatures set to 1.0189.  All but the noise-free data set have Gaussian noise added with standard deviation $\sigma = 0.01$.

For each experiment, the ambient radius $R$ is measured from the data and used to approximate the corresponding tangent plane radius $\hat{r}(R)$ by equation \eqref{eq:solve4r}, from which we compute our bound \eqref{eq:main-result1}.  We then compare this bound with the true subspace recovery error.  Mimicking the experiments of Section \ref{subsec:numerical-results}, the tangent plane at the local origin $x_0$ is computed at each scale $N$ via PCA of the $N$ nearest neighbors of $x_0$, where the distance from $x_0$ (the radius $R$) is now measured in the ambient space $\mathbb{R}^D$.  The true subspace recovery error $\|P-\widehat{P}\|_F$ is then computed at each scale.  The ``true bound'' is again computed by applying Theorem \ref{thm:Stewart} after measuring each perturbation norm directly from the data.  We recall that this ``true bound'' requires no SVDs and utilizes information that is not practically available to represent the best possible bound that we can hope to achieve.  We will compare the mean of the true error and mean of the true bound over 10 trials (with error bars indicating one standard deviation) to the bound given by our Main Result in Theorem \ref{thm:mainresult1}, holding with probability greater than 0.8.  We note that for these experiments, the local origin $x_0$ is given by oracle information and we will consider its recovery in a separate set of experiments.

The results are shown in Figure \ref{fig:experiments_R} and should be compared with those shown in Figure \ref{fig:results}.  Panel (a) shows the noisy curvature-free (linear subspace) result and we observe that the behaviors of the true error (blue), true bound (red), and main result bound (green) match the behaviors of their counterparts in Figure \ref{fig:results}(a) that were computed using $r$.  In particular, the error in Figure \ref{fig:experiments_R}(a) decays as $1/\sqrt{N}$ (note the logarithmic scale of the Y-axis).  Our bound (green) accurately tracks the true error (blue) and is nearly indistinguishable from the true bound (red).
Panel (b) shows the result for a noise-free manifold with tube geometry such that three of the normal directions exhibit high curvature while the others are flatter.
We see that, much like in Figure \ref{fig:results}(b), the main result bound (green) increases monotonically (ignoring the slight numerical instability at extremely small scales) to match the general behavior of the true error (blue) and true bound (red). 
Panel (c) shows the results for the noisy version of the manifold used in panel (b).  We observe that our bound (green) now exhibits blow up at small scales due to noise and blow up at large scales due to curvature, matching the behavior of the true error.
Finally, panel (d) shows the results for the noisy manifold with bowl geometry where all principal curvatures are equal, and indicates that our bound tracks the error at all scales.
The dashed vertical lines in Figure \ref{fig:experiments_R} indicate the locations of the minima of the true error curves (dashed blue) and the Main Result bounds (dashed green).  We see that the location of the minimum of the Main Result bound is, in general, close to the minimum of the true error curve and falls within a range of scales for which the error is quite flat.  

We observe that the results using $R$ in Figure \ref{fig:experiments_R} are similar to those seen in Figure \ref{fig:results} using $r$, while noting that the true error for the tube geometry remains stable at larger scales in Figure \ref{fig:experiments_R} than the true error in Figure \ref{fig:results}.
To understand this observation, we examine the effect of geometry on the radii $R$ and $r$.
Figure \ref{fig:radii_R_compare} shows the radius $R$ in green for the bowl geometry (left) and for the tube geometry (right).  This radius corresponds to the norm of each point $x$ collected as a ball is grown in the ambient space.
Shown in red is the ambient radius of each point $x$ collected as the tangent plane radius, $r$, is grown.
This curve corresponds to the collection of points according to the norm of their tangential projection.
Figure \ref{fig:radii_R_compare} shows that these radii exhibit different behaviors depending on the geometry of the manifold.
When all principal curvatures are equal (bowl geometry), each normal direction exerts the same amount of influence on a point's norm and curvature does not impact the order in which the points are discovered.
Thus, the radii are shown to be identical for the bowl geometry in Figure \ref{fig:radii_R_compare}(a), with the green curve sitting exactly on top of the red curve.
However, the tube geometry allows for curvature in certain normal directions to exert more influence on the norm than others.  In this situation, growing a ball in the ambient space will necessarily discover points exhibiting greater curvature at the larger scales.  In contrast, the ball grown in the tangent space may discover such points at much smaller scales, as the radius measures the norm of only the tangential components.
Thus, at a given scale $r$ of the ball in the tangent plane, we will have collected points exhibiting different amounts of curvature in the unbalanced tube geometry setting.
This is seen in Figure \ref{fig:radii_R_compare}(b), where the ambient radius of the collected points is much larger at a given scale when growing a ball in the tangent plane (red curve) than when growing a ball in the ambient space (green curve).  These observations imply that the true tangent space recovery error is sensitive to the balance, or lack thereof, of the geometry.
Finally, due to this sensitivity to the strongly anisotropic tube geometry, we notice that the true error indicates orthogonality at scales larger than indicated by our bound.  The minimum of our bound therefore remains within the range of scales that provide stable recovery.
\begin{figure}
  \centering
  \subfigure[bowl geometry]{\includegraphics[width=70mm]{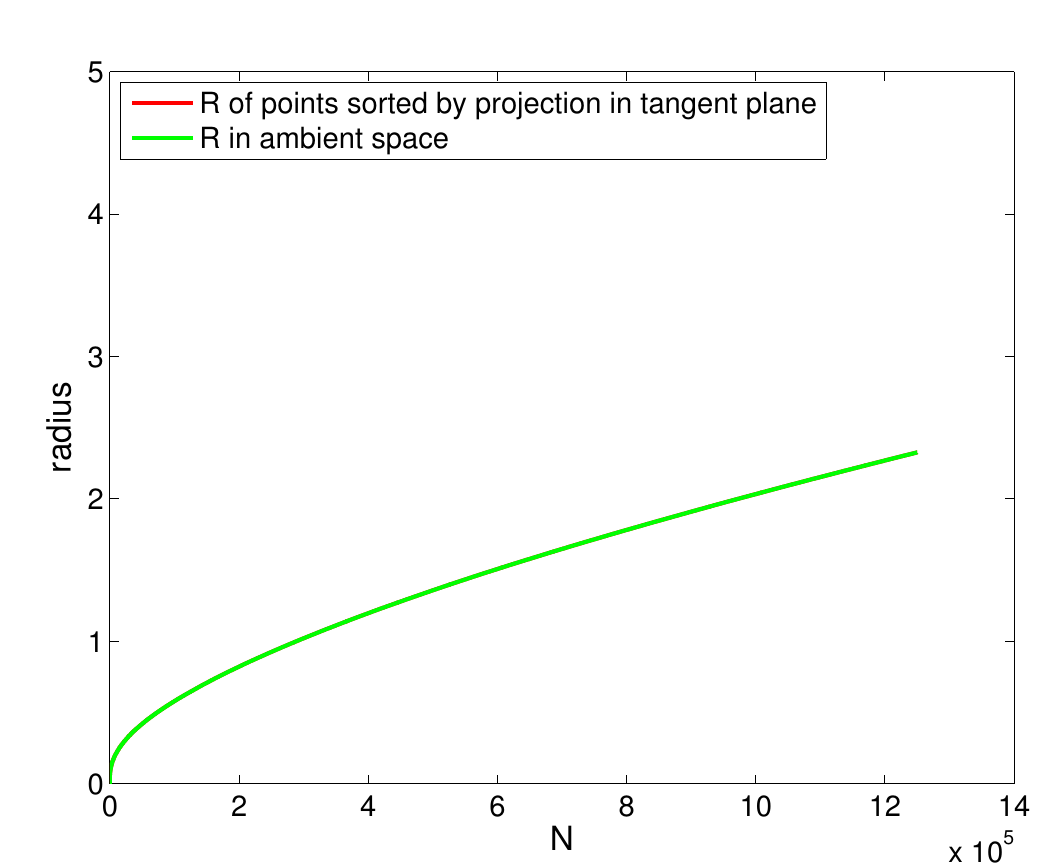}}
  \subfigure[tube geometry]{\includegraphics[width=70mm]{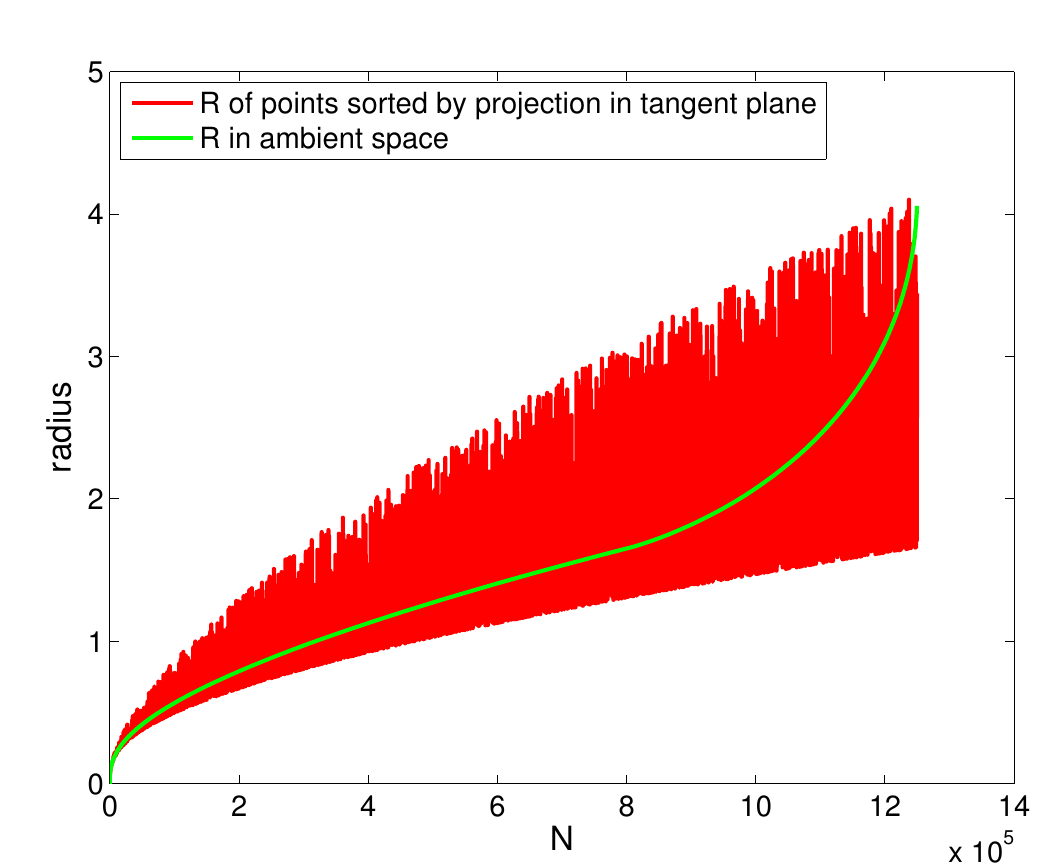}}
  \caption{\textcolor{black}{The radius $R$ is shown sorted according to the order in which points are discovered in the ambient space (green) and according to the order in which points are discovered when projected in the tangent plane (red).  The ordering is identical for the bowl geometry (left), where the green curve is on top of the red curve, because all principal curvatures are equal.  The ordering is very different for the tube geometry (right) where some directions exhibit greater curvature than others.  See text for discussion.}}
  \label{fig:radii_R_compare}
\end{figure}

We conclude this experimental section by noting that equation \eqref{eq:solve4r} provides only an approximation to $r$ and we therefore expect that tighter results are possible.  This avenue should be the subject of future investigation.
Nonetheless, the experimental results presented in this section indicate that our Main Result Theorem \ref{thm:mainresult1} may be used, with suitable modification according to \eqref{eq:solve4r}, to track the tangent space recovery error in the practical setting where only the ambient radius $R$ is available to the user.

Having demonstrated the utility of the Main Result Theorem \ref{thm:mainresult1}, we now turn our attention to the recovery of the unknown local origin.

\subsection{Finding the Local Origin}

\textcolor{black}{As explained previously, here we propose a ``plug-in'' to compute a
  ``clean estimate'', $\hat{x}_0$, of the point $x_0$ on $\mathcal{M}$
  that serves as the origin of the coordinate system in our analysis.
  At first glance, it might seem that a useful perturbation bound
  should assume that the analysis is centered around a noisy point
  and account for this additional source of uncertainty. We advocate that
  this is an unnecessarily pessimistic perspective, and we therefore offer
  an alternate approach: we show that a reliable estimate,
  $\hat{x}_0$, of $x_0$ can be computed from a noisy data set. Using
  $\hat{x}_0$, the reader can directly apply the theoretical bounds
  found in Section \ref{sec:mainresult} to analyze a noisy set of
  points.  The algorithm to compute $\hat{x}_0$ is simple and
  computationally inexpensive (requiring no matrix decompositions),
  and makes use of the geometric information encoded in the trajectory
  of the points' center of mass over several scales.  It is worth
  mentioning that we expect the proposed algorithm to be a universal
  first step for a local, multiscale analysis of the type presented in
  this paper.  Further intuition, details, and experiments are
  presented below.}

  \textcolor{black}{
  It is important to clearly state the role of $x_0$ in the practical
  implementation of this work: given a noisy point $y \in
  \mathbb{R}^D$ selected by the user, $x_0$ is the closest point on
  the ``clean'' manifold $\mathcal{M}$ around which we want to
  estimate the tangent plane $\TP$. Since we assume that
  $\mathcal{M}$ is smooth, there exists a neighborhood about $x_0$
  where the manifold is described by the model \eqref{eq:local-model},
  and $x_0$ is the origin of this model.  Because $x_0$ is the
  projection of $y$ on $\mathcal{M}$, $y-x_0$ is normal to $\TP$, and
  the points $y$ and $x_0$ therefore have the same coordinates in the
  tangential directions.  Rotating the coordinate system to align the
  axes with these directions, our goal is to move from $y$ to $x_0$ in
  the directions normal to the tangent plane.  Figure \ref{fig:algorithm_setup} provides
  an illustration of this framework.  We remark that the rotation of
  the coordinate axes is merely for notational convenience and will be
  discussed below.}
  \begin{figure}
    \centering
    \includegraphics[width=120mm]{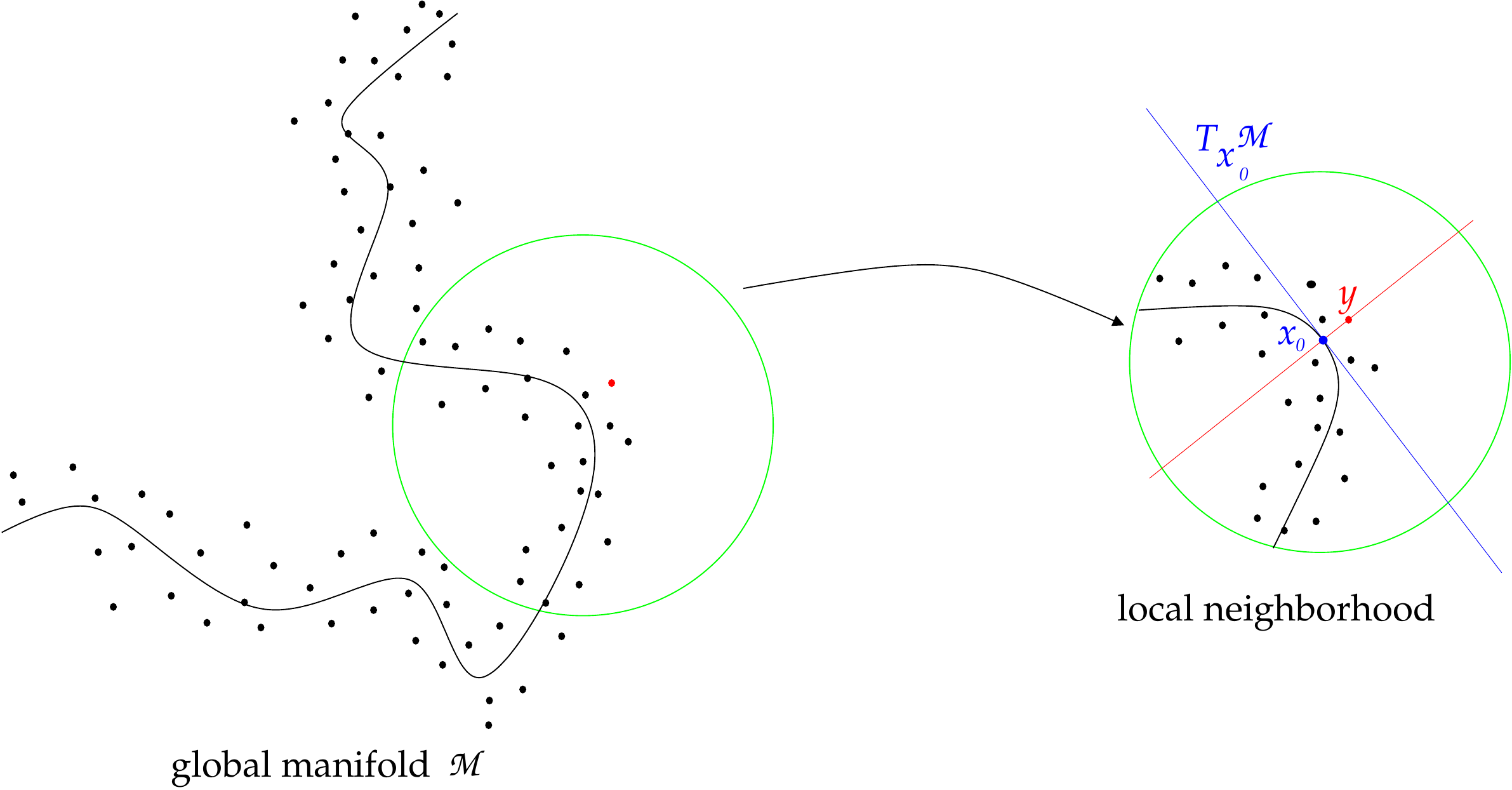}
    \caption{\textcolor{black}{Left: the user selects a noisy point $y$ (in red) close to the smooth 
  manifold $\mathcal{M}$. Right: a local neighborhood is extracted. The point $x_0$ (blue) 
  that is closest to $y$ on the manifold becomes the point at which we compute $\TP$ 
  (blue).  The local coordinate system is defined by the tangent plane $\TP$ (blue) and the 
  normal space $\NS$ (red).  Neither the computation of the perturbation bound nor the estimation 
  of $x_0$ require that the unknown rotation be estimated.}}
    \label{fig:algorithm_setup}
\end{figure}

Our strategy will be to compute the center of mass $\overline{X}$ about $y$ and track the trajectory of each coordinate of $\overline{X}$ as the radius about $y$ grows from small to large scales.  We use the term ``trajectory'' to refer to the coordinate(s) of the sequence of sample means $\overline{X}$ computed over growing radii.  As we will see, these trajectories contain all of the geometric information necessary to recover $x_0$ and is robust to the presence of noise.
The steps for recovering $x_0$ are given below as Algorithm \ref{alg:algorithm1}.

\begin{algorithm}[!h]
  \caption{Recovering the Local Origin $x_0$}
  \label{alg:algorithm1}
  \color{black}
  \noindent \textbf{Input:} Noisy points $X = \{x^{(i)}\}_{i=1}^N$, reference point $y \in \mathbb{R}^D$, scale intervals $\{\mathcal{I}^{(m)}\}_{m=1}^M$ such that $\mathcal{I}^{(m)} = [R^{(m,1)},R^{(m,2)}]$ with $R^{(m,1)} < R^{(m,2)} ~\forall m$ and $\begin{cases} R^{(p,1)} \leq R^{(q,1)} \\ R^{(p,2)} \leq R^{(q,2)} \end{cases}$ for $p > q$\\[6pt]
  \textbf{Outpt:} Estimate $\widehat{x}_{0}$ of the local origin $x_0 \in \mathcal{M}$\\[4pt]

  \noindent \textbf{FOR} each scale interval $\mathcal{I}^{(m)}, m = 1,\dots,M$:

  \begin{enumerate}

  \item Center a ball at $y$ and compute $\overline{X} = \frac{1}{N_B}\sum_{i=1}^{N_B} x^{(i)}$, the mean of the points inside the ball $B^D_y(R_y)$, $\forall R_y \in \mathcal{I}^{(m)}$, where $N_B = |B^D_y(R_y)|$.

  \item \textbf{FOR} each coordinate $j = 1,\dots,D$:

    \begin{enumerate}

    \item Fit (in the least squares sense) the trajectory of $\overline{X}_j$ to the
      model
      \begin{equation*}
        q_y(R_y) = \beta_2 R_y^2 + \beta_0,
      \end{equation*}
      over the range of scales in $\mathcal{I}^{(m)}$,
      explicitly requiring a zero first derivative at $R_y=0$
      
    \item Set $\widehat{x}_{0_j}^{(m)} = \beta_0$

    \end{enumerate}
    
  \item[] \textbf{END}
    
  \item Set $y = \widehat{x}_{0}^{(m)}$
    
  \end{enumerate}
  
  \noindent \textbf{END}\\

  \noindent Return $\widehat{x}_{0} = \widehat{x}_{0}^{(M)}$ as the estimate of the local origin

\end{algorithm}

The trajectory of each coordinate of $\overline{X}$ will be noisy and
unreliable at very small scales.  However, due to the averaging
process, the uncertainty from both the noise and the random sampling
is overcome at large scales.  Thus, the large scale trajectory reaches
a ``steady state behavior'' that is essentially free of uncertainty
and encodes information about the initial state, i.e., the noise-free
trajectory very close to $x_0$.

\textcolor{black}{
\begin{remark}
  The algorithm described in this section can be understood in the
  context of the estimation of the center location of the probability
  density associated with the clean point $x_0$ on $\mathcal
  M$. Indeed, our model assumes that a noisy point $x$ is obtained by
  perturbing a clean point $\ell +c$ by adding Gaussian noise. The
  probability distribution of the noisy points is thus given by the
  convolution of a $D$-dimensional Gaussian density $G_{\sigma}$ with
  the $D$-dimensional probability density $f_{\mathcal{M}}$ of the
  clean points, which is supported solely on $\mathcal{M}$,
\begin{equation*}
f_{\mathcal{M}} \ast G_{\sigma} (x).
\end{equation*}
The goal of the algorithm is to recover the clean point $x_0$ around
which $f_{\mathcal{M}}$ is localized, given some noisy realizations
$X$ sampled from the probability density $f_{\mathcal{M}} \ast
G_{\sigma} (x)$. This can be achieved by removing the effect of the
blurring (a process known as deconvolution \cite{Genovese12b}) caused
by $G_{\sigma}$, and computing a ``sharp'' estimate of the density
$f_{\mathcal{M}}$ around $x_0$. There exists an expansive literature
on such deblurring problems. A very successful approach consists in
reversing the heat equation associated with the blurring at increasing
scales (e.g., \cite{Perona90,Osher90}). This idea is the essence of
our algorithm. By tracking the centroid of a ball of decreasing size,
we can extrapolate this trajectory in the limit where the ball has radius
zero, and effectively compute $\lim_{\sigma\rightarrow 0}
f_{\mathcal{M}} \ast G_{\sigma} (x_0)$. This process yields the initial
origin with very little uncertainty even for very high noise and high
curvature. 
\end{remark}
}

Let us now provide further intuition for why such a procedure will work.  
The reader is asked to be mindful that we will only provide an overview of the results and that a rigorous development of the convergence properties is left for future work.

\subsubsection{Center of Mass Trajectory}
Following the local model \eqref{eq:local-model} with origin $x_0$, a neighboring point $x$ has coordinates of the form
\begin{equation}
  x = x_0 + \ell + c + e = 
  \begin{bmatrix}
    x_{0_1} \\ \\ \\ \vdots \\ \\ \\ x_{0_D}
  \end{bmatrix}
  +
  \begin{bmatrix}
    \ell_1 \\ \vdots \\ \ell_d \\ \\ 0 \\ \vdots \\ 0
  \end{bmatrix}
  +
  \begin{bmatrix}
    0 \\ \vdots \\ 0 \\ \\ c_{d+1} \\ \vdots  \\ c_D
  \end{bmatrix}
  +
  \begin{bmatrix}
    e_1 \\ \\ \\ \vdots \\ \\ \\ e_D
  \end{bmatrix},
\end{equation}
and
coordinate $j$ of $\overline{X}$ is of the form
\begin{equation}
  \overline{X}_j = \frac{1}{N}\sum_{i=1}^N x^{(i)}_j = 
  \begin{cases}
    x_{0_j} + \frac{1}{N}\sum_{i=1}^N \ell^{(i)}_j + \frac{1}{N}\sum_{i=1}^N e^{(i)}_j ,& j \leq d\\
    x_{0_j} + \frac{1}{N}\sum_{i=1}^N c^{(i)}_j + \frac{1}{N}\sum_{i=1}^N e^{(i)}_j ,& j > d.
  \end{cases}
\end{equation}
The sample mean $\overline{X}_j = \frac{1}{N}\sum_{i=1}^N x^{(i)}_j$ approximates $\mathbb{E}[x_j]$ with the uncertainty decaying as $1/\sqrt{N}$.  More precisely, by the Hoeffding inequality and the Gaussian tail bound, we have the following intervals for coordinate $j$ at scale $N$:
\begin{equation}
  \label{eq:confidence-intervals}
  \overline{X}_j \in
  \begin{cases}
    \Big[x_{0_j} -\frac{\sqrt{2}\xi}{\sqrt{N}}(r+\sigma), \; x_{0_j} +\frac{\sqrt{2}\xi}{\sqrt{N}}(r+\sigma)\Big],& j \leq d\\[8pt]
    \Big[\left(x_{0_j} + \frac{K_j r^2}{2(d+2)}\right) - \frac{\sqrt{2}\xi}{\sqrt{N}}\left(\frac{\sqrt{d}K_j^{(+)} r^2}{2} + \sigma \right), \; \left(x_{0_j} + \frac{K_j r^2}{2(d+2)}\right) + \frac{\sqrt{2}\xi}{\sqrt{N}}\left(\frac{\sqrt{d}K_j^{(+)} r^2}{2} + \sigma \right)\Big], & j > d
  \end{cases}
\end{equation}
with probability greater than $1-6e^{-\xi^2}$.
We see that while the coordinates exhibit variation about their means at small scales, they reach their average (steady state) behavior with high probability at large scales.
Thus, the large scale coordinate trajectories are controlled with little uncertainty for densely sampled data.

\begin{remark}
  More generally, we expect to
  observe data in a rotated coordinate system.  Consider the setting in $\mathbb{R}^2$ for a 1-dimensional manifold after applying a rotation to our conventional coordinate system.  The observed coordinates will be of the form
  \begin{align}
    \begin{pmatrix}
      \overline{X}_1\\
      \overline{X}_2
    \end{pmatrix}
    &=
    \begin{pmatrix}
      x_{0_1}\\
      x_{0_2}
    \end{pmatrix}
    +
    \begin{pmatrix}
      Q_{11} & Q_{12}\\
      Q_{21} & Q_{22}
    \end{pmatrix}
    \begin{pmatrix}
      \frac{1}{N}\sum_{i=1}^N \ell^{(i)}_1 + \frac{1}{N}\sum_{i=1}^N e^{(i)}_1 \\
      \frac{1}{N}\sum_{i=1}^N c^{(i)}_2 + \frac{1}{N}\sum_{i=1}^N e^{(i)}_2
    \end{pmatrix}
    \nonumber \\[6pt]
    &=
    \begin{pmatrix}
      x_{0_1} + Q_{11}\mathbb{E}[\ell] + Q_{12}\mathbb{E}[c] + (Q_{11}+Q_{12})\mathbb{E}[e] \pm \mathcal{O}\left(\frac{1}{\sqrt{N}}\right) \\
      x_{0_2} + Q_{21}\mathbb{E}[\ell] + Q_{22}\mathbb{E}[c] + (Q_{21}+Q_{22})\mathbb{E}[e] \pm \mathcal{O}\left(\frac{1}{\sqrt{N}}\right)
    \end{pmatrix}
    \quad \text{(w.h.p.)} \\[6pt]
    &=
    \begin{pmatrix}
      x_{0_1} + Q_{12}\frac{K_2r^2}{2(d+2)}  \pm \mathcal{O}\left(\frac{1}{\sqrt{N}}\right) \\
      x_{0_2} + Q_{22}\frac{K_2r^2}{2(d+2)} \pm \mathcal{O}\left(\frac{1}{\sqrt{N}}\right),
    \end{pmatrix}
    \nonumber
  \end{align}
  where $Q = \left(\begin{smallmatrix} Q_{11} & Q_{12} \\ Q_{21} & Q_{22} \end{smallmatrix}\right)$ is a unitary matrix.
  We see that all coordinates have the same form as coordinates $j > d$ in \eqref{eq:confidence-intervals} with the slight modification introduced by the $Q_{mn}$ terms.  
  In general, we will observe a linear combination of all coordinates with weights $Q_{mn}<1$.
  In particular, all coordinates will be of leading order $r^2$ with a constant intercept (the origin), and all other orders of $r$ appear as finite sample uncertainty terms that decay as $1/\sqrt{N}$.
Because an arbitrary rotation leaves all coordinates with the same form as that of the coordinates $j>d$ in equation \eqref{eq:confidence-intervals}, we proceed with the analysis of these coordinates without loss of generality.
\end{remark}

Continuing from \eqref{eq:confidence-intervals}, we use a calculation similar to \eqref{eq:R^2} to show that $r^2 \approx R^2$ for small $r$.  We therefore expect the coordinate trajectories ($j >d$) to be quadratic functions of the observed radius $R$ with intercept $x_0$ and zero first derivative at $R=0$.  Fitting the observed trajectory of each coordinate to the model
\begin{equation}
  q(R) = \beta_2 R^2 + \beta_0
  \label{eq:smallscale-model-1}
\end{equation}
provides the least squares estimate of the origin $\widehat{x}_{0_j} = \beta_0$.  By explicitly enforcing the zero first derivative condition, the model \eqref{eq:smallscale-model-1} should be robust to uncertainty in the observed data at small scales.  Moreover, initial estimates of $\widehat{x}_{0_j}$ may be obtained from the stable, large scale trajectories to anchor the small scale estimate using \eqref{eq:smallscale-model-1}.  We now examine this procedure in more detail.

\subsubsection{Estimating $x_0$}

Equation \eqref{eq:confidence-intervals} confirms our intuition that the large scale trajectory, smoothed from the averaging process, is very stable due to the $1/\sqrt{N}$ decay of the finite sample uncertainty terms.  We must now cast this trajectory in terms of an observable radius $R_y$, the radius of a ball in $\mathbb{R}^D$ centered about the point $y$ in the presence of noise.  Recall that the intent of the following discussion is to informally derive the correct order for all terms, with complete rigor reserved for future work.

Consider first the effect of measuring the radius about a point other than $x_0$.  Let $\tau$ denote the offset vector,
\begin{equation*}
  \tau = y - x_0 = \Big[0 ~ \cdots ~ 0 ~ ~ \tau_{d+1} \cdots ~ \tau_D\Big]^T,
\end{equation*}
since $y$ and $x_0$ only differ in their normal components.
A calculation similar to \eqref{eq:R^2} shows
\begin{equation}
  \label{eq:Ry}
  R_y^2 = \|x-y\|^2 = \|x_0 + \ell + c - \tau - x_0\|^2 =  \|\ell\|^2 + \|c-\tau\|^2 \leq r^2 + \gamma r^4 + \|\tau\|^2.
\end{equation}
Solving for $r^2$ and injecting into \eqref{eq:confidence-intervals} yields the following expression for $\overline{X}_j$ (coordinates $j > d$) at scale $N$, holding with high probability:
\begin{align}
  \label{eq:largescale-nonoise}
  &\overline{X}_j \in \Bigg[ a_1 R_y + (x_{0_j} - a_0) - \frac{\sqrt{2}\xi}{\sqrt{N}}a_{-1}\; , \; a_1 R_y + (x_{0_j} - a_0) + \frac{\sqrt{2}\xi}{\sqrt{N}}a_{-1} \Bigg], \\[4pt]
  &\text{for} \; R_y > \sqrt{\| \tau\|^2 + \frac{1}{4\gamma}},
\end{align}
where
\begin{equation}
  \label{eq:a1a0}
  a_1 = \frac{K_j}{2(d+2)\sqrt{\gamma}}, \quad
  a_0 = \frac{K_j}{4(d+2)\gamma} + \mathcal{O}\left(\frac{1}{R_y}\right),
\end{equation}
with uncertainty term
\begin{equation}
  a_{-1} = \frac{1}{2}\sqrt{\frac{d}{\gamma}}K_j^{(+)}R_y - \frac{1}{4}\frac{\sqrt{d}}{\gamma}K_j^{(+)} + \mathcal{O}\left(\frac{1}{R_y}\right).
\end{equation}
Next, reasoning in a manner similar to \eqref{eq:R^2}, we introduce the following correction for the presence of the noise, enlarging the radius $R_y$ in \eqref{eq:largescale-nonoise} by $\sigma\sqrt{D}$ :
\begin{equation*}
  R_y \leftarrow R_y + \sigma\sqrt{D}.
\end{equation*}
We finally rewrite \eqref{eq:largescale-nonoise} to yield the expression for $X_j$ (coordinates $j > d$) at scale $N$, holding with high probability:
\begin{align}
  \label{eq:largescale}
  &\overline{X}_j \in \Bigg[ a_1 R_y + \left(x_{0_j} - a_0 + a_1 \sigma\sqrt{D}\right) - \frac{\sqrt{2}\xi}{\sqrt{N}}a_{-1}\; , \; a_1 R_y + \left(x_{0_j} - a_0 + a_1 \sigma\sqrt{D}\right) + \frac{\sqrt{2}\xi}{\sqrt{N}}a_{-1} \Bigg], \\
  &\text{for} \; R_y > \sqrt{\| \tau\|^2 + \frac{1}{4\gamma}}, \nonumber
\end{align}
with $a_1$ and $a_0$ as given by \eqref{eq:a1a0} and uncertainty term $a_{-1}$ now taking the form
\begin{equation}
  a_{-1} = \frac{1}{2}\sqrt{\frac{d}{\gamma}}K_j^{(+)}R_y + \frac{1}{2}\sqrt{\frac{d}{\gamma}}K_j^{(+)}\sigma\sqrt{D} - \frac{1}{4}\frac{\sqrt{d}}{\gamma}K_j^{(+)} + \sigma +\mathcal{O}\left(\frac{1}{R_y}\right).
\end{equation}

While \eqref{eq:largescale} indicates that the large scale trajectory is linear in $R_y$, all of the necessary geometric information for Algorithm \ref{alg:algorithm1} to succeed is encoded in this trajectory.  To see this, we proceed momentarily by taking a path slightly different from that of the proposed algorithm.  Consider fitting the large scale trajectory to the model
\begin{equation}
  \label{eq:linear-model}
  q_y^{linear}(R_y) = \alpha_1 R_y + \alpha_0
\end{equation}
over the range of scale $\mathcal{I}^{(m)} = [R_y^{(m,1)}, R_y^{(m,2)}]$.  Let $R_y^{(m,1)}$ correspond to $N^{(m,1)}$ points, $R_y^{(m,2)}$ correspond to $N^{(m,2)}$ points, $N^{(m,1)} < N^{(m,2)}$, and let $\widetilde{N^{(m)}} = (N^{(m,1)} + N^{(m,2)})/2$.
The least squares fit of the large scale $\overline{X}_j$ trajectory to \eqref{eq:linear-model} yields the coefficients
\begin{equation}
  \alpha_1 \in \Bigg[a_1 - \frac{\xi}{\sqrt{\widetilde{N^{(m)}}}}\sqrt{\frac{d}{2\gamma}}K_j^{(+)} \; , \; a_1 + \frac{\xi}{\sqrt{\widetilde{N^{(m)}}}}\sqrt{\frac{d}{2\gamma}}K_j^{(+)}\Bigg]
\end{equation}
\begin{multline}
  \alpha_0 \in \Bigg[\left(x_{0_j} - a_0 + a_1\sigma\sqrt{D}\right) - \frac{\sqrt{2}\xi}{\sqrt{\widetilde{N^{(m)}}}}\left(\sigma + \frac{1}{2}\sqrt{\frac{d}{\gamma}}K_j^{(+)}\left(\sigma\sqrt{D} - \frac{1}{2\sqrt{\gamma}}\right)\right) \; ,
  \\
  \; \left(x_{0_j} - a_0 + a_1\sigma\sqrt{D}\right) + \frac{\sqrt{2}\xi}{\sqrt{\widetilde{N^{(m)}}}}\left(\sigma + \frac{1}{2}\sqrt{\frac{d}{\gamma}}K_j^{(+)}\left(\sigma\sqrt{D} - \frac{1}{2\sqrt{\gamma}}\right)\right) \Bigg].
\end{multline}
Noting that the (rescaled) mean curvature $K_j$ is encoded in $a_1$ and $a_0$, we may recover a large scale estimate of $x_{0_j}$ by setting
\begin{equation}
  \label{eq:largescale-estimate1}
  \widehat{x}_{0_j}^{(m)} = \alpha_0 - \alpha_1\sigma\sqrt{D} + \alpha_1^2\frac{(d+2)}{K_j}.
\end{equation}
Then we have
\begin{equation}
  \left|x_{0_j} - \widehat{x}_{0_j}^{(m)}\right|
  ~\leq~ 
  \frac{\sqrt{2}\xi}{\sqrt{\widetilde{N^{(m)}}}} \left(\sigma + \frac{\sqrt{d}}{2\gamma}K_j^{(+)} \right) + \frac{\xi^2}{\widetilde{N^{(m)}}}\frac{d(d+2)}{2\gamma}\frac{(K_j^{(+)})^2}{|K_j|},
  \label{eq:error1}
\end{equation}
with high probability.

\begin{remark}
  The $k$th point of the $\overline{X}_j$ trajectory has an uncertainty term that decays as $1/\sqrt{k}$.  For convenience, we have replaced the point-by-point uncertainty decay with a constant factor of $1/\sqrt{\widetilde{N^{(m)}}}$ above, where $\widetilde{N^{(m)}}$ is the number of points in the middle of the current interval.  A more rigorous analysis would account for the heteroskedasticity of the sequence of sample means $\overline{X}_j$ and use, e.g., a weighted least squares fit to the model.
\end{remark}

\textcolor{black}{We may use these calculations to understand the initial large scale exploration performed by Algorithm \ref{alg:algorithm1}.  The estimate produced by the algorithm may be seen as the result of replacing the trajectory with a linear function of $R_y$ as given by \eqref{eq:largescale}.  Then, discarding the data, we work only with this linear approximation over all $R_y$.  By doing so, we are discarding the quadratic behavior expected at small scales near $x_0$, as this part of the trajectory is damaged by the noise.
We then recover the expected quadratic behavior by fitting the linear approximation to the following quadratic model,
\begin{equation}
  \label{eq:smallscale}
  q_y^{quad}(R_y) = \beta_2 R_y^2 + \beta_0,
\end{equation}
where the zero first derivative condition is explicitly enforced.
The estimate for coordinate $j$ of $x_0$ has the form
\begin{equation}
  \widehat{x}_{0_j} = \alpha_0 + \alpha_1 \mathscr{F}(\mathcal{I}^{(m)}),
  \label{eq:larescale-estimate2}
\end{equation}
where
\begin{equation}
  \mathscr{F}(\mathcal{I}^{(m)}) = \frac{(R_y^{(m,2)})^2 + 4R_y^{(m,2)}R_y^{(m,1)} + (R_y^{(m,1)})^2}{6(R_y^{(m,2)}+R_y^{(m,1)})}
\end{equation}
is a function of the scale interval.
Comparing to \eqref{eq:largescale-estimate1}, this estimate is equivalent to the previous large scale procedure when we choose
\begin{equation}
  \mathscr{F}(\mathcal{I}^{(m)}) = \frac{\alpha_1(d+2)}{K_j} - \sigma\sqrt{D}.
  \label{eq:curlyF}
\end{equation}
This choice also can be shown to minimize the error of the estimate in \eqref{eq:larescale-estimate2}.
In summary, if we could very carefully select the range of scales to satisfy \eqref{eq:curlyF}, which requires \textit{a priori} knowledge of curvature, we could compute an estimate of $x_0$ in one step.
While we cannot expect to choose exactly the right interval to satisfy \eqref{eq:curlyF},
we observe in practice (see Section \ref{subsec:origin-numerical}) that the decreasing sequence of intervals used by Algorithm \ref{alg:algorithm1} will contain a proxy that allows for an accurate estimate.
}

\textcolor{black}{
The result of this procedure is an estimate $\widehat{x}_0^{(m)}$ over scale interval $\mathcal{I}^{(m)}$ that is very close to the true $x_0$.  Setting $y = \widehat{x}_0^{(m)}$, we are left with only a very small offset vector $\tau$:
\begin{equation}
  \|\tau\|^2 = \frac{2\xi^2}{\widetilde{N^{(m)}}}  \left(\sigma^2 D + \frac{\sigma\sqrt{d}}{\gamma}\sum_{j=1}^DK_j^{(+)} + \frac{d}{4\gamma^2} (K^{(+)})^2 \right) + \mathcal{O}\left(\frac{1}{(\widetilde{N^{(m)}})^{3/2}}\right).
  \label{eq:tau-norm}
\end{equation}
The trajectories $X_j$ may now be recomputed by centering a ball about $y = \widehat{x}_0^{(m)}$ and the fitting procedure is repeated over scale interval $\mathcal{I}^{(m+1)}$.
The error bound \eqref{eq:tau-norm} shows that if we keep the number of points sufficiently large (given a dense enough sampling), even at small scales, we can decrease the uncertainty on the estimate of $x_0$.
The accurate estimation of $x_0$ by Algorithm \ref{alg:algorithm1} is demonstrated in the next section.
}

\subsubsection{Experimental Results} \label{subsec:origin-numerical}

In this section, we test the performance of Algorithm \ref{alg:algorithm1} on several data sets over a range of parameters and tabulate the results.  MATLAB code implementing Algorithm \ref{alg:algorithm1} is available for download at \url{http://www.danielkaslovsky.com/code}.

Data sets of $N=50,000$ points sampled from $d$-dimensional manifolds embedded in $\mathbb{R}^{D}$ were generated according to the local model \eqref{eq:local-model} in the same manner as for all other experiments (see Section \ref{subsec:numerical-results}).  For each data set, the local origin $x_0 \in \mathbb{R}^D$ was chosen by sampling each coordinate from $\mathcal{U}[-10,10]$, where $\mathcal{U}[a,b]$ is the uniform distribution supported on $[a,b]$.  An initial reference point $y \in \mathbb{R}^D$ was chosen as specified in Table {\ref{tab:origin}} and a random rotation was applied to both the data set and $y$.
Seven different experiments were performed with parameters as listed in Table \ref{tab:origin}.  For each experiment, Algorithm \ref{alg:algorithm1} was used to recover the local origin of 10 data sets starting from the randomly initialized reference point $y$.  The $\ell^{\infty}$ error ($\max_j |x_{0_j} - \widehat{x}_{0_j}|$) and mean squared error ($\sum_{j=1}^D (x_{0_j} - \widehat{x}_{0_j})^2/D$) of each trial were recorded, with the mean and standard deviation over the 10 trials reported in Table \ref{tab:origin}.  The scale intervals were fixed across all experiments to be: $\mathcal{I}^{(1)} = [0.5N, 0.75N]$, $\mathcal{I}^{(2)} = [1, 0.4N]$, $\mathcal{I}^{(3)} = [1, 0.3N]$, and $\mathcal{I}^{(4)} = [1, 0.25N]$.

\begin{table}
\addtolength{\tabcolsep}{-9pt}
  \caption{\textcolor{black}{Parameters for the data sets used to test Algorithm \ref{alg:algorithm1} with the $\ell^{\infty}$ error and MSE reported over 10 trials (mean $\pm$ standard deviation).}\label{tab:origin}}
  \color{black}
  \centering
  \begin{tabular}{|c||c|c|c|c|c||r|r|}
      \hline
      & & & & & & &\\
      & & & $\kappa_n^{(i)}$ & & & &\\
      & & & \footnotesize $(d+1)\leq i \leq D$\normalsize & & $\tau_j = y_j - x_{0_j}$ & &\\
      Experiment & $d$ & $D$ & \footnotesize $1\leq n \leq d$\normalsize & $\sigma$ & \footnotesize $(d+1)\leq j \leq D$\normalsize & \multicolumn{1}{|c|}{$\ell^{\infty}$ error} & \multicolumn{1}{|c|}{MSE} \\
      \hline
      \hline
      Baseline & & & & & & 0.01646 & 6.1321e-5 \\
      (bowl) & 3 & 20 & 1.0189 & 0.05 & $\mathcal{N}(0,4\sigma^2)$ & $\pm$0.00418 & $\pm$2.5291e-5 \\
      \hline
      & & & & & & 0.01171 & 3.0669e-5 \\
      Tube & 3 & 20 & Table \ref{tab:kk2} & 0.05 & $\mathcal{N}(0,4\sigma^2)$ & $\pm$0.00261 & $\pm$1.0460e-5 \\
      \hline
      & & & & & & 0.01658 & 5.8716e-5 \\
      Saddle & 3 & 20 & $\mathcal{U}[-2,2]$ & 0.05 & $\mathcal{N}(0,4\sigma^2)$ & $\pm$0.00680 & $\pm$4.5841e-5 \\
      \hline
      High Curvature & & & & & & 0.06031 & 0.00106 \\
      Saddle & 3 & 20 & $\mathcal{U}[-5,5]$ & 0.05 & $\mathcal{N}(0,4\sigma^2)$ & $\pm$0.02006 & $\pm$0.00076 \\
      \hline
      High-Dimensional & & & & & & 0.08005 & 0.00095 \\
      Saddle & 20 & 100 & $\mathcal{U}[-2,2]$ & 0.05 & $\mathcal{N}(0,4\sigma^2)$ & $\pm$0.00772 & $\pm$0.00012 \\
      \hline
      & & & & & & 0.05541 & 0.00074 \\
      High Noise & 3 & 20 & 1.0189 & 0.15 & $\mathcal{N}(0,4\sigma^2)$ & $\pm$0.00545 & $\pm$0.00013 \\
      \hline
      Large & & & & & & 0.01021 & 2.2915e-5 \\
      Initial Offset & 3 & 20 & 1.0189 & 0.05 & $(-1)^j \times 0.75$ & $\pm$0.00224 & $\pm$9.2499e-6 \\
      \hline
    \end{tabular}
\end{table}

The results in Table \ref{tab:origin} show that Algorithm \ref{alg:algorithm1} was able to accurately locate the true origin for all of the tested settings: bowl, tube, and saddle geometries; high noise; high curvature; high dimension; and large initial offset.  As expected, the largest errors occurred in the high noise and high curvature settings.  The high-dimensional setting also produced a comparatively large error.  However, this is not unexpected, as the noise level and curvature values are quite large for the $\mathbb{R}^{100}$ ambient space.
We see that Algorithm \ref{alg:algorithm1} is quite robust over a very large range of parameters and at relatively high noise levels.  We expect that the quality of approximation will be improved beyond these accurate initial results by using a careful choice of scale intervals $\mathcal{I}^{(m)}$ rather than hard-coded intervals for all data sets.  In particular, the $\mathcal{I}^{(m)}$ should be data-driven functions of dimension, noise, and curvature.

Figure \ref{fig:iteration} shows the convergence of five example coordinates for a ``Baseline'' data set (parameters given in Table \ref{tab:origin}) with $\tau_j$ drawn from the $\mathcal{N}(0,\sigma^2)$ distribution.
The difference between the coordinates of the initial center $y$ and the true origin $x_0$ are shown at iteration 0.  
The error of the estimate $\widehat{x}_{0_j}^{(m)}$ computed at scale interval $\mathcal{I}^{(m)}$ for each subsequent iteration $m$ is shown to decrease for $m > 1$.
The example results shown in the figure indicate that Algorithm \ref{alg:algorithm1} converges in very few iterations.
\begin{figure}
  \centering
  \includegraphics[width=90mm]{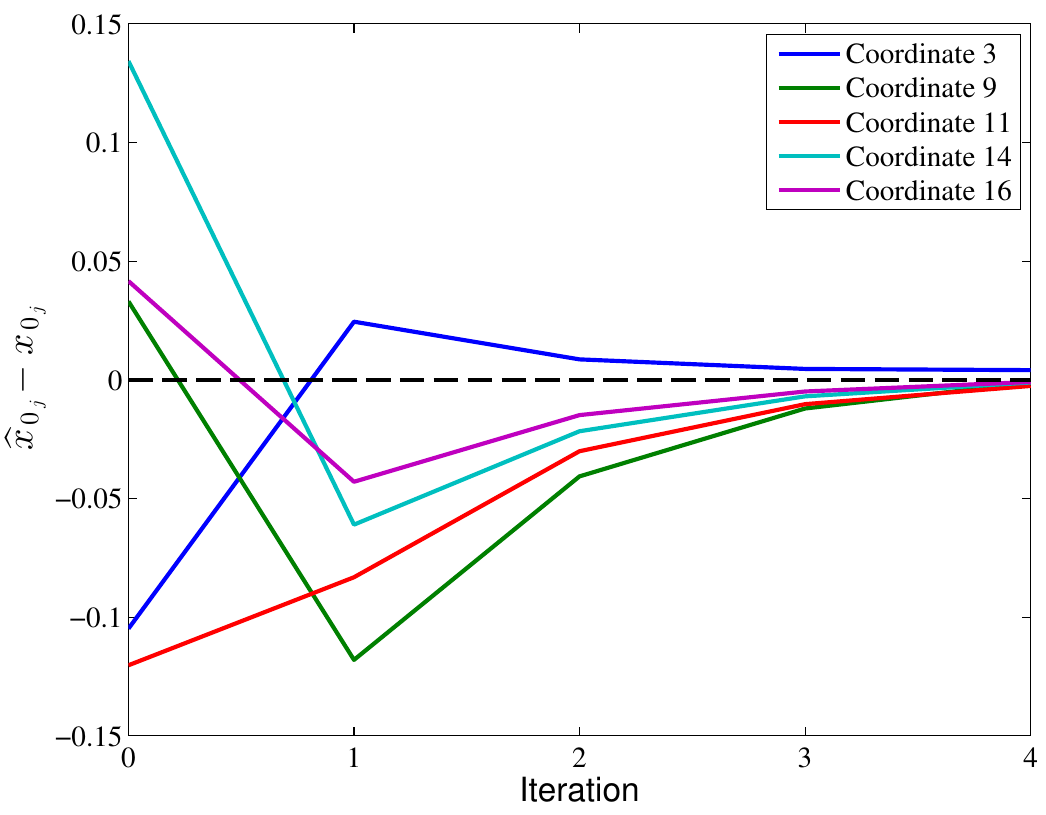}
  \caption{\textcolor{black}{Error of the estimate $\widehat{x}_{0_j}^{(m)}$ (for five example coordinates) at iteration $m$ of Algorithm \ref{alg:algorithm1} for a ``Baseline'' data set (see Table \ref{tab:origin}) with $\tau_j \sim \mathcal{N}(0,\sigma^2)$.}}
  \label{fig:iteration}
\end{figure}

\color{black}

\section{Discussion and Conclusion} \label{sec:discussion}

\subsection{Consistency with Previously Established Results}

Local PCA of manifold-valued data has received attention in several recent works (for example, those referenced in Section \ref{sec:intro}).  In particular, the analyses of \cite{Maggioni-long} and \cite{Singer}, after suitable translation of notation and assumptions, demonstrate growth rates for the PCA spectrum that match those computed in the present work.  The focus of our analysis is the perturbation of the eigenspace recovered from the local data covariance matrix.  We therefore confirm our results with those most similar from the literature.
\textcolor{black}{The most closely related results are those of \cite{Nadler}, in which matrix perturbation theory is used to study the PCA spectrum; \cite{Tyagi}, where neighborhood size and sampling conditions are given to ensure an accurate tangent space estimate from noise-free manifold-valued data; and \cite{Maggioni-MIT}, where theory is developed for the implementation of multiscale PCA to detect the intrinsic dimension of a manifold.}

In \cite{Nadler}, a finite-sample PCA analysis assuming a linear model is presented.  Keeping $N$ and $D$ fixed, the noise level $\sigma$ is considered to be a small parameter.  Much like the analysis of the present paper, the results are derived in the non-asymptotic setting.
However, the bound on the angle between the finite-sample and population eigenvectors is summarized in \cite{Nadler} for the asymptotic regime where $N$ and $D$ become large.  The result, restated here in our notation, takes the form:
\begin{equation*}
  \sin \theta_{\widehat{U}_1,U_1} ~\lesssim~ \frac{\sigma}{\sqrt{\lambda_d}}\sqrt{\frac{D}{N}} + \mathcal{O}(\sigma^2).
\end{equation*}
We note that the main results of \cite{Nadler} are stated for $N \leq D$ and that our analysis expects the opposite in general, although it is not explicitly required.  Nonetheless, by setting curvature terms to zero, our results recover the reported leading behavior following the same asymptotic regime as \cite{Nadler}, where terms $\mathcal{O}(1/\sqrt{N})$ are neglected and $\sigma$ is treated as a small parameter.  After setting all curvature terms to zero, we assume condition 1 holds such that the denominator $\delta$ is
sufficiently well conditioned and we may drop all terms other than
$\lambda_d$.  Then our Main Result has the form:
\begin{equation*}
 \sin \theta_{\widehat{U}_1,U_1} ~\lesssim~ \frac{1}{\sqrt{N}} \frac{1}{\lambda_d} \sigma\sqrt{d(D-d)}\left[\frac{r}{\sqrt{d+2}} + \sigma \right] = \frac{\sigma}{\sqrt{\lambda_d}}\frac{\sqrt{d(D-d)}}{\sqrt{N}} + \mathcal{O}(\sigma^2).
\end{equation*}
Setting $d=1$ to match the analysis in \cite{Nadler} recovers its curvature-free result.

Next, \cite{Tyagi} presents an analysis of local PCA differing from ours in two crucial ways.  First, the analysis of \cite{Tyagi} does not include high-dimensional noise perturbation and the data points are assumed to be sampled directly from the manifold.  Second, the sampling density is not fixed, whereas the neighborhood size determines the number of sample points in our analysis.  In fact, a goal of the analysis in \cite{Tyagi} is to determine a sampling density that will yield an accurate tangent space estimate.

Allowing for a variable sampling density has the effect of decoupling the condition number $\delta^{-1}$ from the norm $\|U_2^T \Delta U_1\|_F$ measuring the amount of ``lift'' in directions normal to the tangent space due to the perturbation.  The analysis of \cite{Tyagi} proceeds by first determining the optimal neighborhood radius $r^*$ in the asymptotic limit of infinite sampling, $N \rightarrow \infty$.  This approach yields the requirement that the spectra associated with the tangent space and curvature be sufficiently separated.  Translating to our notation, setting noise terms to zero, and assuming the asymptotic regime of \cite{Tyagi} such that we may neglect finite-sample correction terms, we recover condition 1 of our Main Result Theorem \ref{thm:mainresult1}:
\begin{equation}
 \lambda_d - \|U_2^T \Delta U_2\|_F = \lambda_d - \|U_2^T \frac{1}{N}CC^T U_2\|_F > 0.
 \label{eq:Tyagi-compare1}
\end{equation}
Thus, Theorem 1 of \cite{Tyagi} requires that $r$ be chosen such that the subspace recovery problem is well conditioned in the same sense that we require by condition 1.  Substituting the expectations for each term in \eqref{eq:Tyagi-compare1} yields
\begin{equation*}
 \frac{r^2}{(d+2)} - \frac{K^2r^4(d+1)}{2(d+2)^2(d+4)} > 0,
\end{equation*}
implying the choice $r < c/K$ (for a constant $c>0$), in agreement with the analysis of \cite{Tyagi}.  Once the proper neighborhood size has been selected, the decoupling assumed in \cite{Tyagi} allows a choice of sampling density large enough to ensure a small angle.  Again translating to our result \eqref{eq:main-result1}, once $r$ is selected so that the denominator $\delta$ is well conditioned, the density may be chosen such that the $1/\sqrt{N}$ decay of the numerator $\|U_2^T \Delta U_1\|_F$ allows for a small recovery angle.  Thus, we see that in the limit of infinite sampling and absence of noise, our results are consistent with those of \cite{Tyagi} in the fixed density setting.

\textcolor{black}{Finally, the recent work \cite{Maggioni-MIT} studies multiscale PCA and the growth of the corresponding spectrum to detect the intrinsic dimension of a manifold (or, more generally, a point cloud of random samples from a distribution concentrated around a low-dimensional manifold).
The authors prove, under appropriate conditions, that the empirical covariance of noisy points localized in a Euclidean ball about a noisy center is close to the population covariance of the underlying distribution, with high probability.
In particular, the authors' very detailed analysis shows that one may estimate the population covariance from the empirical covariance of noisy points that are localized before noise is added.  Then, following the work in \cite{LittleThesisDuke}, further effort in \cite{Maggioni-MIT} examines the effect of centering the multiscale analysis about a noisy origin.}

\textcolor{black}{Given an appropriate translation of the assumptions, the key results in \cite{Maggioni-MIT} are of the same order as those in the present work.  Using our notation, \cite{Maggioni-MIT} proceeds with an analysis of the geometric terms contained in the covariance $\frac{1}{N}\widetilde{X}\widetilde{X}^T$ and bounds the difference from the population covariance by controlling the perturbation due to the noisy center and the localization process.  In both the present analysis and that of \cite{Maggioni-MIT}, the empirical covariance $\frac{1}{N}\widetilde{X}\widetilde{X}^T$, computed from points localized before adding noise, provides the leading order terms that drive the behavior of $\|P-\widehat{P}\|_F$.  By moving the analysis from $r$ to $R$ in Section \ref{sec:numerical2}, we allow both curvature and noise to affect the localization of points and experimentally verify that $\|P-\widehat{P}\|_F$ is consistent with our Main Result.  Indeed, the results in Section \ref{sec:numerical2} experimentally test and confirm that the perturbation caused by such localization is small, as is theoretically derived in \cite{Maggioni-MIT}.
The effect of centering about a noisy origin is addressed in \cite{Maggioni-MIT} through a rescaling of the observable radius, and conditions are given that allow for the covariance of the set of points localized about a noisy origin to be close to the covariance of the points localized about the true origin.
The algorithm introduced in the present work, Algorithm \ref{alg:algorithm1} of Section \ref{sec:numerical2}, provides a simple method for recovering the true origin that may be used in practice.
Through a very different framework than that of the analysis in \cite{Maggioni-MIT}, our method uses the geometric information encoded in the center of mass to compute the true origin of the local neighborhood.
Our results therefore offer an algorithmic companion to the analysis presented in \cite{Maggioni-MIT}.}

\subsection{Algorithmic Considerations}

\subsubsection{Parameter Estimation}

\textcolor{black}{Practical methods must be developed to recover parameters such as dimension, curvature, and noise.  Such parameters are necessary for any analysis or algorithm and should be recovered directly from the data rather than estimated by \textit{a priori} fixed values.
  The experimental results presented above suggest the particular importance of accurately estimating the intrinsic dimension $d$, for which there exist several algorithms.
  Fukunaga introduced a local PCA-based approach for estimating $d$ in \cite{Fukunaga-Olsen}.
  The recent work in \cite{Maggioni-long} presents a multiscale approach that estimates $d$
  in a pointwise fashion.  Performing an SVD at each scale, $d$ is
  determined by examining growth rate of the multiscale singular values.
  It would be interesting to investigate if this approach remains robust
  if only a coarse exploration of the scales is performed, as it may be
  possible to reduce the computational cost through an SVD-update
  scheme.  Another scale-based approach is presented in
  \cite{Wang-Marron} and the problem was studied from a dynamical
  systems perspective in \cite{Froehling}.}

\textcolor{black}{There exist statistical methods for estimating the noise level present in a data set that should be useful in the context of this work (see, for example, \cite{Broomhead, Donoho2}).
We experimentally obtain a reliable estimate of the noise level from the median of the smallest singular values over several small neighborhoods (results not shown).
  In \cite{Maggioni-long}, the smallest multiscale singular values are used as an estimate for the noise level and a scale-dependent estimate of noise variance is suggested in \cite{Jones-new} for curve-denoising.
  Methods for estimating curvature (e.g., \cite{Williams-curvature, Kresk-curvature}) have been developed for application to computer vision and extensions to the high-dimensional setting should be explored. 
  Further, if one is willing to perform many SVDs of large matrices, our method of tracking the center of mass presented in Section \ref{sec:numerical2} combined with the growth rates for the PCA spectrum presented in \cite{Maggioni-long} might yield the individual principal curvatures.}

\subsubsection{Sampling}
For a tractable analysis, assumptions about sampling must be made.  In this work we have assumed uniform sampling in the tangent plane.  This is merely one choice and we have conducted initial experiments uniformly sampling the manifold rather than the tangent plane.  Results suggest that for a given radius, sampling the manifold yields a smaller curvature perturbation than that from sampling the tangent plane.  While more rigorous analysis and experimentation is needed, it is clear that consideration must be given to the sampling assumptions for any practical algorithm.

\subsubsection{From Tangent Plane Recovery to Data Parameterization}

The tangent plane recovered by our approach may not provide the best
approximation over the entire neighborhood from which it was derived.
Depending on a user-defined error tolerance, a smaller or larger sized
neighborhood may be parameterized by the local chart.  If high
accuracy is required, one might only parameterize a neighborhood of
size $N < N^*$ to ensure the accuracy requirement is met.  Similarly,
if an application requires only modest accuracy, one may be able to
parameterize a larger neighborhood than that given by $N^*$.

Finally, we may wish to use tangent planes recovered from different
neighborhoods to construct a covering of a data set.  There exist
methods for aligning local charts into a global coordinate system (for
example \cite{Brand, Roweis-Global, Zhang-Zha}, to name a few).  Care
should be taken to define neighborhoods such that a data set may be
optimally covered.

\section*{Funding}

This work was supported by the National Science Foundation [DMS-0941476 to F.G.M. and D.N.K., ACI-1226362 and DGE-0801680 to D.N.K.]; and the Department of Energy [DE-SC0004096 to F.G.M.].

\textcolor{black}{
\section*{Acknowledgements}
The authors are grateful to the anonymous reviewers for their insightful comments and suggestions that greatly improved the content and presentation of this manuscript.
}

\ifx\undefined\BySame
\newcommand{\BySame}{\leavevmode\rule[.5ex]{3em}{.5pt}\ }
\fi
\ifx\undefined\textsc
\newcommand{\textsc}[1]{{\sc #1}}
\newcommand{\emph}[1]{{\em #1\/}}
\let\tmpsmall\small
\renewcommand{\small}{\tmpsmall\sc}
\fi

\appendix

\section*{Appendix}

Technical calculations are presented in this appendix.  In particular,
the norm of each random matrix contributing to the perturbation term
$\Delta$, defined in equation \eqref{eq:fullDelta}, is bounded with
high probability.  The analysis is divided between three cases:
(1)  norms of products of bounded random matrices; (2) norms of
products of unbounded random matrices; and (3)   norms of products of
bounded and unbounded random matrices. 

Each case requires careful attention to derive a tight result that
avoids large union bounds and ensures high probability that is
independent of the ambient dimension $D$.  The analysis proceeds by
bounding the eigenvalues of the \textcolor{black}{covariance matrices of} $(L-\overline{L})$,
$(C-\overline{C})$, and $(E-\overline{E})$ using results from random
matrix theory and properties of the spectral norm.
A detailed analysis of each of the three cases follows.

Before we start the proofs, one last comment is in order. The reader
will notice that we sometimes introduce benign assumptions about the
number of samples $N$ or the dimensions $d$ or $D$ in order to
provide bounds that are simpler to interpret.  These assumptions are
not needed to derive any of the results; they are merely introduced to
help us simplify a complicated expression, and introduce upper bounds
that hold under these fairly benign assumptions. This should help the
reader interpret the size of the different terms.

\section*{Notation}
We often vectorize matrices by concatenating the columns of a
matrix. If $M= [m^{(1)} | \cdots |m^{(N)}]$, then we define
\begin{equation*}
  \vect{m} = \vct{M} = 
  \begin{bmatrix}
    m^{(1)}\\
    \vdots\\
    m^{(N)}
  \end{bmatrix}.
\end{equation*}
We denote the largest and smallest eigenvalue of a \textcolor{black}{square} matrix $M$ by
\begin{equation*}
  \lambda_{\max}(M) \quad \text{and} \quad \lambda_{\min}(M),
\end{equation*}
respectively.
In the main body of the paper, we use the standard notation
$\overline{X}$ to denote the sample mean of $N$ columns from the matrix
$X$. In this appendix, we introduce a second notation to denote the same
concept,
\begin{equation*}
  \wE{X} = \overline{X} = \frac{1}{N} \sum_{n=1}^N x^{(n)}.
\end{equation*}
Finally, we denote by $\E[X]$ the expectation of the random matrix $X$ and by $\prob[\mathscr{E}]$ the probability of event $\mathscr{E}$.
\section{Eigenvalue Bounds}

\subsection{Linear Eigenvalues} \label{subsec:linearEigenvalues}

We seek a bound on the maximum and minimum (nonzero) eigenvalue of the matrix
\begin{equation} \label{eq:LL^T}
  \frac{1}{N}(L-\overline{L})(L-\overline{L})^T
  = \frac{1}{N}\sum_{k=1}^N (\ell^{(k)} - \overline{\ell})(\ell^{(k)} - \overline{\ell})^T.
\end{equation}
As only the nonzero eigenvalues are of interest, we proceed by
considering only the nonzero upper-left $d \times d$ block of the
matrix in \eqref{eq:LL^T}, or equivalently, by ignoring the trailing
zeros of each realization $\ell^{(k)}$.  Thus, momentarily abusing
notation, we consider the matrix in \eqref{eq:LL^T} to be of dimension
$d \times d$. The analysis utilizes the following theorem found in
\cite{Tropp}.
\begin{theorem}[Matrix Chernoff II, \cite{Tropp}] \label{thm:Tropp}
  Consider a finite sequence $\{X_k\}$ of independent, random, self adjoint matrices that satisfy
  \begin{equation*}
    X_k \succcurlyeq 0 \qquad \text{and} \qquad \lambda_{\max}(X_k) \leq \lambda_{\infty} \quad \text{almost surely}.
  \end{equation*}
  Compute the minimum and maximum eigenvalues of the sum of expectations,
  \begin{equation*}
    \mu_{\min} := \lambda_{\min}\left(\sum_{k=1}^N \E \textcolor{black}{[}X_k\textcolor{black}{]} \right) \quad \text{and} \quad \mu_{\max} := \lambda_{\max}\left(\sum_{k=1}^N \E \textcolor{black}{[}X_k\textcolor{black}{]} \right).
  \end{equation*}
  Then
  \begin{align*}
    &\mathbb{P}\left[\lambda_{\min}\left(\sum_{k=1}^N X_k\right) \leq
      (1-\delta)\mu_{\min} \right] \leq
    d\left[\frac{e^{-\delta}}{(1-\delta)^{(1-\delta)}}\right]^{\mu_{\min}/\lambda_{\infty}},
    \quad \text{for } \delta \in [0,1], \text{ and} \\ 
    &\mathbb{P}\left[\lambda_{\max}\left(\sum_{k=1}^N X_k\right) \geq
      (1+\delta)\mu_{\max} \right] \leq
    d\left[\frac{e^{\delta}}{(1+\delta)^{(1+\delta)}}\right]^{\mu_{\max}/\lambda_{\infty}},
    \quad \text{for } \delta \geq 0.
  \end{align*}
\end{theorem}
\noindent We apply this result to 
\begin{equation*}
  X_k = \frac{1}{N}(\ell^{(k)} - \overline{\ell})(\ell^{(k)} - \overline{\ell})^T.
\end{equation*}
Clearly $X_k$ is a symmetric \textcolor{black}{positive semi-definite} matrix and we have $X_k \succcurlyeq 0$.  Next,
\begin{equation*}
  \lambda_{\max}(X_k) = \left\|\frac{1}{N}(\ell^{(k)} -
    \overline{\ell})(\ell^{(k)} - \overline{\ell})^T\right\|_2
  \textcolor{black}{=}
  \frac{1}{N}\left\|\ell^{(k)} - \overline{\ell}\right\|^2 \leq
  \frac{1}{N}\left(\|\ell\| + \|\overline{\ell}\|\right)^2 \leq
  \frac{4r^2}{N} 
\end{equation*}
and we set $\lambda_{\infty} = 4r^2/N$. Simple computations yield
\begin{equation*}
  \lambda_{\max}\left(\sum_{k=1}^N\E[X_k]\right) =
  \frac{r^2}{d+2}\left[1-\frac{1}{N}\right]^2, \quad \text{and}\quad 
  \lambda_{\min}\left(\sum_{k=1}^N\E[X_k]\right) =
  \frac{r^2}{d+2}\left[1-\frac{1}{N}\right]^2,
\end{equation*}
and we set
\begin{equation*}
  \mu_{\max} = \mu_{\min} = \mu = \frac{r^2}{d+2}\left[1-\frac{1}{N}\right]^2.
\end{equation*}
By Theorem \ref{thm:Tropp} and using standard manipulations, we have
the following result bound for the smallest eigenvalue, $\lambda_d$ in
our notation,
\begin{equation*}
  \lambda_d > \frac{r^2}{d+2}\left[1-\frac{1}{N}\right]^2\left[1 -
    \xi_{\lambda_d}
    \frac{1}{\sqrt{N}}\frac{\sqrt{8(d+2)}}{(1-\frac{1}{N})} \right]
\end{equation*}
with probability greater than $1-de^{-\xi_{\lambda_d}^2}$.  Similarly,
the following result holds for the largest eigenvalue, $\lambda_1$ in
our notation:
\begin{equation} 
  \label{eq:lambda_1} 
  \lambda_1 \leq
  \frac{r^2}{d+2}
  \left[1+
    \xi_{\lambda_1}
    \frac{5\sqrt{d+2}}{\sqrt{N}}
  \right] 
\end{equation}
with probability greater than $1-de^{-\xi_{\lambda_1}^2}$, as soon as
$N > 3$. We define the last upper bound as
\begin{equation}
  \lambda_{\text{bound}}(\xi) =\frac{r^2}{d+2}
  \left[1+  \xi
    \frac{5\sqrt{d+2}}{\sqrt{N}}
  \right],
  \label{eig-linear-bound}
\end{equation}
and we can use this bound to control the size of all the eigenvalues
of the matrix $  \frac{1}{N}(L-\overline{L})(L-\overline{L})^T$,
\begin{equation}
  \prob_{\ell} \left [
    \lambda_i \leq \lambda_{\text{bound}}(\xi), i =1,\ldots, d
  \right ]
  \ge 1-de^{-\xi^2}.
  \label{eig-linear-prob}
\end{equation}
Now that we have computed the necessary bounds for all nonzero linear
eigenvalues, we return to our standard notation for the remainder of
the analysis: each $\ell^{(k)}$ is of length $D$ with $\ell_j^{(k)} =
0$ for $d+1 \leq j \leq D$ and $L = [\ell^{(1)} | \ell^{(2)} | \cdots |
\ell^{(N)}]$ is a $D \times N$ matrix.
\subsection{Curvature Eigenvalues}
To bound the largest eigenvalue, $\gamma_1$, of
$\frac{1}{N}(C-\overline{C})(C-\overline{C})^T$ we note that the
spectral norm is bounded by the Frobenius norm and we use the bound on
the Frobenius norm derived in Section \ref{sec:pure-curvature}.  We
can use this bound to control the size of all the eigenvalues of the
matrix $ \frac{1}{N}(C-\overline{C})(C-\overline{C})^T$,
\begin{equation}
  \prob_{\ell} \left [
    \gamma_i \leq \gamma_{\text{bound}}(\xi), i =1,\ldots, \textcolor{black}{D-d}
  \right ]
  \ge 1-2e^{-\xi^2}.
  \label{eig-curvature-prob}
\end{equation}
where
\begin{equation} 
  \gamma_{\text{bound}}(\xi) =
  \frac{r^4}{2(d+4)(d+2)^2}\sqrt{\sum_{i,j=d+1}^D \negthickspace\negthickspace
    [(d+1)K_{nn}^{ij}-K_{mn}^{ij}]^2}
  + \frac{(K^{(+)})^2 r^4}{4 \sqrt{N}}\left[(2+\xi\sqrt{2}) +
    \frac{(2+\xi\sqrt{2})^2}{\sqrt{N}} \right].
  \label{eq:eigCC^T-bound} 
\end{equation}
The proof of the bound on the Frobenius norm is delayed until Section
\ref{sec:pure-curvature}. \\

\begin{remark}
  A different (possibly tighter) bound may be
  derived using Theorem \ref{thm:Tropp}.  However, such a bound would
  hold with a probability that becomes small when the ambient dimension
  $D$ is large.  We therefore proceed with the bound
  \eqref{eq:eigCC^T-bound} above, noting that we sacrifice no additional
  probability by using it here since it is required for the analysis in
  Section \ref{sec:pure-curvature}.
\end{remark}

\subsection{Noise Eigenvalues}
We may control the eigenvalues of
$\frac{1}{N}(E-\overline{E})(E-\overline{E})^T$ using standard results
from random matrix theory.  In particular, let $s_{\min}(E)$ and
$s_{\max}(E)$ denote the smallest and largest singular value of matrix
$E$, respectively.  The following result (Corollary 5.35 of
\cite{Vershynin}) gives a tight control on the size of $s_{\min}(E)$
and $s_{\max}(E)$ when $E$ has Gaussian entries.

\begin{theorem}[\cite{Vershynin,Edelman88}] \label{thm:GaussianSVs} Let
  $A$ be a $D \times N$ matrix whose entries are independent standard
  normal random variables.  Then for every $t \geq 0$, with
  probability at least $1-2\exp(-t^2/2)$ one has
  \begin{equation}
    \sqrt{N} - \sqrt{D} - t ~\leq~ s_{\min}(A) ~\leq~ s_{\max}(A) ~\leq~ \sqrt{N} + \sqrt{D} + t.
  \end{equation}
\end{theorem}
\noindent Define
\begin{equation}
  \alpha = \left(\sigma\sqrt{1-\frac{1}{N}}\right)^{-1}
  \label{alpha}
\end{equation}
and note that the entries of $Z = \alpha(E-\overline{E})$ are
independent standard normal random variables.
\textcolor{black}{This normalization by $\alpha$
  allows us to use Theorem \ref{thm:GaussianSVs}
  and we divide by $\alpha^2$ to recover the result for $(E-\overline{E})(E-\overline{E})^T$.}
Let us partition the
Gaussian vector $e$ into the first $d$ coordinates, $e_1$, and last
$D-d$ coordinates, $e_2$,
\begin{equation}
  e =
  \begin{bmatrix}
    e_1 & |& e_2
  \end{bmatrix}
  ^T,
  \label{e1e2}
\end{equation}
and observe that the matrix $U_1^T \frac{1}{N}
(E-\overline{E})(E-\overline{E}^T )U_1$ only depends on the
realizations of $e_1$. Similarly, the matrix $U_2^T \frac{1}{N}
(E-\overline{E})(E-\overline{E}^T )U_2$ only depends on the
realizations of $e_2$.  By Theorem \ref{thm:GaussianSVs}, we have
\begin{equation}
  \lambda_{\max}  \left(\frac{1}{N} U_1^T
    (E-\overline{E})(E-\overline{E}^T )U_1\right)  \leq   
  \textcolor{black}{\sigma^2\left(1+ \frac{5}{2\sqrt{N}} (\sqrt{d} + \xi_{e_1})\right)}
  \label{eq:EE-eig11}
\end{equation}
with probability at least $1-e^{-\xi_{e_1}^2}$ over the random
realization of $e_1$, as soon as $N >
4(\sqrt{d} + \xi_{e_1})$, a condition easily satisfied for any reasonable
sampling density. Similarly,
\begin{equation}
  \lambda_{\max}\left(\frac{1}{N} U_2^T (E-\overline{E})(E-\overline{E}^T )U_2\right) \leq
  \textcolor{black}{\sigma^2\left(1+ \frac{5}{2\sqrt{N}}  (\sqrt{D-d} + \xi_{e_2})\right)}
  \label{eq:EE-eig22}
\end{equation}
with probability at least $1-e^{-\xi_{e_2}^2}$ over the random
realization of $e_2$, as soon as $N > 4(\sqrt{D-d} + \xi_{e_2})$.

\section{Products of Bounded Random Matrices}
\subsection{Curvature Term: $CC^T$ \label{sec:pure-curvature}}
Begin by recalling the notation used for the curvature constants,
\begin{equation}
  K_i = \sum_{n=1}^d \kappa_n^{(i)}, \quad
  K = \left(\sum_{i=d+1}^D K_i^2\right)^{\frac{1}{2}} \negthickspace, \quad
  K_{nn}^{ij} = \sum_{n=1}^d \kappa_n^{(i)} \kappa_n^{(j)},  \quad
  K_{mn}^{ij} = \sum_{\substack{m,n=1\\m\neq n}}^d \kappa_m^{(i)} \kappa_n^{(j)}.
  \label{eq:Kieq:K}
\end{equation}
The constant $K_i$ quantifies the curvature in \textcolor{black}{normal direction} $i$, for $i
= (d+1),$ $\dots,D$.  The overall compounded curvature of the local
model is quantified by $K$ and is a natural result of our use of the
Frobenius norm.  We note that $K_iK_j = K_{nn}^{ij} + K_{mn}^{ij}$.
We also recall the positive constants
\begin{equation}
  K_{i}^{(+)} = \left(\sum_{n=1}^d |\kappa_n^{(i)}|^2\right)^{\frac12}, \quad \text{and} \quad K^{(+)} = \left(\sum_{i=d+1}^D (K_i^{(+)})^2\right)^{\frac{1}{2}}.
\end{equation}

Our strategy for bounding the matrix norm $\|U_2^T
\frac{1}{N}(C-\overline{C})(C-\overline{C})^T U_2\|_F$ begins with the
observation that $\frac{1}{N}(C-\overline{C})(C-\overline{C})^T$ is a
sample mean of $N$ covariance matrices of the vectors
$(c^{(k)}-\overline{c})$, $k=1,\dots,N$.  That is,
\begin{equation}
  \frac{1}{N}(C-\overline{C})(C-\overline{C})^T = \wE{(c-\wE{c})(c-\wE{c})^T}.
\end{equation}
We therefore expect that
$\frac{1}{N}(C-\overline{C})(C-\overline{C})^T$ converges toward the
centered covariance matrix of $c$.  We will use the following result
of Shawe-Taylor and Cristianini \cite{shawe} to bound, with high
probability, the norm of the difference between this sample mean and
its expectation.
\begin{theorem} \label{thm:shawe} \emph{(Shawe-Taylor \& Cristianini,
    \cite{shawe})}.\quad Given $N$ realizations  of a random matrix
  $Y$ distributed  with probability distribution $\prob_Y$, we have
  \begin{equation}
    \prob_Y \left\{ \left\|\E[Y] - \wE{Y}\right\|_F \leq
      \frac{R}{\sqrt{N}}\left(2+\xi\sqrt{2}\right) \right\} \ge 1-e^{-\xi^2}.
    \label{eq:shawe-taylor}
  \end{equation}
  The constant $R = \sup_{\substack{\text{supp}(\prob_Y)}}\|Y\|_F$, where
  \text{supp}$(\prob_Y)$ is the support of distribution $\prob_Y$.
\end{theorem}
\noindent We note that the original formulation of the result involves
only random vectors, but since the Frobenius norm of a matrix is
merely the Euclidean norm of its vectorized version, we formulate the
theorem in terms of matrices. We also note that the choice of $R$ in
\eqref{eq:shawe-taylor} need not be unique.  Our analysis will proceed
by using upper bounds for $\|Y\|_F$ which may not be suprema.
Let 
\begin{equation*}
  R_{c} = \sup_{\textcolor{black}{c}} \|U_2^Tc\|_F.
\end{equation*}
Using Theorem \ref{thm:shawe} and modifying slightly the proof of
Corollary 6 in \cite{shawe}, which uses standard inequalities, we arrive at
\begin{align*}
  \bigg| \left\|\E[U_2^T(c - \E[c])(c-\E[c])^TU_2]\right\|_F -
  \left\|\wE{U_2^T(c-\wE{c})(c-\wE{c})^TU_2}\right\|_F \bigg| &
  \nonumber \\
  \leq \left\|\E[U_2^Tcc^TU_2] - \wE{U_2^Tcc^TU_2}\right\|_F
  \negthickspace + \left\|\E[U_2^Tc] - \wE{U_2^Tc}\right\|_F^2 \leq
  \frac{R_{c}^2}{\sqrt{N}}\left(2+\xi_{c}\sqrt{2}\right) & +
  \frac{R_{c}^2}{N}\left(2+\xi_{c}\sqrt{2}\right)^2
\end{align*}
with probability greater than $1-2e^{-\xi_{c}^2}$ over the random selection of the sample points.  
To complete the bound we must compute $R_c$ and $\left\|\E[U_2^T(c - \E[c])(c-\E[c])^TU_2]\right\|_F$.
A simple norm calculation shows
\begin{align*}
  \|U_2^Tc\|_F^2 =
  \frac{1}{4}\sum_{i=d+1}^D\left(\kappa_1^{(i)}\ell_1^2 + ~\dots~ +
    \kappa_d^{(i)}\ell_d^2\right)^2 \leq
  \frac{r^4}{4}\sum_{i=d+1}^D \left(K_i^{(+)}\right)^2
  = \frac{(K^{(+)})^2 r^4}{4},
\end{align*}
and we set $R_c = K^{(+)}r^2/2$.
Next, the expectation takes the form
\begin{equation*}
  \Big\|\E[U_2^T(c - \E[c])(c-\E[c])^TU_2]\Big\|_F ~=~
  \Big\|\E[U_2^Tcc^TU_2] - \E[U_2^Tc]\E[c^TU_2] \Big\|_F.
\end{equation*}
\textcolor{black}{We calculate}
\begin{equation*}
  \E[c_ic_j] = \frac{\left[3K_{nn}^{ij} + K_{mn}^{ij}\right]r^4}{4(d+2)(d+4)},\quad\text{and}\quad
  \E[c_i]\E[c_j] = \frac{\left[K_{nn}^{ij} +
      K_{mn}^{ij}\right]r^4}{4(d+2)^2}
\end{equation*}
\textcolor{black}{and compute the norm}
\begin{align*}
  \Big\|\E[U_2^T(c - \E[c])(c-\E[c])^TU_2]\Big\|_F =
  \frac{r^4}{2(d+2)^2(d+4)}
  \sqrt{\sum_{i,j=d+1}^D\left[(d+1)K_{nn}^{ij}-K_{mn}^{ij}\right]^2}.
\end{align*}
Finally, putting it all together, we conclude that
\begin{align*}
  \left\|U_2^T \frac{1}{N}(C-\overline{C})(C-\overline{C})^T U_2 \right\|_F \leq ~
  &\frac{r^4}{2(d+2)^2(d+4)}
  \sqrt{\sum_{i,j=d+1}^D \left[(d+1)K_{nn}^{ij}-K_{mn}^{ij}\right]^2}
  \nonumber \\ 
  & + \frac{1}{\sqrt{N}}\frac{(K^{(+)})^2
    r^4}{4}\left[\left(2+\xi_{c}\sqrt{2}\right) +
    \frac{1}{\sqrt{N}}\left(2+\xi_{c}\sqrt{2}\right)^2   \right] 
\end{align*}
with probability greater than $1-2e^{-\xi_{c}^2}$ over the random
selection of the sample points. 

\subsection{Curvature-Linear Cross-Terms: $CL^T$}
Our approach for bounding the matrix norm $\|U_2^T
\frac{1}{N}(C-\overline{C})(L-\overline{L})^T U_1\|_F$ mirrors that of
Section \ref{sec:pure-curvature}. Here, we use that $\E[\ell_i] = 0$
for $1 \leq i \leq d$ and proceed as follows.  We have
\begin{equation*}
  R_{\ell} = \sup_\ell \|\ell^T U_1\|_F = r.
\end{equation*}
Reasoning as in the previous section, we have
\begin{align*}
  \bigg| \left\|\E[U_2^T(c - \E[c])(\ell-\E[\ell])^TU_1]\right\|_F -
  \left\|\wE{U_2^T(c-\wE{c})(\ell-\wE{\ell})^TU_1}\right\|_F \bigg| &
  \nonumber \\ 
  \leq \left\|\E[U_2^Tc\ell^TU_1] - \wE{U_2^Tc\ell^TU_1}\right\|_F +
  \left\|\wE{\ell^TU_1} - \E[\ell^TU_1]\right\|_F
  \bigg(\left\|\wE{U_2^Tc} - \E[U_2^Tc]\right\|_F   + &
  \left\|\E[U_2^Tc]\right\|_F\bigg) \nonumber \\
  \leq
  \frac{R_{c}R_{\ell}}{\sqrt{N}}\left(2+\xi_{c\ell}\sqrt{2}\right) +
  \frac{R_{\ell}}{\sqrt{N}}\left(2+\xi_{\ell}\sqrt{2}\right)
  \bigg[\frac{R_{c}}{\sqrt{N}}\left(2+\xi_{c}\sqrt{2}\right) +
  \left\|\E[U_2^Tc]\right\|_F \bigg]& 
\end{align*}
with probability greater than
$1-e^{-\xi_{c\ell}^2}-e^{-\xi_{\ell}^2}-e^{-\xi_{c}^2}$ over the
random selection of the sample points. Finally, we set $\xi_{\ell} =
\xi_{c\ell}$ and conclude
\begin{align*}
  \left\|U_2^T \frac{1}{N}(C-\overline{C})(L-\overline{L})^T U_1
  \right\|_F \leq 
  \frac{K^{(+)}r^3}{2 \sqrt{N}}\left[\frac{d+3}{d+2}\left(2 +
      \xi_{c\ell}\sqrt{2}\right) + \frac{1}{\sqrt{N}}\left(2 +
      \xi_{c}\sqrt{2}\right)\left(2 +
      \xi_{c\ell}\sqrt{2}\right)\right]
\end{align*}
with probability greater than $1-2e^{-\xi_{c\ell}^2}-e^{-\xi_{c}^2}$
over the random selection of the sample points.
\section{Products of Unbounded Random Matrices: $EE^T$}
We seek bounds for the matrix norms of the form 
\begin{equation}
  \left\|U_n^T
    \frac{1}{N}(E-\overline{E})(E-\overline{E})^T U_m\right\|_F \quad 
  \text{for $(n,m) = (1,1), (2,2)$, and $(2,1)$}.
  \label{noise-noise}
\end{equation}
Because $E$ is composed of $N$ columns of independent realizations of
a $D$-dimensional Gaussian vector, the matrix $A$ defined by
\begin{equation*}
  A = \frac{1}{N-1}\frac{1}{\sigma^2}(E-\overline{E})(E-\overline{E})^T
  = \frac{\alpha^2}{N}(E-\overline{E})(E-\overline{E})^T
\end{equation*}
is Wishart $W_D\left(N-1,\frac{1}{N-1}I_D\right)$, where $\alpha =
\left(\sigma\sqrt{1-\frac{1}{N}}\right)^{-1}\negthickspace$. As a
result, we can quickly compute bounds on the terms (\ref{noise-noise})
since they can be expressed as the norm of blocks of $A$. Indeed, let
us partition $A$ as follows
\begin{equation*}
  A =
  \begin{bmatrix}
    A_{11} & A_{12} \\ A_{21} & A_{22}
  \end{bmatrix},
\end{equation*}
where $A_{11}$ is $d \times d$, $A_{22}$ is $(D-d) \times (D-d)$.  We
now observe that $A_{nm}$ is not equal to $\frac{\alpha^2}{N} U_n^T
(E-\overline{E})(E-\overline{E})^T U_m$, but both matrices
have the same Frobenius norm. Precisely, the two matrices differ only by
a left and a right rotation, as explained in the next few lines.

Since only the first $d$ entries of each column in $U_1$ are nonzero,
we can define two matrices $P_1$ and $Q_1$ that extract the first $d$
entries and apply the rotation associated with $U_1$, respectively, as
follows
\begin{equation*}
  U_1 = 
  \begin{bmatrix}
    & & \\
    & Q_1 & \\
    & & \\
    0 & & 0\\
    \vdots & & \vdots\\
    0 & & 0
  \end{bmatrix} 
  =
  \begin{bmatrix}
    1 & & 0\\
    & \diagdown& \\
    0 & & 1 \\
    0 & & 0\\
    \vdots & & \vdots\\
    0 & & 0
  \end{bmatrix} 
  Q_1
  =
  P_1 Q_1.
\end{equation*}
We define similar matrices  $P_2$ and $Q_2$ such that $U_2 = P_2 Q_2$.
We conclude that
\begin{equation*}
  \left\|U_n^T (E-\overline{E})(E-\overline{E})^T U_m\right\|_F 
  = \left\| P_n^T (E-\overline{E})(E-\overline{E})^T P_m \right\|_F 
  = \frac{N}{\alpha^2} \left\| A_{nm} \right\|_F.
\end{equation*}
In summary, we can control the size of the norms (\ref{noise-noise})
by controlling the norm of the sub-matrices of a Wishart matrix. We
first estimate the size of $\|A_{11}\|_F$ and $\|A_{22}\|_F$. This is
a straightforward affair, since we can apply Theorem
\ref{thm:GaussianSVs} with $P_1 \textcolor{black}{(E-\overline{E})}$ and $P_2 \textcolor{black}{(E-\overline{E})}$, respectively, to get
the spectral norm of $A_{11}$ and $A_{22}$. We then apply a standard
inequality between the spectral and the Frobenius norm of a matrix $M$,
\begin{equation}
  \|M\|_F \leq \sqrt{\rank{M}}\|M\|_2.
  \label{eq:frobenius<spectral}
\end{equation}
This bound is usually quite loose and equality is achieved only for
the case where all singular values of matrix $A$ are equal. It turns
out that this special case holds in expectation for the matrices in
the analysis to follow, and thus \eqref{eq:frobenius<spectral}
provides a tight estimate of the Frobenius norm.  Using
\eqref{eq:EE-eig11} and \eqref{eq:frobenius<spectral}, we have the
following bound
\begin{equation*}
  \left\| U_1^T \frac{1}{N}(E-\overline{E})(E-\overline{E})^T
    U_1\right\|_F   \leq \sigma^2\sqrt{d}
  \left[1 + \frac{5}{2}\frac{1}{\sqrt{N}}(\sqrt{d} + \xi_{e_1}\sqrt{2}) \right]
\end{equation*}
with probability greater than $1-e^{-\xi_{e_1}^2}$ over the random realization of the noise.
By \eqref{eq:EE-eig22}, we also have
\begin{equation*}
  \left\|U_2^T \frac{1}{N}  (E-\overline{E})(E-\overline{E})^T 
    U_2\right\|_F
  \leq \sigma^2\sqrt{D-d}
  \left[1 + \frac{5}{2}\frac{1}{\sqrt{N}}(\sqrt{D-d} + \xi_{e_2}\sqrt{2}) \right]
\end{equation*}
with probability greater than $1-e^{-\xi_{e_2}^2}$ over the random realization of the noise.

It remains to bound $\|A_{21}\|$. Here we proceed by conditioning on
the realization of the last $D-d$ coordinates of the noise vectors in
the matrix $E$; in other words, we freeze $P_2E$. Rather than working
with Gaussian matrices, we prefer to vectorize the matrix $A_{21}$ and
define
\begin{equation*}
  \vect{a_{21}} = \vct{A^T_{21}}.
\end{equation*}
Note that here we unroll the matrix $A_{21}$ row by row to build
$\vect{a_{21}}$.  Because the Frobenius norm of $A_{21}$ is the
Euclidean norm of $\vect{a_{21}}$, we need to find a bound on
$\|\vect{a_{21}}\|$.  Conditioning on the realization of $P_2E$, we
know (Theorem 3.2.10 of \cite{Muirhead}) that the distribution of
$\vect{a_{21}}$ is a multivariate Gaussian variable
$\mathcal{N}(\vect{0},S)$, where $\vect{0}$ is the zero vector of
dimension $d(D-d)$ and $S$ is the $d(D-d) \times d(D-d)$ block
diagonal matrix containing $d$ copies of $\frac{1}{N}A_{22}$
\begin{equation*}
  S = \frac{1}{N}
  \begin{bmatrix}
    A_{22} & & & \\
    & A_{22} & & \\
    & & \ddots & \\
    & & & A_{22}
  \end{bmatrix}.
\end{equation*}
Let $S^\dagger$ be a generalized inverse of $S$ (such that $S
S^\dagger S = S$), then (see e.g. Theorem 1.4.4 of \cite{Muirhead})
\begin{equation*}
  \vect{a_{21}}^T S^\dagger\vect{a_{21}} \sim \chi^2(\rank{S}).
\end{equation*}
\textcolor{black}{Now, using only the bound for the smallest singular value in Theorem \ref{thm:GaussianSVs},
  $A_{22}$ has full rank, $(D-d)$, with probability $1-e^{-(\sqrt{N}-\sqrt{D-d})^2/2}$, and therefore $S$ has full rank, $d(D-d)$, with the same probability.}
In the following, we derive an upper bound on the size of
$\|\vect{a_{21}}\|$ when $A_{22}$ has full rank.  A similar -- but
tighter -- bound can be derived when $S$ is rank deficient; we only
need to replace $(D-d)$ by the rank of $A_{22}$ in the bound that follows. Because the bound
derived when $A_{22}$ is full rank will hold when $A_{22}$ is rank
deficient (an event which happens with very small probability,
anyway), we only worry about this case in the following. In this case,
$S^\dagger = S^{-1}$ and
\begin{equation*}
  \vect{a_{21}}^T S^{-1} \vect{a_{21}} \sim \chi^2(d(D-d)).
\end{equation*}
Finally, using a corollary of Laurant and Massart
(immediately following Lemma 1 of \cite{LaurantMassart}), we get that,
\begin{equation}
  \vect{a_{21}}^T S^{-1} \vect{a_{21}} \leq d(D-d) +
  2\xi_{e_3}\sqrt{d(D-d)} + 2\xi_{e_3}^2
  \label{norm-chi}
\end{equation}
with probability greater than $1-e^{-\xi_{e_3}^2}$.  
In the following, we assume that $\xi_{e_3} \leq 0.7
\sqrt{d(D-d)}$, which happens as soon as $d$ or $D$ have a moderate
size. Under this mild assumption we have
\begin{equation*}
  \sqrt{
    d(D-d) + 2\xi_{e_3}\sqrt{d(D-d)} + 2\xi_{e_3}^2
  }
  \leq \sqrt{d(D-d)} 
  \left( 1 + \frac{6}{5}  \frac{\xi_{e_3}}{\sqrt{d(D-d)}}
  \right).
\end{equation*}
In order to compare $\|\vect{a_{21}}\|^2$ to $\vect{a_{21}}^T S^{-1}
\vect{a_{21}}$, we compute the eigendecomposition of $S$,
\begin{equation*}
  S = O~\varPi~O^T,
\end{equation*}
where $O$ is a unitary matrix and $\varPi$ contains the eigenvalues of
$\frac{1}{N}A_{22}$, repeated $d$ times.  Letting $\lambda_{\max}
\left(\frac{1}{N} A_{22}\right)$ be the largest eigenvalue of
$\frac{1}{N} A_{22}$, we get the following upper bound,
\begin{equation*}
  \|\vect{a_{21}}\|^2  \leq 
  \lambda_{\max} \left(\frac{1}{N} A_{22}\right) 
  \vect{a_{21}}^T\;O^T \varPi^{-1} O\;\vect{a_{21}} 
  = \lambda_{\max} \left(\frac{1}{N} A_{22}\right) 
  \vect{a_{21}}^T S^{-1} \vect{a_{21}}.
\end{equation*}
We conclude that, conditioned on a realization of the last $D-d$
entries of $E$, we have 
\begin{equation}
  \prob_{e_1} \left\{
    \|\vect{a_{21}}\| \leq 
    \sqrt{\lambda_{\max} \left(\frac{1}{N} A_{22}\right) }
    \sqrt{d(D-d)}\left[1 +
      \frac{6}{5}\frac{\xi_{e_3}}{\sqrt{d(D-d)}}\right]
    | e_2\right\} \ge 1-e^{-\xi_{e_3}^2}.
  \label{conditioned-wishart}
\end{equation}
To derive a bound on $\|\vect{a_{21}}\|$ that holds with high
probability, we consider the event
\begin{equation*}
  \mathscr{E}_{\varepsilon,\xi}= 
  \left\{
    \|\vect{a_{21}}\| \leq 
    \frac{\sqrt{d(D-d)} }{\sqrt{N}} 
    \left(1 + \frac{\sqrt{D-d} + \varepsilon}{\sqrt{N}} \right)
    \left[1 + \frac{6}{5}\frac{\xi}{\sqrt{d(D-d)}}\right]
  \right\}.
\end{equation*}
As we will see in the following, the event
$\mathscr{E}_{\varepsilon,\xi}$ happens with high probability. This
event depends on the random realization of the top $d$ coordinates,
$e_1$, of the Gaussian vector $e$ (see \eqref{e1e2}). Let us define a second likely
event, which depends only on $e_2$ (the last $D-d$ coordinates of $e$),
\begin{equation*}
  \mathscr{E}_{e_2}= 
  \left\{
    \sqrt{\lambda_{\max}\left(\frac{1}{N}A_{22}\right)}
    \leq 
    \frac{1}{\sqrt{N}} 
    \left(1 + \frac{\sqrt{D-d} + \varepsilon}{\sqrt{N}} \right)
  \right\}.
\end{equation*}
Theorem \ref{thm:GaussianSVs} tells us that the event
$\mathscr{E}_{e_2}$ is very likely, and
$\prob_{e_2}{(\mathscr{E}_{\textcolor{black}{e_2}}^c)} \leq e^{-\varepsilon^2/2}$.  We now
show that the probability of $\mathscr{E}_{\varepsilon,\xi}^c$ is also
very small,
\begin{equation*}
  \prob_{e_1,e_2} (\mathscr{E}_{\varepsilon,\xi}^c) =
  \prob_{e_1,e_2} (\mathscr{E}_{\varepsilon,\xi}^c \cap \mathscr{E}_{\textcolor{black}{e_2}})  + 
  \prob_{e_1,e_2} (\mathscr{E}_{\varepsilon,\xi}^c \cap \mathscr{E}^c_{\textcolor{black}{e_2}})  
  \leq
  \prob_{e_1,e_2} (\mathscr{E}_{\varepsilon,\xi}^c \cap \mathscr{E}_{\textcolor{black}{e_2}})  + 
  \prob_{e_2} (\mathscr{E}^c_{\textcolor{black}{e_2}}).  
\end{equation*}
In order to bound the first term, we condition on $e_2$, 
\begin{equation*}
  \prob_{e_1,e_2} (\mathscr{E}_{\varepsilon,\xi}^c \cap \mathscr{E}_{\textcolor{black}{e_2}})  =
  \E_{e_2} \left[\prob_{e_1} (\mathscr{E}_{\varepsilon,\xi}^c \cap \mathscr{E}_{\textcolor{black}{e_2}}| e_2)\right]. 
\end{equation*}
Now the two conditions,
\begin{equation*}
  \begin{cases}
    \|\vect{a_{21}}\| >
    \sqrt{d(D-d)} \frac{1}{\sqrt{N}} 
    \left(1 + \frac{\sqrt{D-d} + \varepsilon}{\sqrt{N}} \right)
    \left[1 + \frac{6}{5}\frac{\xi}{\sqrt{d(D-d)}}\right] \\
    \frac{1}{\sqrt{N}} 
    \left(1 + \frac{\sqrt{D-d} + \varepsilon}{\sqrt{N}} \right)
    \ge   \sqrt{\lambda_{\max}\left(\frac{1}{N}A_{22}\right)}
  \end{cases}
\end{equation*}
imply that
\begin{equation*}
  \|\vect{a_{21}}\| >
  \sqrt{d(D-d)}
  \sqrt{\lambda_{\max}\left(\frac{1}{N}A_{22}\right)}
  \left[1 + \frac{6}{5}\frac{\xi}{\sqrt{d(D-d)}}\right] ,
\end{equation*}
and thus
\begin{equation*}
  \prob_{e_1} (\mathscr{E}_{\varepsilon,\xi}^c \cap \mathscr{E}_{\textcolor{black}{e_2}}| e_2) \leq
  \prob_{e_1} \left(
    \|\vect{a_{21}}\| >
    \sqrt{d(D-d)}
    \sqrt{\lambda_{\max}\left(\frac{1}{N}A_{22}\right)}
    \left[1 + \frac{6}{5}\frac{\xi}{\sqrt{d(D-d)}}\right] 
    \lvert e_2 \right).
\end{equation*}
Because of (\ref{conditioned-wishart}) the probability on the right-hand
side is less than $e^{-\xi^2}$, which does not depend on
$e_2$. We conclude that 
\begin{equation*}
  \prob_{e_1,e_2} (\mathscr{E}_{\varepsilon,\xi}^c)
  \leq e^{-\varepsilon^2/2} + e^{-\xi^2}.
\end{equation*}
Finally, since 
\begin{equation*}
  \left\|U_2^T \frac{1}{N}(E-\overline{E})(E-\overline{E})^T
    U_1\right\|_F = \sigma^2 \left( 1 -\frac{1}{N}\right) \|A_{21}\|_F
  =  \sigma^2 \left( 1 -\frac{1}{N}\right)  
  \|\vect{a_{21}}\|,
\end{equation*}
we have 
\begin{equation*}
  \left\| \frac{1}{N}U_2^T(E-\overline{E})(E-\overline{E})^T
    U_1\right\|_F \leq \frac{\sigma^2 \sqrt{d(D-d)} }{\sqrt{N}} 
  \left(\negthinspace 1 + \frac{\sqrt{D-d} + \xi_{e_2}\sqrt{2}}{\sqrt{N}} \right) \negthickspace
  \left[1 + \frac{6}{5}\frac{\xi_{e_3}}{\sqrt{d(D-d)}}\right]
\end{equation*}
with probability greater than $1 - \textcolor{black}{e^{-\xi_{e_2}^2}} - e^{-\xi_{e_3}^2}$
over the realization of the noise.

\section{Products of Bounded and Unbounded Random Matrices}

\subsection{Linear-Noise Cross-Terms: $EL^T$}
Our goal is to bound the matrix norm
$\frac{1}{N}\|U_m^T(E-\overline{E})(L-\overline{L})^TU_1\|_F$, with
high probability, for $m =\{1,2\}$.  We detail the analysis for the
case where $m=1$ and note that the analysis for $m=2$ is identical up
to the difference in dimension. Using the decomposition of the matrix
$U_1 = P_1 Q_1$ defined in the previous section, we have
\begin{equation}
  \frac{1}{N}
  \left\|U_1^T (E-\overline{E})(L-\overline{L})^T U_1\right\|_F 
  =
  \frac{1}{N}
  \left\|P_1^T (E-\overline{E})(L-\overline{L})^T  P_1 \right\|_F.
  \label{U1Q1}
\end{equation}
Before proceeding with a detailed analysis of this term, let us
\textcolor{black}{derive a bound, which will prove to be very precise, using a back
  of the envelope analysis.} The entry $(i,j)$ in the matrix 
$\frac{1}{N}P_1^T(E-\overline{E})(L-\overline{L})^TP_1$ is given by
\begin{equation*}
  \frac{1}{N}\sum_{k=1}^N 
  (e^{(k)}_i  - \overline{e}_i)
  (\ell^{(k)}_j  - \overline{\ell}_j),
\end{equation*}
and it measures the average correlation between coordinate $i \leq d$
of the (centered) noise term and coordinate $j \leq d$ of the
linear tangent term. Clearly, this empirical correlation has zero
mean, and an upper bound on its variance is given by
\begin{equation*}
  \frac{1}{N} \sigma^2 \lambda_1,
\end{equation*}
where the top eigenvalue $\lambda_1$ measures the largest variance of
the random variable $\ell$, measured along the first column of
$U_1$. Since the matrix $P_1^T (E-\overline{E})(L-\overline{L})^T P_1$
is $d \times d$, we expect
\begin{equation*}
  \frac{1}{N} \left\|P_1^T (E-\overline{E})(L-\overline{L})^T P_1
  \right\|_F \approx
  \frac{1}{\sqrt{N}} \sigma \sqrt{\lambda_1} d.
\end{equation*}
We now proceed with the rigorous analysis. The singular value
decomposition of $P_1^T(L- \overline{L})$ is given by
\begin{equation}
  P_1^T(L- \overline{L}) = Q_1 \Sigma V^T,
\end{equation}
where $\Sigma$ is the $d \times d$ matrix of the singular values, and
$V$ is a matrix composed of $d$ orthonormal column vectors of size
$N$. Injecting the SVD of $P_1^T(L - \overline{L})$ we have
\begin{align}
  \frac{1}{N}
  \left\|P_1^T (E-\overline{E})(L-\overline{L})^T  P_1 \right\|_F 
  & = \frac{1}{N} \left\|P_1^T (E-\overline{E})V \Sigma Q^T_1 \right\|_F 
  = \frac{1}{N}\left\|P_1^T (E-\overline{E})V \Sigma\right\|_F \nonumber
  \\
  & \leq \frac{\sqrt{\lambda_1}}{\sqrt{N}}
  \left\|P_1^T (E-\overline{E})V \right\|_F.
  \label{LE-EV}
\end{align}
Define
\begin{equation*}
  Z_1 = \alpha P_1^T (E - \overline{E}) V.
\end{equation*}
Each row of $Z_1$ is formed by the projections of the corresponding
row of $\alpha P_1^T (E - \overline{E})$ onto the $d$-dimensional
subspace of $R^N$ formed by the columns of $V$. As such, the projected
row is a $d$-dimensional Gaussian vector, the norm of which scales
like $\sqrt{d}$ with high probability. 

The only technical difficulty involves the fact that the columns of
$V$ change with the different realizations of $L$. We need to check
that this random rotation of the vectors in $V$ does not affect the
size of the norm of $Z_1$. Proceeding in two steps, we first freeze a
realization of $L$, and compute a bound on $ \left\|P_1^T
  (E-\overline{E})V \right\|_F$ that does not depend \textcolor{black}{on} $L$. We then
remove the conditioning on $L$, and compute the probability that
$\|Z_1\|_F$ \textcolor{black}{is} very close to $d$.

Instead of working with $Z_1$, we define  the $d^2$-dimensional vector
\begin{equation*}
  \vect{z_1} = \vct{Z_1^T}.
\end{equation*}
Consider the $N\!d$-dimensional Gaussian vector 
\begin{equation*}
  \vect{g_1} = \alpha \; \vct{P_1^T(E- \overline{E})} \sim \mathcal{N}(0,I_{N\!d}).
\end{equation*}
In the next few lines, we construct an orthogonal projector
$\mathscr{P}$ such that $\vec{z}_1 = \mathscr{P} \vec{g}_1$. As a
result, we will have that $\vect{z_1} \sim \mathcal{N}(0,I_{d^2})$,
and using standard results on the concentration of the Gaussian
measure, we will get an estimate
of $\|P_1^T(E-\overline{E})V\|_F = \alpha^{-1} \|\vec{z}_1\|$. \\

First, consider the following $d^2 \times Nd$ matrix
\begin{equation*}
  \mathscr{V}=
  \begin{bmatrix}
    V^T & & & \\
    & V^T &  & \\
    & & \ddots & \\
    & & &V^T
  \end{bmatrix},
\end{equation*}
formed by stacking $d$ copies of $V^T$ in a block diagonal fashion with no overlap
(note that $V^T$ is not a square matrix).
We observe that because no overlap exists between the blocks,
the rows of $\mathscr{V}$ are orthonormal and $\mathscr{V}$
is an orthogonal projector from $\R^{N\negthinspace d}$ to $\R^{d^2}$.

Now, we consider the $N\!d \times N \! d$ permutation matrix $\Omega$
constructed as follows. We first construct the $d \times Nd$ matrix
$\Omega_1$ by interleaving blocks of zeros of size $d \times (N-1)$ between
the columns vectors of the $d \times d$ identity matrix,
\begin{equation*}
  \Omega_1=
  \begin{bmatrix}
    \begin{vmatrix}
      1 \\
      0\\
      \vdots\\
      0
    \end{vmatrix}
    \begin{matrix}
      0 & \cdots &0 \\
      0 & \cdots &0 \\
      \vdots\\
      0 & \cdots &0 \\
    \end{matrix}
    \begin{vmatrix}
      0 \\
      1\\
      \vdots\\
      0
    \end{vmatrix}
    \begin{matrix}
      0 & \cdots &0 \\
      0 & \cdots &0 \\
      \vdots\\
      0 & \cdots &0 \\
    \end{matrix}
    \quad \cdots \quad
    \begin{vmatrix}
      0 \\
      0\\
      \vdots\\
      1
    \end{vmatrix}
    \begin{matrix}
      0 & \cdots &0 \\
      0 & \cdots &0 \\
      \vdots\\
      0 & \cdots &0 \\
    \end{matrix}
  \end{bmatrix}.
\end{equation*}
Now consider the matrix $\Omega_2$ obtained by performing a circular
shift of the columns $\Omega_1$ to the right by one index,
\begin{equation*}
  \Omega_2=
  \begin{bmatrix}
    \begin{matrix}
      0 \\
      0\\
      \vdots\\
      0
    \end{matrix}
    \begin{vmatrix}
      1 \\
      0\\
      \vdots\\
      0
    \end{vmatrix}
    \begin{matrix}
      0 &\cdots & 0 \\
      0 & \cdots & 0 \\
      & \vdots& \\
      0 & \cdots& 0 \\
    \end{matrix}
    \begin{vmatrix}
      0 \\
      1\\
      \vdots\\
      0
    \end{vmatrix}
    \begin{matrix}
      0 & \cdots &0 \\
      0 & \cdots &0 \\
      \vdots\\
      0 & \cdots &0 \\
    \end{matrix}
    \quad \cdots\quad
    \begin{matrix}
      0\\
      0\\
      \vdots\\
      0\\
    \end{matrix}
    \begin{vmatrix}
      0 \\
      0\\
      \vdots\\
      1
    \end{vmatrix}
    \begin{matrix}
      0 & \cdots  \\
      0 & \cdots  \\
      \vdots\\
      0 & \cdots  \\
    \end{matrix}
  \end{bmatrix}.
\end{equation*}
We can iterate this process $N-1$ times and construct $N$ such matrices,
$\Omega_1,\ldots,\Omega_N$. Finally, we stack these $N$ matrices to
construct the $N\!d \times N\!d$ permutation matrix
\begin{equation*}
  \Omega =
  \begin{bmatrix}
    & & \Omega_1 & & \\
    & & \vdots &  & \\
    & & \Omega_N & & \\
  \end{bmatrix}.
\end{equation*}
By construction, $\Omega$ only contains a single nonzero entry, equal
to one, in every row and every column, and therefore is a permutation matrix.
Finally, the matrix $\Omega$ allows to move the action of $V$ from the
right of $E$ to the left, and we have
\begin{equation}
  \vect{z_1} = \mathscr{V} \Omega \vect{g_1}.
  \label{orthoproj1}
\end{equation}
Putting everything together, we conclude that the matrix defined by
\begin{equation}
  \mathscr{P} = \mathscr{V} \Omega
\end{equation}
is an orthogonal projector, and
therefore $\vect{z_1} \sim \mathcal{N}(0,I_{d^2})$. Using again the 
previous bound (\ref{norm-chi}) on the norm of a Gaussian vector, we
have 
\begin{equation}
  \prob_{e}\left(\|\vect{z_1} \| \leq \left(d + \frac{6}{5}\varepsilon\right)|
    L\right) \ge 1 - e^{-\varepsilon^2}. 
  \label{z1-conditioned}
\end{equation}
To conclude the proof, we remove the conditioning on $L$, and using
(\ref{z1-conditioned}) we have
\begin{equation*}
  \prob_{e,\ell}\left(\|\vect{z_1} \| \leq \left(d + \frac{6}{5}\varepsilon\right)\right) =
  \E_{\ell} \prob_{e}\left(\|\vect{z_1} \| \leq \left(d + \frac{6}{5}\varepsilon\right)| L\right) 
  \ge 1 - e^{-\varepsilon^2}.
\end{equation*}
Since $\|P_1^T(E - \overline{E})V\|_F= \alpha^{-1} \|\vect{z_1} \|$, we
have 
\begin{equation}
  \prob_{e,\ell}
  \left(
    \|P_1^T(E - \overline{E}) V \|_F \leq 
    \sigma\sqrt{1-\frac{1}{N}} \left(d + \frac{6}{5}\varepsilon\right)
  \right) 
  \ge 1 - e^{-\varepsilon^2}.
  \label{norm-z1}
\end{equation}
Finally, combining (\ref{eig-linear-bound}), (\ref{eig-linear-prob}),
(\ref{U1Q1}), (\ref{LE-EV}), and (\ref{norm-z1}) we conclude that
\begin{align}
  \prob_{e,\ell} &\left(
    \frac{1}{N}
    \left\|U_1^T (E-\overline{E})(L-\overline{L})^T  U_1 \right\|_F 
    \leq 
    \frac{\sigma \sqrt{\lambda_{\text{bound}}(\xi)}}{\sqrt{N}}
    \sqrt{1-\frac{1}{N}} \left(d + \frac{6}{5}\varepsilon\right)\right) \nonumber \\
  & \ge (1 - e^{-\varepsilon^2}) (1 - d e^{-\xi^2})
  \label{LE-lambda1}
\end{align}
which implies
\begin{align*}
  \prob_{e,\ell} &\left(
    \frac{1}{N}
    \left\|U_1^T (E-\overline{E})(L-\overline{L})^T  U_1 \right\|_F 
    \leq 
    \frac{\sigma \; r}{\sqrt{N}\sqrt{d+2}}
    \left[1+
      \xi  \frac{5\sqrt{d+2}}{\sqrt{N}}
    \right]
    \left(d + \frac{6}{5}\varepsilon\right)\right) \nonumber \\
  & \ge (1 - e^{-\varepsilon^2/2}) (1 - d e^{-\xi^2}).
\end{align*}
A similar bound holds for $\left\|U_2^T \frac{1}{N}
  (E-\overline{E})(L-\overline{L})^T U_1\right\|_F$. Indeed,
we define
\begin{equation}
  Z_2 = \alpha P_2^T (E - \overline{E}) V, \quad \quad 
  \vect{z_2} = \vct{Z_2}, \quad \text{and} \quad
  \vect{g_2} = \alpha \; \vct{P_2^T(E- \overline{E})}.  
\end{equation}
\textcolor{black}{
  Again, we can construct an orthogonal projector $\mathscr{P^\prime}$ with size $d(D-d) \times N(D-d)$ so that 
  \begin{equation}
    \vect{z_2} = \mathscr{P^\prime} \vect{g_2}.
    \label{orthoproj2}
  \end{equation}
}
By combining (\ref{orthoproj1}) and (\ref{orthoproj2}), we can control
the concatenated vector $\begin{bmatrix} \vect{z_1} &
  \vect{z_2}\end{bmatrix}^T$ by estimating the norm of $\begin{bmatrix} \vect{g_1} &
  \vect{g_2}\end{bmatrix}^T$. We conclude that 
\begin{equation}
  \begin{bmatrix}
    \\ \|U_1^T\frac{1}{N}(E-\overline{E})(L-\overline{L})^TU_1\|_F \\ \\
    \|U_2^T\frac{1}{N}(E-\overline{E})(L-\overline{L})^TU_1\|_F \\ \\
  \end{bmatrix}
  \leq
  \frac{\sigma \; r}{\sqrt{N}\sqrt{d+2}}
  \left[1+
    \xi_{\lambda_1}  \frac{5\sqrt{d+2}}{\sqrt{N}}
  \right]
  \begin{bmatrix}
    \\ d &+ & \frac{6}{5} \xi_{e\ell}\\ \\ \sqrt{d(D-d)} &+ &
    \frac{6}{5} \xi_{e\ell} \\ \\
  \end{bmatrix}
  \label{LE-All}
\end{equation}
with probability greater than
$(1-de^{-\xi_{\lambda_1}^2})(1-e^{-\xi_{e\ell}^2})$ over the joint
random selection of the sample points and random realization of the
noise.

\subsection{Curvature-Noise Cross-Terms: $CE^T$}
The analysis to bound the matrix norm
\begin{equation*}
  \frac{1}{N} \left \| U_2^T(C-\overline{C})(E-\overline{E})^TU_m\right \|_F =
  \frac{1}{N} \left\| U_m^T(E-\overline{E})(C-\overline{C})^TU_2 \right \|_F 
\end{equation*}
for $m =\{1,2\}$ proceeds in an identical manner to that for the bound
on $\| \frac{1}{N}U_m^T(E-\overline{E})(L-\overline{L})^TU_1\|_F$.  We
therefore give only a brief outline here.
Mimicking the reasoning that leads to (\ref{LE-lambda1}), we get
\begin{align*}
  \prob_{e,\ell} &\left(
    \frac{1}{N}
    \left\|U_1^T (E-\overline{E})(C-\overline{C})^T  \textcolor{black}{U_2} \right\|_F 
    \leq 
    \frac{\sigma \sqrt{\gamma_{\text{bound}}(\xi)}}{\sqrt{N}}
    \sqrt{1-\frac{1}{N}} \left(\sqrt{d(D-d)} + \frac{6}{5}\varepsilon\right)\right) \nonumber \\
  & \ge (1 - e^{-\varepsilon^2}) (1 - d e^{-\xi^2}),
\end{align*}
where $ \gamma_{\text{bound}}(\xi)$ is the bound on all the
eigenvalues of $\frac{1}{N}U_2^T(C-\overline{C})
(C-\overline{C})^TU_2$ defined in (\ref{eq:eigCC^T-bound}). This leads to a
bound similar to (\ref{LE-All}) for the tangential and curvature
components of the noise,
\begin{equation}
  \begin{bmatrix}
    \\ \|U_2^T\frac{1}{N}(C-\overline{C})(E-\overline{E})^TU_1\|_F \\ \\
    \|U_2^T\frac{1}{N}(C-\overline{C})(E-\overline{E})^TU_2\|_F \\ \\
  \end{bmatrix}
  \leq
  \frac{\sigma \sqrt{\gamma_{\text{bound}}(\xi_c)}}{\sqrt{N}}
  \begin{bmatrix}
    \\ \sqrt{d(D-d)} &+ & \frac{6}{5}\xi_{ce}\\ \\ (D-d) &+ &
    \frac{6}{5} \xi_{ce} \\ \\
  \end{bmatrix}
  \label{CE-All}
\end{equation}
with probability greater than $(1-2e^{-\xi_{c}^2})(1-e^{-\xi_{ce}^2})$
over the joint random selection of the sample points and random
realization of the noise.

\listoffigures
\listoftables


\begin{references}

\bibitem{Brand}
  \textsc{Brand, M.}  (2003) Charting a Manifold. in \emph{Adv. Neural Inf.
    Process. Syst. 15}, pp. 961--968. MIT Press.

\bibitem{Broomhead}
  \textsc{Broomhead, D.  {\small \&} King, G.}  (1986) Extracting Qualitative
  Dynamics From Experimental Data. \emph{Phys. D}, \textbf{20}(2-3), 217--236.

\bibitem{Maggioni-long}
  \textsc{Chen, G., Little, A., Maggioni, M.  {\small \&} Rosasco, L.}  (2011)
  Some Recent Advances in Multiscale Geometric Analysis of Point Clouds. in
  \emph{Wavelets and Multiscale Analysis: Theory and Applications}, ed. by
  J.~Cohen,  {\small \&} A.~Zayed, pp. 199--225. Springer.

\bibitem{Davis}
  \textsc{Davis, C.  {\small \&} Kahan, W.}  (1970) The Rotation of Eigenvectors
  by a Perturbation {III}. \emph{SIAM J. Numer. Anal.}, \textbf{7}, 1--46.

\bibitem{Donoho2}
  \textsc{Donoho, D.  {\small \&} Johnstone, I.}  (1995) Adapting to Unknown
  Smoothness via Wavelet Shrinkage. \emph{J. Amer. Statist. Assoc.},
  \textbf{90}, 1200--1224.

\bibitem{Edelman88}
  \textsc{Edelman, A.}  (1988) Eigenvalues and Condition Numbers of Random
  Matrices. \emph{SIAM J. Matrix Anal. Appl.}, \textbf{9}(4), 543--560.

\bibitem{Jones-new}
  \textsc{Feiszli, M.  {\small \&} Jones, P.}  (2011) Curve Denoising by
  Multiscale Singularity Detection and Geometric Shrinkage. \emph{Appl. Comput.
    Harmon. Anal.}, \textbf{31}, 392--409.


\bibitem{Federer59}
\textcolor{black}{\textsc{Federer, H.}  (1959) Curvature measures. \emph{Transactions of the
  American Mathematical Society}, \textbf{93}(3), 418--491.}

\bibitem{Froehling}
  \textsc{Froehling, H., Crutchfield, J., Farmer, D., Packard, N.  {\small \&}
    Shaw, R.}  (1981) On Determining the Dimension of Chaotic Flows. \emph{Phys.
    D}, \textbf{3}, 605--617.

\bibitem{Fukunaga-Olsen}
  \textsc{Fukunaga, K.  {\small \&} Olsen, D.}  (1971) An Algorithm for Finding
  Intrinsic Dimensionality of Data. \emph{IEEE Trans. Comput.},
  \textbf{c-20}(2), 176--183.

\bibitem{Genovese12}
\textcolor{black}{\textsc{Genovese, C.~R., Perone-Pacifico, M., Verdinelli, I.  {\small \&}
  Wasserman, L.}  (2012) Minimax manifold estimation. \emph{Journal of Machine
  Learning Research}, \textbf{13}, 1263--1291.}

\bibitem{Genovese12b}
\textcolor{black}{\textsc{Genovese, C.~R., Perone-Pacifico, M., Verdinelli, I.  {\small \&}
  Wasserman, L.}  (2012a) Manifold estimation and singular deconvolution under
  Hausdorff loss. \emph{The Annals of Statistics}, \textbf{40}(2), 941--963.}

\bibitem{Giaquinta}
  \textsc{Giaquinta, M.  {\small \&} Modica, G.}  (2009) \emph{Mathematical
    Analysis: An Introduction to Functions of Several Variables}. Springer.

\bibitem{Gray74}
\textcolor{black}{\textsc{Gray, A.}  (1974) The volume of a small geodesic ball of a Riemannian
  manifold.. \emph{The Michigan Mathematical Journal}, \textbf{20}(4),
  329--344.}

\bibitem{GVL}
  \textsc{Golub, G.  {\small \&} Loan, C.~V.}  (1996) \emph{Matrix Computations}.
  JHU Press.

\bibitem{Johnstone}
  \textsc{Johnstone, I.}  (2001) On the Distribution of the Largest Eigenvalue in
  Principal Component Analysis. \emph{Ann. Statist.}, \textbf{29}, 295--327.

\bibitem{Jones}
  \textsc{Jones, P.}  (1990) Rectifiable sets and the Traveling Salesman Problem.
  \emph{Invent. Math.}, \textbf{102}, 1--15.

\bibitem{Jung-Marron}
  \textsc{Jung, S.  {\small \&} Marron, J.}  (2009) \mbox{PCA} Consistency in
  High Dimension, Low Sample Size Context. \emph{Ann. Statist.}, \textbf{27},
  4104--4130.
  
  \bibitem{Kam-Leen}
  \textsc{Kambhatla, N.  {\small \&} Leen, T.}  (1997) Dimension Reduction by
  Local Principal Component Analysis. \emph{Neural Comput.}, \textbf{9},
  1493--1516.
  
  \bibitem{Kaslovsky11a}
\textcolor{black}{\textsc{Kaslovsky, D.  {\small \&} Meyer, F.}  (2011) Image Manifolds:
  Processing Along the Tangent Plane. in \emph{7th International Congress on
  Industrial and Applied Mathematics - \mbox{ICIAM 2011}}.}

  \bibitem{Kaslovsky11b}
\textcolor{black}{\textsc{\BySame{}}  (2011) Optimal Tangent Plane Recovery from Noisy Manifold Samples.
  \url{http://arxiv.org/abs/1111.4601v1}.}

\bibitem{Kaslovsky12a}
\textcolor{black}{\textsc{\BySame{}}  (2012) Overcoming Noise, Avoiding Curvature: Optimal Scale
  Selection for Tangent Plane Recovery. in \emph{Proc. IEEE Workshop on
  Statistical Signal Processing}, pp. 904--907.
  \url{http://dx.doi.org/10.1109/SSP.2012.6319851}.}
  
\bibitem{Kresk-curvature}
  \textsc{Krsek, P., Lukacs, G.  {\small \&} Martin, R.~R.}  (1998) Algorithms
  for Computing Curvatures from Range Data. in \emph{The Mathematics of
    Surfaces VIII, Information Geometers}, pp. 1--16.

\bibitem{LaurantMassart}
  \textsc{Laurant, B.  {\small \&} Massart, P.}  (2000) Adaptive Estimation of a
  Quadratic Functional by Model Selection. \emph{Ann. Statist.},
  \textbf{28}(5), 1302--1338.

\bibitem{Lin-Tong}
  \textsc{Lin, T.  {\small \&} Zha, H.}  (2008) Riemannian Manifold Learning.
  \emph{IEEE Trans. Pattern Anal. Mach. Intell.}, \textbf{30}, 796--809.
  
\bibitem{LittleThesisDuke}
\textcolor{black}{\textsc{Little, A.~V.}  (2011) Estimating the Intrinsic Dimension of
  High-Dimensional Data Sets: A Multiscale, Geometric Approach. Ph.D. thesis,
  Duke University.}

\bibitem{Maggioni-MIT}
\textcolor{black}{\textsc{Little, A.~V., Maggioni, M.  {\small \&} Rosasco, L.}  (2012)
  Multiscale Geometric Methods for Data Sets I: Multiscale SVD, Noise and
  Curvature. Discussion Paper MIT-CSAIL-TR-2012-029, Massachusetts Institute of
  Technology.}

\bibitem{Meyer12b}
\textcolor{black}{\textsc{Meyer, F., Kaslovsky, D.  {\small \&} Wohlberg, B.}  (2012) Analysis of
  Image Patches: A Unified Geometric Perspective. SIAM
  Conference on Imaging Science (IS12).}

\bibitem{Guibas-normal-estimation}
  \textsc{Mitra, N., Nguyen, A.  {\small \&} Guibas, L.}  (2004) Estimating
  Surface Normals in Noisy Point Cloud Data. \emph{Internat. J. Comput. Geom.
    Appl.}, \textbf{14}(4--5), 261--276.

\bibitem{Muirhead}
  \textsc{Muirhead, R.}  (1982) \emph{Aspects of Multivariate Statistical
    Theory}. Wiley.

\bibitem{Nadler}
  \textsc{Nadler, B.}  (2008) Finite Sample Approximation Results for Principal
  Component Analysis: A Matrix Perturbation Approach. \emph{Ann. Statist.},
  \textbf{36}, 2792--2817.

\bibitem{Niyogi11}
\textcolor{black}{\textsc{Niyogi, P., Smale, S.  {\small \&} Weinberger, S.}  (2011) A
  topological view of unsupervised learning from noisy data. \emph{SIAM Journal
  on Computing}, \textbf{40}(3), 646--663.}

\bibitem{Ohtake}
  \textsc{Ohtake, Y., Belyaev, A.  {\small \&} Seidel, H.-P.}  (2006) A Composite
  Approach to Meshing Scattered Data. \emph{Graph. Models}, \textbf{68},
  255--267.

\bibitem{LLE}
  \textsc{Roweis, S.  {\small \&} Saul, L.}  (2000) Nonlinear Dimensionality
  Reduction by Locally Linear Embedding. \emph{Science}, \textbf{290},
  2323--2326.

\bibitem{Osher90}
\textcolor{black}{\textsc{Osher, S.  {\small \&} Rudin, L.~I.}  (1990) Feature-oriented image
  enhancement using shock filters. \emph{SIAM Journal on Numerical Analysis},
  \textbf{27}(4), 919--940.}

\bibitem{Perona90}
\textcolor{black}{\textsc{Perona, P.  {\small \&} Malik, J.}  (1990) Scale-space and edge
  detection using anisotropic diffusion. \emph{Pattern Analysis and Machine
  Intelligence, IEEE Transactions on}, \textbf{12}(7), 629--639.}

\bibitem{Roweis-Global}
  \textsc{Roweis, S., Saul, L.  {\small \&} Hinton, G.}  (2002) Global
  Coordination of Locally Linear Models. in \emph{Adv. Neural Inf. Process.
    Syst. 14}, pp. 889--896. MIT Press.
    
\bibitem{Rudelson-Isotropic}
  \textcolor{black}{\textsc{Rudelson, M.}  (1999) Random Vectors in the Isotropic Position.
  \emph{J. Funct. Anal.}, \textbf{164}(1), 60--72.}    

\bibitem{shawe}
  \textsc{Shawe-Taylor, J.  {\small \&} Cristianini, N.}  (2003) Estimating the
  Moments of a Random Vector with Applications. in \emph{Proc. of GRETSI 2003
    Conference}, pp. 47--52.

\bibitem{Singer}
  \textsc{Singer, A.  {\small \&} Wu, H.-T.}  (2012) Vector Diffusion Maps and
  the Connection Laplacian. \emph{Comm. Pure Appl. Math.}, \textbf{64},
  1067--1144.

\bibitem{Stewart}
  \textsc{Stewart, G.  {\small \&} Sun, J.}  (1990) \emph{Matrix Perturbation
    Theory}. Academic Press.

\bibitem{Tropp}
  \textsc{Tropp, J.}  (2011) User-Friendly Tail Bounds for Sums of Random
  Matrices. \emph{Found. Comput. Math.}, \textbf{12}(4), 389--434.

\bibitem{Tyagi}
  \textsc{Tyagi, H., Vural, E.  {\small \&} Frossard, P.}  (2013) Tangent Space
  Estimation for Smooth Embeddings of \mbox{R}iemannian Manifolds. 
  \emph{Information and Inference}, \textbf{2}(1), 69--114.

\bibitem{Vershynin-HowClose}
\textcolor{black}{\textsc{Vershynin, R.}  (2012) How Close is the Sample Covariance
  Matrix to the Actual Covariance Matrix. \emph{J. Theoret. Probab.}, 
  \textbf{25}(3), 655--686.}

\bibitem{Vershynin}
  \textsc{Vershynin, R.}  (2012) Introduction to the Non-Asymptotic Analysis of
  Random Matrices. in \emph{Compressed Sensing, Theory and Applications}, ed.
  by Y.~Eldar,  {\small \&} G.~Kutyniok, pp. 210--268. Cambridge.

\bibitem{Wang-Marron}
  \textsc{Wang, X.  {\small \&} Marron, J.}  (2008) A Scale-based Approach to
  Finding Effective Dimensionality in Manifold Learning. \emph{Electron. J.
    Stat.}, \textbf{2}, 127--148.

\bibitem{Williams-curvature}
  \textsc{Williams, D.  {\small \&} Shah, M.}  (1992) A Fast Algorithm for Active
  Contours and Curvature Estimation. \emph{Comput. Vis. Image Und.},
  \textbf{55}(1), 14--26.

\bibitem{Yang}
  \textsc{Yang, L.}  (2008) Alignment of Overlapping Locally Scaled Patches for
  Multidimensional Scaling and Dimensionality Reduction. \emph{IEEE Trans.
    Pattern Anal. Mach. Intell.}, \textbf{30}, 438--450.

\bibitem{Lerman}
  \textsc{Zhang, T., Szlam, A., Wang, Y.  {\small \&} Lerman, G.}  (2010)
  Randomized Hybrid Linear Modeling by Local Best-fit Flats. in \emph{CVPR},
  pp. 1927--1934.

\bibitem{Zhang-Zha}
  \textsc{Zhang, Z.  {\small \&} Zha, H.}  (2004) Principal Manifolds and
  Nonlinear Dimensionality Reduction via Tangent Space Alignment. \emph{SIAM J.
    Sci. Comput.}, \textbf{26}, 313--338.

\end{references}
\end{document}